\documentclass[%
 reprint, 
 amsmath,amssymb,
 aps, 
 physrev,
]{revtex4-2}

\usepackage{graphicx}
\usepackage{dcolumn}
\usepackage{placeins}
\usepackage{bm}
\usepackage{tabularx}
\usepackage{enumitem}

\usepackage{hyperref}       
\usepackage{url}            
\usepackage{booktabs}       
\usepackage{amsfonts}       
\usepackage{nicefrac}       
\usepackage{microtype}      
\usepackage{xcolor}         
\usepackage{tikz-cd}
\usepackage{amsmath}
\usepackage{mathrsfs}
\usepackage{comment}
\usepackage{graphicx}
\usepackage{subcaption}
\usepackage{amssymb}
\usepackage{mathtools}
\usepackage{amsthm}

\theoremstyle{plain}
\newtheorem{theorem}{Thm.}[section]
\newtheorem{proposition}[theorem]{Proposition}
\newtheorem{lemma}[theorem]{Lemma}
\newtheorem{corollary}[theorem]{Corollary}
\theoremstyle{definition}
\newtheorem{definition}[theorem]{Definition}

\theoremstyle{remark}
\newtheorem{remark}[theorem]{Remark}
\newtheorem{example}[theorem]{Example}

\newcommand{\norm}[1]{\left|\left|#1\right|\right|}
\newcommand{\stheta}{S^{\theta}}
\newcommand{\rhot}{\tilde{\rho}}

\begin{document}

\preprint{APS/123-QED}

\title{\textbf{The Score Hamiltonian: Mapping Diffusion Models to Adiabatic Transport
}
}

\author{Peter Halmos}
 \affiliation{Computer Science Department, Princeton University.}

\author{Boris Hanin}
 \email{Contact author: bhanin@princeton.edu}
\affiliation{
 ORFE Department, Princeton University.
}

\date{\today}

\begin{abstract}
We exhibit an \textit{exact} correspondence between sampling with score-based diffusion models and adiabatic transport of ground states for a family of Schr\"odinger operators we call Score Hamiltonians, built from the learned score's quantum potential. We obtain novel density reconstruction bounds and principled annealing schedules via adiabatic theorems for Fokker-Planck equations with time-varying potentials. We find the fundamental limit of sampling is set by the ratio of squared score-matching error to Score Hamiltonian spectral gap---the inverse Poincar{\'e} constant of the data density.
\end{abstract}

\maketitle

\section{Introduction}

Modern machine learning has revolutionized our ability to generate samples from complex, high-dimensional distributions. 
For generative tasks related to sampling from continuous distributions such as those involving images \cite{ho2020denoising, rombach2021highresolution},  small molecules \cite{pmlr-v162-hoogeboom22a, pmlr-v202-xu23n} and proteins \cite{Watson2023}, a particularly important approach has been to use  score-based diffusion models \cite{hyvarinen05a, denoising_sm, pmlr-v37-sohl-dickstein15, song2021scorebased, ho2020denoising}. Based on principles from non-equilibrium thermodynamics \cite{pmlr-v37-sohl-dickstein15,song2021scorebased}, such models couple samples from a target distribution $\rho_{\mathrm{data}}$ to a heat bath, connecting $\rho_T=\rho_{\mathrm{data}}$ to a simple reference distribution $\rho_{0}$ (e.g. a Gaussian) through a family of intermediate distributions $\rho_{t}$. This thermalization is used to learn the score $S_{t} = \nabla \log \rho_{t}$, which gives a principled way of drawing samples from the reference and deforming them into samples of $\rho_{\mathrm{data}}$ through a Langevin diffusion without ever needing to compute the partition function.

While score-based diffusion models are widely used in practice, a conceptual and theoretical understanding of why they work is far from complete. Most prior theory work uses tools from stochastic analysis to bound how score-estimation error along the trajectory degrades the quality of Langevin sampling \cite{chen2023sampling, lee2022convergence, benton2024nearly, yang2022convergence}. In this article, we provide a fundamentally new perspective: an exact mapping from inference with a trained score-based diffusion model to the adiabatic transport of ground states \cite{Born1928,Jansen2007} for a canonically-defined family of Schr{\"o}dinger operators built from the learned score via its Bohm quantum potential \cite{Madelung1927,Bohm1952}. We call these operators \emph{Score Hamiltonians} (see \eqref{eq:score-hamiltonian}); they have several remarkable properties:

\begin{enumerate}[itemsep=0pt, parsep=0pt, topsep=0pt]
    \item The Score Hamiltonian of a conservative score $S=\nabla \log \rho$ has ground state $\sqrt{\rho}$, and for non-conservative $S$ has well-defined induced ground-state $\sqrt{\tilde{\rho}}$. This frames generative density reconstruction as an eigenvalue problem.
    \item Sampling with a score-based diffusion model is adiabatic transport of ground states of the associated Score Hamiltonians. See \eqref{eq:ground-state-transform-dynamics} and Thm. \ref{thrm:adiabatic_FK} for the resulting quantitative adiabatic theorem and Thm. \ref{thrm:ad_diffusion} (also Figs \ref{fig:Adiabatic} and \ref{fig:AdiabaticAlloc}) for the application to diffusion models with imperfectly learned scores.
    \item The spectral gap $\Delta$ of the Score Hamiltonian--the inverse Poincar{\'e} constant \cite{Helffer2005} of $\rho_{t}$--controls the difficulty of sampling: we prove non-asymptotic upper and matching lower bounds (Thms.~\ref{thrm:ad_diffusion},Prop.~\ref{prop:LB_tight}) giving a terminal error floor $\epsilon_{\rm score}/\sqrt{\Delta}$.
\end{enumerate}

\section{Main Results}

At inference time score-based diffusion models use their learned parametric approximation $S_t^\theta $ of the true score to evolve an
initial sample $x_0 \sim \rho_0$ into an approximate sample from $\rho_{\mathrm{data}}$. Assuming for the moment that $S_t^\theta= \nabla \log \rho^\theta_t$ is conservative, this evolution can be written as a stochastic differential equation
\begin{equation}\label{eq:learned-langevin}
    d x_\tau
    =
    \nabla \log \rho^\theta_{t_\tau}(x_\tau)\,d\tau
    +
    \sqrt{2}\,dW_\tau.
\end{equation}
Here $\tau$ is the \textit{algorithmic time} and $t_\tau=t(\tau)$ is a monotone annealing schedule, governing how quickly the sampler moves from the reference law at $t_0=0$ toward the data law. We will later drop the assumption that $S_t^\theta$ is conservative and generalize in  Appendix~\ref{sec:diffusion_VP} to other common diffusion-model parameterizations.

The law
$\mu_\tau^\theta$ of $x_\tau$ follows the time-inhomogeneous Fokker-Planck equation \cite{Risken1996}
\begin{equation}
    \partial_\tau \mu_\tau^\theta
    =
    \nabla^2 \mu_\tau^\theta 
    -
    \nabla \cdot
    \left(
        \mu_\tau^\theta \nabla \log \rho^\theta_{t_\tau}
    \right).
    \label{eq:learned-fokker-planck}
\end{equation}
For fixed $t$ the distribution $\rho_t^\theta$ is the unique stationary distribution for \eqref{eq:learned-fokker-planck}. Sampling with score-based diffusion models is thus a relaxation towards a moving equilibrium. This raises two fundamental questions:
\begin{enumerate}[itemsep=0pt, parsep=0pt, topsep=0pt]
    \item How closely does $\mu_{\tau}$ \eqref{eq:learned-fokker-planck} track the model's own instantaneous equilibrium \(\rho_{t_{\tau}}^\theta \)?
    \item If \(\nabla\log \rho^{\theta}_t\) is an \emph{estimator} of the true score $\nabla \log \rho_t$, how close are the laws $\mu_\tau^\theta$ and $\rho_{t_{\tau}}$ as a function of the $L^2$ error at each $t$ of estimating the true score.
\end{enumerate}

\begin{figure}[tbp]
    \centering
    \includegraphics[width=\linewidth]{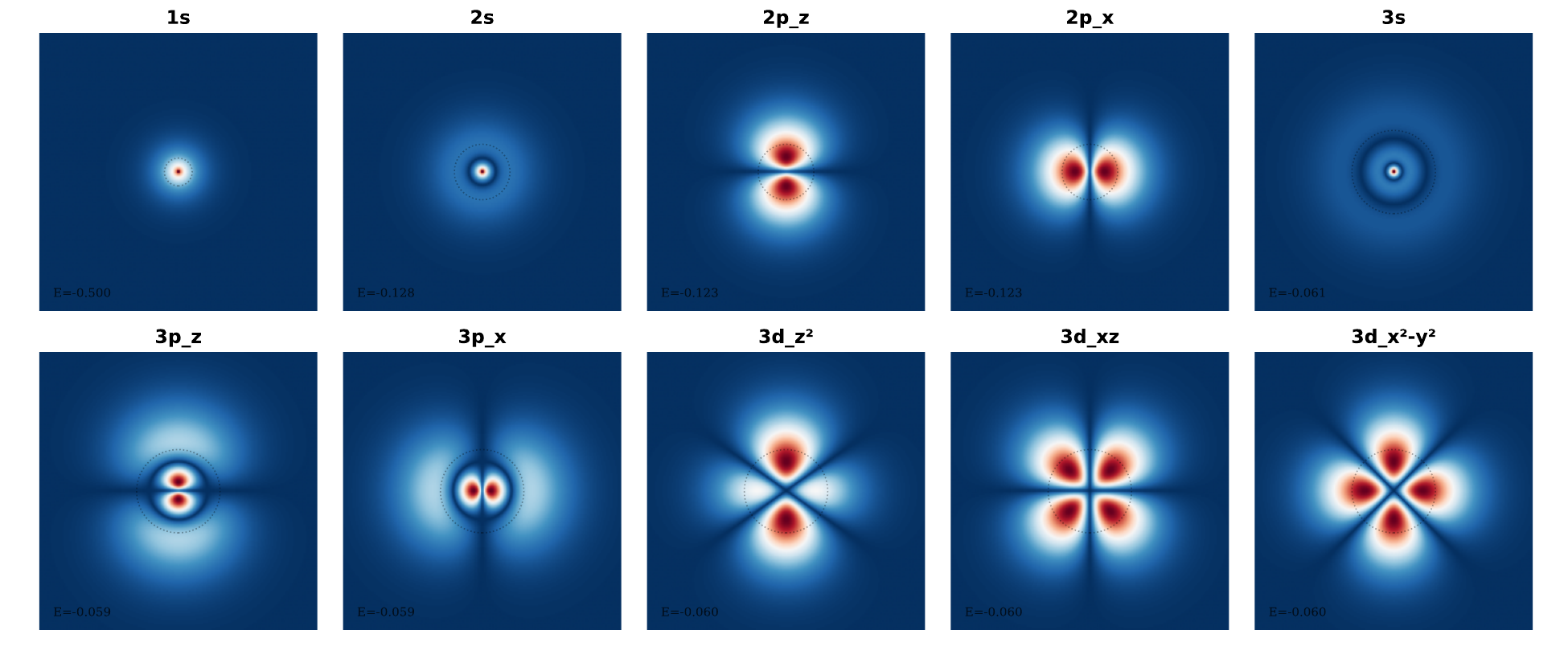}
    \caption{Orbitals $|\psi_{nl}(x)|^2$ of the Hydrogen Atom inferred from the spectrum of a Score Hamiltonian $\widehat{H}_{\theta}$ of a diffusion-model trained on $1\mathrm{s}$ ground-state samples.}\label{fig:HydrogenOrbitals}
\end{figure}


To address these questions we map an amplitude $R \in L^{2}$ to a
Schr\"odinger operator 
\begin{equation}
\widehat H
=-\nabla^2 + Q,\qquad Q:=\nabla^{2}R/ R,
\end{equation}
where $Q$ is the negative quantum potential \cite{Bohm1952, Madelung1927} of $R$. In the case $R=\sqrt{\rho}$ for nodeless $\rho>0$, $\widehat{H}$ is the $0$-form Witten Laplacian \cite{Witten1982, Helffer2005}. Using $Q$ as the potential inverts the usual construction: the Witten Laplacian is built from a potential $\phi$ with density derived $\rho \propto e^{-2\phi},$ whereas $Q$ treats amplitude $R$ (or score $S$) as the primitive. When $\rho>0$, a standard identity \cite{Nelson1966, NELSON2020, Fiscaletti2017, SBITNEV2009} expresses $Q$ through the score $S$ alone: $\nabla^{2}\sqrt{\rho}/\sqrt{\rho} =  \tfrac{1}{2}\nabla \cdot S + \tfrac{1}{4}\|  S\|^{2} $. Any learned score $S_t^\theta$ thus yields a canonical family of Schr\"odinger operators, the \emph{Score Hamiltonians}
\begin{equation}\label{eq:score-hamiltonian}
\widehat{H}_t^\theta = -\nabla^{2}+\frac{1}{2} \nabla \cdot \stheta + \frac{1}{4} \|  \stheta \|^{2}.
\end{equation}
These are well-defined and self-adjoint for any smooth, possibly non-conservative $S^\theta,$ reflecting how diffusion models are parametrized in practice \cite{scoreconserv}. For conservative scores 
the ground state identity $
\widehat H \sqrt{\rho} = 0    
$ holds, so the instantaneous model density $\rho_t^\theta$ is recovered by heat flow on $\widehat{H}_t^\theta$. 


Finally, and most importantly for the purposes of the present article, the time $t_\tau$ ground-state transform maps the time-inhomogeneous Fokker-Planck density $\mu_\tau^\theta$ to
\[
    \Psi_\tau^\theta = \mu_\tau^\theta / \sqrt{\smash[b]{\rho_{t_\tau}^{\theta}}}
\]
and the Langevin dynamics \eqref{eq:learned-fokker-planck}, which we first address in the conservative case $\stheta=\nabla\log\rho^{\theta}$, 
becomes
\begin{align}
    \partial_\tau \Psi_\tau^\theta
    &=
    -\widehat H_{t_\tau}^\theta \Psi_\tau^\theta
    -
    \frac{\dot t_\tau}{2}
    \partial_t \log \rho_{t_\tau}^\theta\,\Psi_\tau^\theta .
    \label{eq:ground-state-transform-dynamics}
\end{align}
This formula is remarkably simple. The first term is a heat flow on $\widehat H_{t_\tau}^\theta$, which drives $\Psi_\tau^\theta$ to the ground state. The second term is a gauge correction whose strength is modulated by the annealing speed $\dot t(\tau)$. Together, they turn sampling with a score-based diffusion model into an adiabatic ground-state tracking problem for $\widehat{H}_{t_\tau}^\theta$.

\begin{figure}[tbp]
    \centering
    \includegraphics[width=0.6\linewidth]{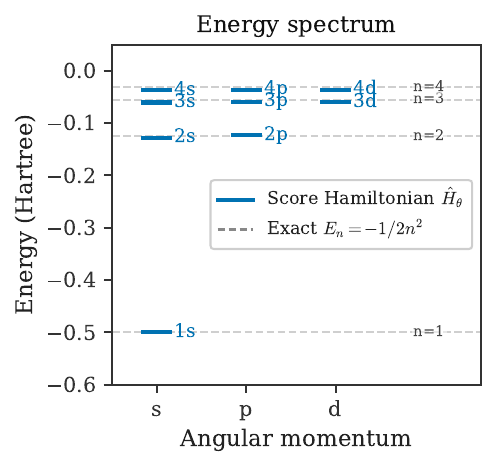}
    \caption{Spectrum of the learned Score Hamiltonian $\widehat{H}_{\theta}$ against the exact energies $E_{n}=-1/2n^{2}$ of Hydrogen.}\label{fig:HydrogenSpectrum}
\end{figure}

\subsection{Adiabatic Theorem for Fokker-Planck with Time-Varying Score} \label{sec:oracle-results} 
In this section we answer question 1 above. That is, we consider the time-inhomogeneous Fokker-Planck equation~\eqref{eq:learned-fokker-planck} driven by a family of conservative scores $S_{t}^\theta = \nabla \log \rho_{t}^\theta$. Since these scores are arbitrary, we omit in this section the dependence on $\theta$ and defer to \S \ref{sec:results-diffusion}  quantifying the effect of a mismatch between a learned score $S_t^\theta$ and a true score $S_t$. 

Our first theorem is an adiabatic bound for the ground-state tracking error of \eqref{eq:learned-fokker-planck}, extending quantitative adiabatic theorems~\cite{Born1928,Jansen2007,Morita2008} to time-inhomogeneous Fokker-Planck equations. Unlike prior adiabatic theorems for general stochastic generators \cite{Avron2012}, our result is non-asymptotic with lower bound (Prop.~\ref{prop:LB_tight}) that matches locally in the rate $|\dot{t}|$. Let $\chi^2(\mu\|\nu)
    =
    \int \left(\mu/\nu-1\right)^2\,\nu dx$ be the $\chi^2$ divergence between densities $\mu,\nu$, and let
\[
    M_{t,-}(x)
    =
    \max\left\{0,-\tfrac{1}{2} \partial_t \log \rho_t(x)\right\}.
\]
\begin{theorem}[Adiabatic theorem for Fokker-Planck]\label{thrm:adiabatic_FK}
Suppose $S_{t_\tau}=\nabla \log \rho_{t_\tau}$ is such that,
for all $t$, $\widehat H_t$ is self-adjoint, the spectral gap $\Delta_t$ of
$\widehat H_t$ is strictly positive,
$\partial_t \log \rho_t \in L^2(\rho_t)$, and there exist $a(t),b(t)\geq 0$ so that
\begin{equation}
    \langle u,M_{t,-}u\rangle_{L^2}
    \leq
    a(t)\langle u,\widehat H_t u\rangle_{L^2}
    +
    b(t)\langle u,u\rangle_{L^2}
    \label{eq:form-bound}
\end{equation}
for all $u$ orthogonal to the ground state $\sqrt{\rho_t}$. Then
\begin{align}\label{eq:adiabatic-bound}
\notag \sqrt{
        \chi^2\!\left(
            \mu_\tau\,\middle\|\,\rho_{t_{\tau}}
        \right)
}&\leq    \sqrt{\chi^2(\mu_0\|\rho_0)}
\exp\!\left(
        -\int_0^\tau \Gamma(s)\,ds
\right)\\
&+
\sup_{s\le \tau}
\frac{
        |\dot t(s)|\,
        \|\partial_t\log\rho_{t(s)}\|_{L^2(\rho_{t(s)})}
}{
        \Gamma(s)
},
    \end{align}
where $\Gamma(\tau)
    :=
    \left(1-|\dot t(\tau)|a(t_\tau)\right)\Delta_{t_\tau}
    -
    |\dot t(\tau)|b(t_\tau)$ is the effective spectral gap, and we take
$\dot t=dt/d\tau$ sufficiently small that $\Gamma(\tau)$ is strictly positive
for all $\tau$.
\end{theorem}
Theorem \ref{thrm:adiabatic_FK} is proved in \S \ref{appendix:main-proofs}. In the bound \eqref{eq:adiabatic-bound}, one can also replace the $\sqrt{\chi^2}$ distance by total variation  distance (see \eqref{eq:TV-Hellinger}). Before applying this theorem to score-based diffusion models in \S \ref{sec:results-diffusion}, we comment
briefly on Thm. \ref{thrm:adiabatic_FK} assumptions and constituents of the upper bound \eqref{eq:adiabatic-bound}. Regarding assumptions, self-adjointness of $\widehat{H}_t$ follows from a mild quadratic lower-envelope condition on the Bohm potentials $Q(t)$ at infinity (see~\ref{prop:score_hamiltonian_props}). Next, the spectral gap
$\Delta_\tau$ of the Score Hamiltonian $\widehat H(t_\tau)$ is the
reciprocal of the Poincar\'e constant of $\rho_{t_\tau}$ \cite{Helffer2005}, which measures the
intrinsic difficulty of adiabatically transporting $\rho_t$. Further, differentiating the  relation $\widehat{H} \sqrt{\rho_t} = 0$ and applying the spectral gap shows that for each $t$ $\norm{\partial_t \log \rho_t}_{L^2(\rho_t)}^2\leq 2\mathrm{Var}_{\rho_t}[\partial_t\, Q(t)]/\Delta_t$
and hence the condition that $\partial_t \log \rho_t \in L^2(\rho_t)$ is a generic finite energy requirement on the potential perturbation in the family $\widehat{H}_t$. Finally, the most restrictive condition is \eqref{eq:form-bound}.  As explained in Appendix \S~\ref{cor:lsi_form}, a
log-Sobolev inequality \cite{Bakry2014} for $\rho_t$ gives a clean sufficient condition for
the required form bound \eqref{eq:form-bound}. More generally, Kato-type integrability $M_{t,-} \in L^p(\rho_t)$ for $p>d/2$ ensures condition \eqref{eq:form-bound} \cite{Reed1975}.

Let us now briefly interpret the terms in the upper bound. The first is the memory of the initial condition, which decays exponentially in the average effective spectral gap $\Gamma(s)$ along the adiabatic trajectory. The second term is the finite-rate lag, which can be
made arbitrarily small by choosing a sufficiently slow schedule. In Appendix \S \ref{prop:LB_tight}, we derive lower bounds for $\sqrt{\chi^2(\mu_\tau\|\rho_{t_\tau})}$, which are sharp: there exist hard regimes for which the upper bound matches the lower bound on the adiabatic lag. Section~\S\ref{sec:diffusion_VP} collects the corresponding variance and confinement assumptions under which the form bound \eqref{eq:form-bound} and positive spectral gap hold.

When $S$ is \emph{non-conservative}, $\widehat{H}(t)$ remains self-adjoint with ground state $\sqrt{\tilde{\rho}_{t}}$ and induced conservative score $\nabla \log \tilde{\rho}_{t}$. This $\tilde{\rho}_{t}$ coincides with a non-equilibrium steady-state (NESS) for isoclinic solenoidal fields (\S~\ref{sec:noncons_exten}). Let $A_t =  S_{t}-\nabla \log \tilde{\rho}_{t}$ be this rotational gap to the induced score. Then, $\widehat{H}(t)$ has a (non-zero) ground state energy $\lambda(t) = \tfrac{1}{4}\mathbb{E}_{\tilde{\rho}_{t}} \lVert A_t \rVert^{2}$. We prove in Proposition~\ref{prop:adiabatic-bound_NC} that Theorem~\ref{thrm:adiabatic_FK} holds on $\psi_{0}=\sqrt{\tilde{\rho}_{t}}$ with a term in the supremum \emph{independent of} $|\dot{t}|$ quantifying the rotational kinetic energy, ${(\mathrm{Var}_{\tilde{\rho}_{t_s}}|A_t |^2)^{1/2}}/{2\Gamma_{\mathrm{rot}}(s)}$, with modified spectral gap $\Gamma_\mathrm{rot}(\tau)$ relative to the conservative gap $\Delta_t$. We also show that the difference to the conservative gap is bounded by difference in the rotational kinetic energy over the ground-state and first excited-state (Section~\ref{sec:noncons_exten}).


\subsection{Spectral Adiabatic Bounds for Diffusion Models}\label{sec:results-diffusion} Our primary application of Thm.~\ref{thrm:adiabatic_FK} is to analyze the quality of a trained score-based diffusion model in terms of the two squared score errors
\[
    \epsilon_{\rm score}^2(t)
    =
    \tfrac{1}{2}\|S^\theta_t-S_t\|^2_{L^2(\rho_t)},
    \,
    \epsilon_{\theta,{\rm score}}^2(t)
    =
    \tfrac{1}{2}\|S^\theta_t-S_t\|^2_{L^2(\rho^\theta_t)},
\]
target interpolation density $\rho_t$, Score Hamiltonian $\widehat{H}_t$,
spectral gap $\Delta_t$, and those of the model $\rho^{\theta}_t$, $\widehat{H}_t^{\theta}$, $\Delta_{\theta,t}$. Suppose
$\mu_\tau^\theta$ evolves according to \eqref{eq:learned-fokker-planck} with
score field $S^\theta_t$.


\begin{theorem}[Adiabatic bounds for diffusion models] \label{thrm:ad_diffusion}
Given the assumptions of the adiabatic theorem hold on $S^\theta_{t_\tau}$,
the error between the Fokker--Planck marginal $\mu_{\tau_T}^\theta$ of a diffusion
model and the target $\rho_T$ is bounded above as
\begin{align}
    d_{\rm TV}(\mu_{\tau_T}^\theta,\rho_T)
    \leq\;&
    \frac{1}{2}
    \sqrt{\chi^2(\mu_0\|\rho_0)}
    \exp\left(
        -\int_0^\tau \Gamma(s)\,ds
    \right)                                                   \nonumber \\
    &+
    \sup_{s\leq \tau}
    \frac{
        |\dot t(s)|\,
        \|\partial_t \log \rho_{t(s)}^\theta\|_{L^2(\rho_{t(s)}^\theta)}
    }{
        2\Gamma(s)
    }                                                        \nonumber \\
    &+
    \min\left\{
        \frac{\epsilon_{\theta,{\rm score}}(T)}{\sqrt{\Delta_T}},
        \frac{\epsilon_{\rm score}(T)}{\sqrt{\Delta_{\theta,T}}}
    \right\}.
    \label{eq:diffusion-model-bound}
\end{align}
\end{theorem}

Thm. \ref{thrm:ad_diffusion} gives a transparent upper bound on the error in generative modeling with a trained score-based diffusion model. The first two errors are shared by the adiabatic theorem: a decaying error from misalignment to the initial reference density and an adiabatic error accumulated over the trajectory. Ensuring a sufficiently slow annealing speed $|\dot t|$ and alignment with the simple reference $\mu_0$ makes these terms arbitrarily small. What remains is the ratio of the score-estimation error to the root spectral gap $\sqrt{\Delta_{T}}$ of the terminal time data Hamiltonian itself. It is only here that the matching error $\epsilon_{\rm score}(T)$ appears. 
This is in contrast with standard stochastic-path bounds, which use Girsanov to control the score error over the entire trajectory \cite{chen2023sampling, lee2022convergence, benton2024nearly, yang2022convergence}.

Our results apply to annealed score-based Langevin \cite{annealed_langevin}, variance exploding, and variance-preserving parameterizations \cite{song2021scorebased} of diffusion models by choice of target family $\rho^\theta_t$ and schedule $\dot{t}$ (Appendix Remark~\ref{rem:parametrization}). Proposition~\ref{prop:non_conservative_score} and ~\ref{prop:adiabatic-bound_NC} detail the non-conservative score corrections. If $\stheta$ is non-conservative, one tracks $\tilde{\rho}^{\theta}_t$ with gap $\Gamma_{\mathrm{rot}}(s)$. In addition to the bound of~\eqref{eq:diffusion-model-bound} one incurs a terminal penalty of $\sqrt{\frac{2 \lambda_{\theta}(T)}{\Delta_T}}$ for $\lambda_{\theta}(T) = \tfrac{1}{4}\mathbb{E}_{\tilde{\rho}_{T}}\lVert A^{\theta}_T \rVert^{2}$ and the dynamic penalty in the supremum of Prop.~\ref{prop:adiabatic-bound_NC}.

To ensure the adiabatic lag is below a tolerance $\eta_{\rm ad}$, Thm. \ref{thrm:adiabatic_FK} suggests a diffusion schedule
\[
    |\dot t|
    \lesssim
    \eta_{\rm ad}
        \Gamma(s)/
        \|\partial_t \log \rho_t\|_{L^2(\rho_t)},
\]
analogous to the adiabatic brachistochrone \cite{QAB}. A more conservative and tractable alternative is to take
\begin{equation}
    |\dot t|
    \lesssim
    \eta_{\rm ad}\Delta_t^{3/2}/ \|\partial_t S_t\|_{L^2(\rho_t)}
    ,
    \label{eq:score-based-adiabatic-schedule}
\end{equation}
with annealing speed scaling polynomially in the reciprocal of the spectral gap. See Fig. \ref{fig:Adiabatic}.

\subsection{Numerical Experiments} 

\noindent \textit{Recovering Excited States of the Hydrogen Atom.} As a first experiment, we demonstrate Hamiltonian-Identification (Thm.~\ref{thrm:Ham_Identification}) by solving a quantum inverse-problem: we perform score-matching on samples of the 1s Hydrogen atom orbital ground-state to recover the underlying $\hat{H}$. We compare the Score Hamiltonian $\hat{H}^{\theta}$ to Hamiltonians computed from two classical potentials used in generative modeling. The first is thermodynamic integration of the score $-\int (\stheta(x) - \stheta(\infty)) \cdot dx$; since we train a conservative score, we equivalently take $-\log \rho_{\theta}$ for the same network defining $\stheta=\nabla\log \rho_{\theta}$. The second, $-\log\rho_{\eta}(x)$, is the likelihood-based flow \cite{pmlr-v37-rezende15, No2019}.

\begin{figure}[tbp]
    \centering\includegraphics[width=0.95\linewidth]{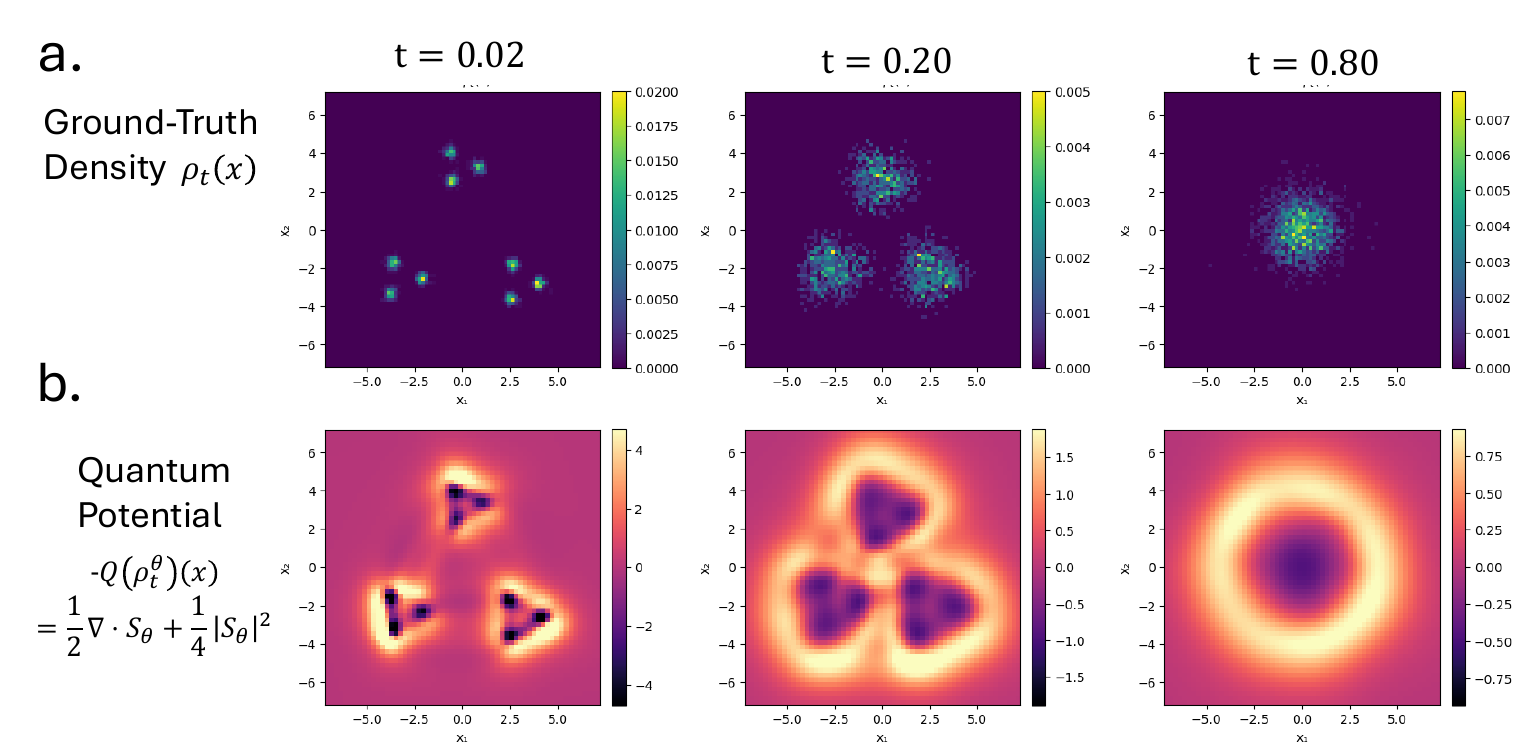}
    \caption[Generation on a Hierarchical Density.]{\textbf{(a)} Interpolation target $\rho_{t}$ for a variance-preserving \cite{song2021scorebased} diffusion model. \textbf{(b)} Quantum potential of the trained diffusion model $\stheta$.
    }
    \label{fig:Hierarch_Ex}
\end{figure}

\begin{figure}[tbp]
    \centering\includegraphics[width=\linewidth]{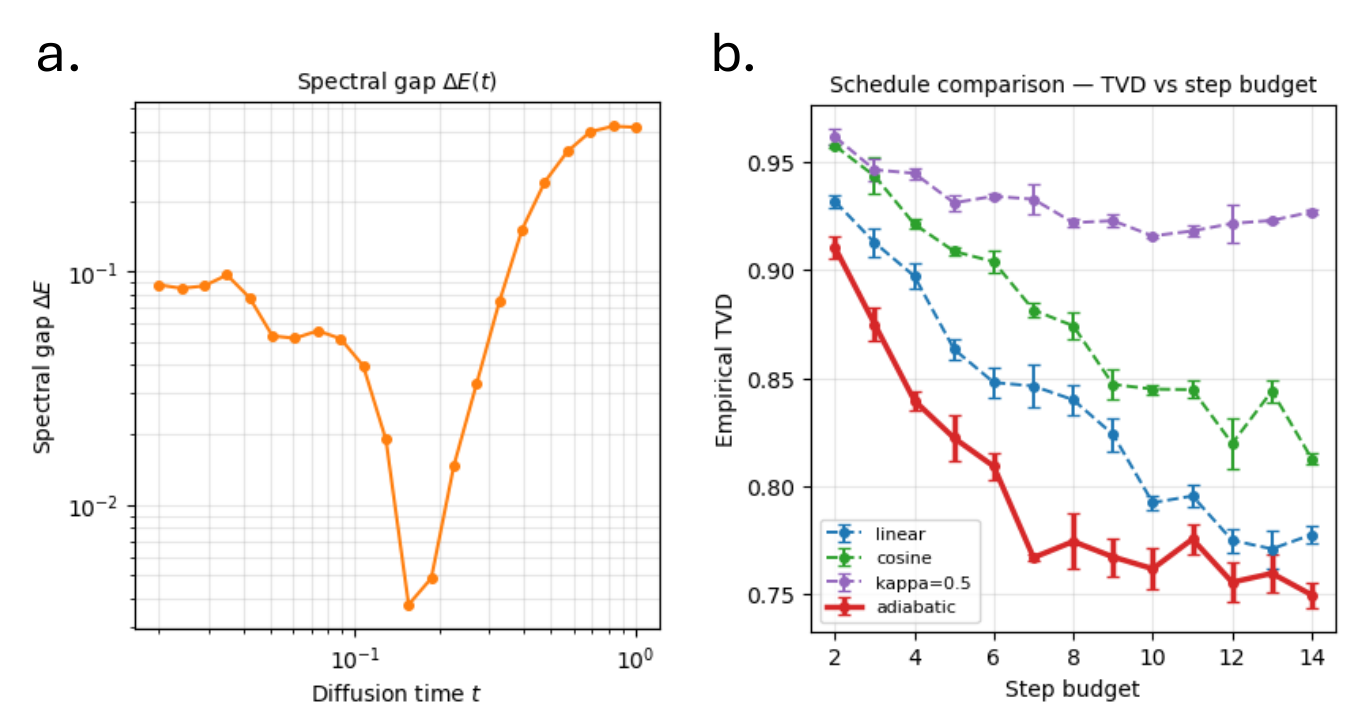}
    \caption{ \textbf{(a)} Spectral gap $\Delta({t_s})$ across diffusion time and \textbf{(b)} diffusion model error $d_{\mathrm{TV}}$ across schedules $|\dot{t}|$.
    }
    \label{fig:Adiabatic}
\end{figure}

From the exact ground state $\rho  \propto e^{-2r}$, these classical treatments predict $-\log \rho \propto 2 r + \mathrm{const}$, a linear confining well. Meanwhile (in units $-\frac{1}{2}\nabla^{2}$) one has $Q$ scaling as $1/2-1/r$, coinciding with the Coulomb potential. The empirical results in Fig.~\ref{fig:HydrogenPotential} match these predictions. Moreover, 
the Score Hamiltonian predicts the exact excited energies $E_n = -1/(2n^2)$ with absolute error $|\hat{E}_{n} - E_{n}|$ across the $2\mathrm{s}, 2\mathrm{p}, 3\mathrm{s}, 3\mathrm{p},$ and $3\mathrm{d}$ in the range $0.0020$-$0.0051$ (see Table~\ref{tab:hydrogen_spectrum}, Fig.~\ref{fig:HydrogenSpectrum}). Thermodynamic treatment of the exact same score and Boltzmann generator yield unbound states. The recovered orbital densities $|\psi_{nl}(x)|^2$ are shown in Fig.~\ref{fig:HydrogenOrbitals}. For analogous results for the coupled Harmonic oscillator see \eqref{eq:coupled_HO}, \S\ref{sec:CHO}, and Table~\ref{tab:cho_metrics_supp} and for experimental details see \S \ref{appendix:numerical-setup}.
\,\\
\,\\
\noindent \textit{Double well.} In Section~\ref{sec:demonstration} we consider the analytically tractable setting of transporting a harmonic oscillator to a double-well (Gaussian-mixture), demonstrating the non-asymptotic upper-bound (Thm.~\ref{thrm:ad_diffusion}) holds on an empirically trained neural diffusion model. We find that in the limit $|\dot{t}| \downarrow0$ that $d_{\mathrm{TV}}$ scales as $\epsilon_{\theta}/\sqrt{\Delta}$ ($R=0.953$, see Fig.~\ref{fig:gmm_bounds_validation}a, b). The score error $\epsilon_{\theta}$ and gap $\Delta$ are collinear with $r(\epsilon_{\theta},\Delta)=-0.91$, so we use partial correlations. Controlling for the gap reduces the $r(\mathrm{TVD},\epsilon_{\theta})$ correlation from $0.87$ to $0.35$ while controlling for $\epsilon_{\theta}$ reduces the $r(\mathrm{TVD},1/\sqrt{\Delta})$ correlation from $0.95$ to $0.78$. These results confirm the necessity of the spectral gap and the $\epsilon_{\theta}/\sqrt{\Delta}$ scaling of Thm.~\ref{thrm:ad_diffusion}. We also add curl to a conservative score, increasing the ground-state energy of the Score Hamiltonian. We compute the ground-state $E_{0}=\lambda_{\theta}$ of $\widehat{H}_{\theta}$, and find that added non-conservativity recovers $\sqrt{2\lambda_{\theta}/\Delta} \propto \mathrm{TVD}$ of Proposition~\ref{prop:non_conservative_score} with $R=0.993$ (Fig.~\ref{fig:gmm_bounds_validation}c).
\,\\
\,\\
\noindent \textit{Adiabatic Brachistochrone.} In Fig.~\ref{fig:Adiabatic} we demonstrate the adiabatic brachistochrone schedule of Equation~\eqref{eq:score-based-adiabatic-schedule} on a hierarchical density exhibiting a transient decay in the spectral gap during mode formation. We find the spectrally-adaptive adiabatic schedule~\eqref{eq:score-based-adiabatic-schedule} outperforms standard diffusion model linear, cosine, or power schedules for $t_{\tau}$ with the same step-budget (Fig.~\ref{fig:Adiabatic}). In Supplemental Fig.~\ref{fig:AdiabaticAlloc} we show the adaptive schedule takes large steps when $\Delta(s) \uparrow$ and decelerates at the critical time of mode formation and branching where the spectral gap $\Delta(s) \downarrow 0$ closes. Mode collapse corresponds to failed tracking of the ground state $\psi_{0}$, and adoption of excited state components $\{\psi_{k \geq 1}\}$.

\subsection{Discussion.} We constructed the Score Hamiltonian, an operator connecting score-based generative modeling to imaginary-time adiabatic transport and allowing for a quantitative study of diffusion models via spectral theory and quantum mechanics.  As diffusion models are trained by minimizing the Fisher-divergence $\min_{\theta}\,\lVert \nabla \log \rho - \nabla \log \rho^{\theta} \rVert_{L^{2}(\rho)}^{2}$ \cite{denoising_sm, hyvarinen05a}, their natural metric on densities is the Fisher-Rao metric \cite{Amari2016} whose associated functional generating its flow is the Fisher Information $I[\rho]= \frac{1}{4}\int |\nabla \log \rho |^{2} \rho$. The quantum potential $Q$ \cite{Madelung1927, Bohm1952} is precisely its first variation \cite{Gianazza2008, Tsekov2008} $Q \propto {\delta I}/{\delta \rho}$. Thus, the Score Hamiltonian is not chosen, but \emph{canonical}: its potential is the variational derivative of the functional underlying score-matching. Further, Madelung flow \cite{Madelung1927} itself conserves the sum of the kinetic energy and Fisher-divergence to the ground-state (Prop.~\ref{prop:Schrodinger_relFisher}) implying the canonicality of the Score Hamiltonian extends to real-time. The Score Hamiltonian can spectrally decompose any data distribution learned by a score-based diffusion model, opening the possibility of using spectral tools for analyzing mode-formation, metastability, and the geometry of generative processes.

\begin{acknowledgments}
BH gratefully acknowledges support by a 2024 Sloan Fellowship in Mathematics, NSF CAREER grant DMS-2143754, and NSF grant DMS-2133806, and DARPA AIQ grant (HR001124S0029). PH is supported by NIH/NCI grant U24CA248453 (PI: Benjamin J. Raphael).
\end{acknowledgments}

\bibliography{apssamp}
\,\\ 
\,\\ 

\newpage
\,\\ 
\newpage

\appendix

\section*{Overview of Supplementary Material}

The supplementary material is organized as follows.  Appendix \ref{appendix:technical} fixes
notation and provides the dictionary between diffusion models and the spectral formulation used in the main text.  It recalls the
ground-state transform, derives the moving-frame gauge term in the time-inhomogeneous Fokker--Planck equation, explains how the variance-exploding and variance-preserving diffusion parameterizations fit into the canonical Langevin form, and records the Fisher and Bohm identities underlying the Score Hamiltonian.

Appendix \ref{appendix:main-proofs} contains the proofs of the main results. It introduces the Information-Hamiltonian formulation, proves the basic spectral properties of the Score Hamiltonian, establishes the adiabatic theorem and its matching lower bound, and derives the perturbation estimates that convert score-matching error and non conservativity into density-reconstruction error.

Appendix \ref{appendix:numerical-setup} contains the numerical details for the physical Hamiltonian-identification experiments: the hydrogen atom and the coupled harmonic oscillator. Appendix~\ref{sec:demonstration} contains the numerical details for the synthetic diffusion-model experiments, including the Gaussian-mixture benchmarks, spectral-gap computations, total-variation measurements, and adiabatic scheduler comparisons.

\section{Technical Preliminaries}\label{appendix:technical}
To bridge the standard objects of Langevin dynamics and the Fokker-Planck equation \cite{Risken1996} with our spectral study of an operator on Lebesgue space $L^{2}$, we recall the ground-state transform.
\begin{definition}[Ground-State Transform]
Let $\rho(x) \propto e^{-2\varphi(x)}$ be a Gibbs measure. Define the ground-state transform $\hat{U}_\rho : L^2(\rho) \to L^2(dx)$ by multiplication with the density amplitude:
$$\hat{U}_\rho f = f\sqrt{\rho}.$$
\end{definition}
Suppose now that $\mu_s$ is a curve of probability measures satisfying the Langevin evolution
\[
\partial_{s} \mu_{s} = L_{t}^{*} \mu_{s} = \nabla^{2} \mu - \nabla \cdot (S\mu),\quad S = \nabla \log \rho
\]
Conjugating minus the Langevin generator $L_t$ by $\widehat{U}_\rho$ by the ground state transform yields the unitarily equivalent operator \cite{setti}
\[
\widehat{H}_t = \widehat{U}_\rho (-L_t) \widehat{U}_{\rho}^{-1}
\]
a self-adjoint Schrödinger operator, called the Witten Laplacian \cite{Witten1982, Helffer2005}. Explicitly:
\[
\widehat{H}_t = -\nabla^2 + ||\nabla \varphi_{t}||^2 - \nabla^2 \varphi_{t},
\]
and a score-based inversion of $\widehat{H}_t$ is our central analytic focus. The density-ratio amplitude between the Fokker-Planck marginal $\mu_{s}$ and the equilibrium $\sqrt{\rho}$,
\[
\Psi_s = \mu_{s} / \sqrt{\rho}
\]
evolves under pure imaginary-time Schrödinger evolution with respect to the Witten Laplacian, i.e. 
\[
\partial_s \Psi_s = -\widehat{H} \Psi_s.
\]
Thus, for $e^{-s \widehat{H}}$ the contraction semi-group, one has the correspondence between the Langevin marginal $\mu_{s}$ at relaxation time $s$, the initialization marginal $\mu_{0}$ of the Langevin dynamics, the equilibrium (ground-state) $\sqrt{\rho}$ given by, 
\begin{align}\label{eq:LangevinSchrodinger_correspondence}
&\mu_{s} = \sqrt{\rho} \, e^{-s \widehat{H}} \left( \frac{\mu_{0}}{\sqrt{\rho}} \right)
\end{align}
Our aim is to analyze the \emph{time-inhomogeneous} Fokker-Planck evolution
\begin{equation}\label{eq:adiabatic-with-gague}
\partial_\tau \mu_\tau = L^*_{t_{\tau}}\mu_\tau= \nabla^{2} \mu_{\tau} + 2\,\nabla \cdot (\nabla \varphi_{t_{\tau}}\,\mu_{\tau}),    
\end{equation}
where $\rho_{t_\tau}\propto e^{-2\varphi_{t}}$. In this case, as $t_{\tau}$ is annealed across time, one adopts a moving-frame Gauge correction which distinguishes the Fokker-Planck dynamics from adiabatic evolution. We find
for
\begin{equation}\label{eq:psi-def}
    \Psi_\tau:=\frac{\mu_\tau}{\sqrt{\rho_{t_\tau}}}
\end{equation}
that 
\begin{align*}
\partial_\tau\Psi_\tau &= \partial_\tau\left( \frac{\mu_{\tau}}{ \sqrt{\rho_{t_{\tau}}}} \right) = \frac{\partial_{\tau} \mu_{\tau}}{\sqrt{\rho_{t_{\tau}}}} - \frac{\dot{t}(\tau)}{2} \frac{\partial_t \rho_{t_{\tau}}}{\rho_{t_{\tau}}} \left( \frac{\mu_{\tau}}{\sqrt{\rho_{t_{\tau}}}} \right) \\
&= -\widehat{H}(t_{\tau})\Psi_\tau - \frac{\dot{t}(\tau)}{2} \partial_t\log\rho_{t_{\tau}}\, \Psi_\tau .\label{eq:gauge}
\end{align*}
This introduces the Gauge-correction 
\[
\widehat{G}_{\tau}=\frac{\dot{t}(\tau)}{2} \partial_t\log\rho_{t_{\tau}} = \dot{t}(\tau) M_{t_{\tau}},
\]
in the conservative case. We detail the non-conservative extension in Section~\ref{sec:noncons_exten}.

\subsection{Score-Matching, Diffusion Models, and Fisher Geometry.}\label{sec:diffusion_VP}

Diffusion models depend on a training phase that learns the Fisher-score $\nabla_x \log p_s(x)$, and an inference phase that uses it for sampling. One first defines a target family of intermediate distributions $(p_s )_{s\in [0,T]}$, with $p_{0}=p_{\mathrm{data}}$ and a tractable reference such as a Gaussian $p_{T} \approx p_{\mathrm{ref}}$, and learns its scores $\nabla \log p_{s}$. For sampling, we use the reversed interpolation coordinate $t=T-s$ so that $t=0$ indexes the reference end and $t=T$ indexes the data end. Given this family, one performs score-matching \cite{hyvarinen05a, song2021scorebased} to fit a parametric score $\stheta_{s}(x)$ ($:=\nabla\log\rho_s^{\theta}(x)$ if conservative) to minimize the Fisher-divergence to scores from this family,
\vspace{-0.5em}
\begin{equation}\label{eq:FishDiv}
\min_{\theta} \frac{1}{2}\int_{\mathbb{R}^{d} \times [0,T]}  |\stheta_{s}(x)-\nabla\log p_s (x)|^{2} p_s (dx) \,ds.
\end{equation}
\vspace{-0.8em}

The training phase is often done via denoising score matching \cite{denoising_sm} which uses the Gaussian-convolved family $p_s = p_{\mathrm{data}} * \mathcal{N}(0,\,\sigma_s^2 )$ and provides a statistically consistent surrogate using samples from $p_{\rm data}$. This is often called the \emph{forward} noising process.

Diffusion models \cite{song2021scorebased} extend score-matching by introducing an annealing schedule along diffusion time $t \in [0,T]$, a time-indexed \emph{effective} interpolation density $(\rho_t^{\leftarrow })_{t\in[0,T]}$, and an algorithmic time $\tau$ mapped to diffusion time by the schedule $\tau(t)$.

For stochastic score-based samplers considered here, sampling can be expressed in the canonical time-inhomogeneous Langevin form on the time-indexed interpolation density $(\rho_t^{\leftarrow })_{t\in[0,T]}$ in algorithmic time $\tau$:
\begin{equation}\label{eq:Langevin_score}
dx_{\tau}= \nabla \log \rho_{t_{\tau}}^{\leftarrow}(x_{\tau})d\tau + \sqrt{2}\,dW_{\tau},
\end{equation}
with marginal $\mu_{\tau} = \mathrm{Law}(x_{\tau})$. Two free choices -- the interpolation family $(\rho_t^{\leftarrow})_{t\in[0,T]}$ and the schedule $\tau(t)$ -- together determine the fast relaxation dynamics and rate at which the target moves in the canonical Langevin formulation~\eqref{eq:Langevin_score}.

To study diffusion models, we consider the general form \eqref{eq:Langevin_score}, agnostic to the specific interpolation density $(\rho_{t}^{\leftarrow})$ and schedule $\tau(t)$ while discussing the three parametrizations below.

\begin{remark}[Diffusion-Model Parametrizations]\label{rem:parametrization}

We first remark on how the proposed adiabatic framework fits in the context of diffusion modeling parameterizations. While denoising score-matching is usually formulated using a data-to-noise coordinate, throughout this work we reindex the same family in the opposite direction (time-forward to data), so that $t=0$ and $t=T$ index the reference and data distributions. Let us then denote:
\begin{align*}
    s &:= \text{Forward noising time (data} \to \text{noise)}, \\
    t &:= \text{Sampling coordinate, } T-s \text{ (noise} \to \text{data)}, \\
    \tau &:= \text{Algorithmic time}, \\
    t_\tau &:= \text{Interpolation parameter along } \tau.
\end{align*}

And let $\nabla_x\log p_s(x)$ be the score learned during training, distinguished from the effective score of the interpolation $S_t(x)=\nabla\log\rho_t^{\leftarrow}$ on the instantaneous frozen equilibrium $\rho_t^{\leftarrow}$.

\paragraph{Annealed Score-Based Langevin Dynamics \cite{annealed_langevin}.}

The framework of score-based generative modeling via an annealed Langevin dynamics was introduced in \cite{annealed_langevin} and is the simplest realization of the proposed adiabatic framework. It defines a sequence of forward Gaussian-convolved marginals $(p_{s})_{s \in[0,T]}$ along noising coordinate $s$, with $t=T-s$ denoting the corresponding noise-to-data coordinate. The continuous-time interpolation of the annealed score-based Langevin dynamics \cite{annealed_langevin} takes the form
\[
dx_{u} = \frac{1}{2}\nabla_x\log p_{T-t(u)}(x_u) du + d \bar{W}_{u}.
\]
For the canonical clock $\tau = u/2$, this is then 
\[
dx_{\tau} = \nabla_x\log p_{T-t_\tau}(x_\tau) d\tau + \sqrt{2}d \bar{W}_{\tau}.
\]
Thus,~\eqref{eq:Langevin_score} is recovered with
\[
\rho_{t}^{\leftarrow}=p_{T-t} , \quad S_{t}(x) = \nabla \log p_{T-t} (x).
\]
So that the annealed score-based Langevin dynamics directly recovers the canonical adiabatic setup: the smoothed diffusion marginal $p_{T-t_\tau}$ of the data density is the exact frozen equilibrium $\rho_{t_{\tau}}^{\leftarrow}$, the stationary law of the Langevin generator $L_{t_{\tau}}^{*}\,p_{T-t_\tau}=0$ for fixed time, and $\mu_{\tau} = \mathrm{Law}(x_{\tau})$ is the evolved sampling law. Thus, Thm.~\ref{thrm:adiabatic_FK} directly quantifies when $\mu_{\tau}$ tracks this moving equilibrium $p_{T-t_\tau}$ as the noise-scale is annealed.

\paragraph{Reverse-SDE Parametrizations \cite{song2021scorebased}.}

Conventional reverse-SDE parametrizations differ from the pure form of annealed Langevin considered in \cite{annealed_langevin} as their reverse drift need not have the forward marginal $p_{T-t_\tau}$ as its frozen equilibrium. They generally define a conservative drift $S_{t}(x) = \nabla \log \rho_{t}^{\leftarrow} (x)$ with respect to a different effective frozen equilibrium family $\rho_{t_{\tau}}^{\leftarrow}$. The corresponding effective families are derived below.

In score-based diffusion one defines a forward stochastic differential equation (SDE) \cite{song2021scorebased} of the general form 
\[
dx=f(x,s)\,ds+g(s)dW.
\]
The corresponding reverse-time SDE \cite{Anderson1982ReversetimeDE}, which generates samples from the target distribution, is given by:
\begin{align*}
dx=[f(x,s)-&g^2(s) \nabla_x\log p_{s}(x)]ds\\
&+g(s)d\bar{W}
\end{align*}
which, in our sampling coordinate $t$ is given by
\begin{align*}
    dx=[-f(x,T-t)+&g^2(T-t) \nabla_x\log p_{T-t}(x)]dt\\
&+g(T-t)d\bar{W}
\end{align*}

To map this generative process into the canonical imaginary-time adiabatic transport governed by Equation~\eqref{eq:learned-fokker-planck}, we define the algorithmic time clock $\tau$
\[
d\tau = \frac{1}{2} g^{2}(T-t)dt , \quad \frac{dt}{d\tau} = \frac{2}{g^{2}(T-t)}.
\]
Under this clock, the dynamics are mapped to
\begin{equation}\label{eq:rev_sde_tau}
dx_{\tau}=\left[2 \nabla_x\log p_{T-t}(x_{\tau}) -\frac{2f(x_{\tau},T-t)}{g(T-t)^{2}} \right]d\tau+\sqrt{2}d\bar{W}_{\tau}.
\end{equation}
To write this in the standard form~\eqref{eq:Langevin_score} implies that the effective target density $\rho_t^{\leftarrow}$ whose conservative score exactly matches the reverse drift must satisfy for a scalar potential $U_{t}(x)$
\begin{align}\label{eq:eff_drift}
& \nabla\log\rho_t^{\leftarrow}(x) = \nabla \log (p_{T-t}^{2}(x) e^{U_{t}(x)}),\\
&\nabla U_{t}(x) = -\frac{2f(x,T-t)}{g(T-t)^{2}}.
\end{align}
so that the effective frozen equilibrium is
\begin{equation}\label{eq:target}
\rho_t^{\leftarrow}(x) = \frac{1}{Z_{t}} p_{T-t}^{2}(x) e^{U_{t}(x)}
\end{equation}
whenever the partition function $Z_{t}<\infty$ and the effective drift~\eqref{eq:rev_sde_tau} is conservative ~\eqref{eq:eff_drift}. We detail two common cases below.

\paragraph{Variance-Exploding (VE) Schedule \cite{song2021scorebased}.} The VE parameterization defines a process where the noise scale grows unboundedly (convolved with expanding Gaussian kernels). It has forward SDE,
\[
f(x,s)=0, \,\, \textrm{and}\,\,g(s)=\sqrt{\frac{d}{ds}[\sigma^2(s)]},
\]
and forward interpolation family given by the Gaussian convolution $p_s(x)=\rho_{\mathrm{data}}*\mathcal{N}(0,\sigma^2(s)\,I)$. The associated reverse SDE, in the convention of \cite{song2021scorebased} with $s$ decreasing toward the data, is
\[
dx_{s}=-g^2(s)\nabla\log p_s(x)ds+g(s)d\bar{W}.
\]
In our convention, we notate time $t=T-s$ increasing in the noise-to-data direction so that
\[
dx_{t}=g^2(T-t)\nabla\log p_{T-t}(x)dt+g(T-t)d\bar{W}_{t}.
\]
In this case, $U_{t}=0$ so that the target~\eqref{eq:target} is $\rho_t^{\leftarrow} \propto p_{T-t}^{2}$. 

\paragraph{Variance-Preserving (VP) Schedule \cite{song2021scorebased}.}
The VP parameterization confines the process harmonically, bounding the variance so that the distribution converges to a standard normal prior $\mathcal{N}(0,{I})$ in the forward process. The forward SDE is given by:
\[
f(x,s)=-\frac{1}{2}\beta(s)x, \quad \textrm{and}\quad g(s)=\sqrt{\beta(s)},
\]
with interpolation family given by the distribution of the scaled convolution for the forward kernel
\[
x_s|x_0\,\sim\mathcal{N}(x_0 \,e^{-\frac{1}{2}\int_0^s\beta(r)dr},\,I\,(1-e^{-\int_0^s\beta(r)dr})).
\]
The reversal \cite{Anderson1982ReversetimeDE} is given by the SDE
\begin{equation}\label{eq:rev_sde}
dx_{s} = \left[ - \frac{1}{2} \beta(s) x_{s} - \beta(s) \nabla \log p_{s}(x_{s})  \right]ds + \sqrt{\beta(s)} d\bar{W}_{s},
\end{equation}
where the process is in the convention of \cite{song2021scorebased} with decreasing $s$. Parametrizing with sampling coordinate $t=T-s$ with $\frac{dt}{ds} = -1$, this is
\begin{align*}
d x_{t} = \left[  \frac{1}{2} \beta(T-t) x_{t} + \beta(T-t) \nabla \log p_{T-t}(x_{t})  \right]&dt \\
 + \sqrt{\beta(T-t)} &d\bar{W}_{t}.
\end{align*}
To express in the form~\eqref{eq:Langevin_score}, one observes $U_{t}(x)=\frac{1}{2}\lVert x \rVert_{2}^{2}$ so the target~\eqref{eq:target} is
\[
\rho_t^{\leftarrow}(x) = \frac{1}{Z_{t}} p_{T-t}^{2}(x) e^{\lVert x \rVert_{2}^{2}/2 }.
\]
Here, if the target $p_{T-t}(x) \propto e^{-\lVert x \rVert_{2}^{2}/2}$ then $\rho_t^{\leftarrow}(x) \propto p_{T-t}(x)$.


Thm.~\ref{thrm:adiabatic_FK} applies if $\rho_t^{\leftarrow}$ satisfies a Poincar{\'e} inequality and if relative form-boundedness holds. The sufficient Log-Sobolev conditions of Corollary~\ref{cor:lsi_form} immediately imply this requirement.
\end{remark}

We remark that \(p_s\), \(\rho_t^{\leftarrow}\), and \(\mu_\tau\) play distinct roles. Under exact initialization and exact score, the reindexed marginal $p_{T-t}$ is a particular non-stationary solution of the reverse-time Fokker-Planck equation \cite{Anderson1982ReversetimeDE, song2021scorebased}. Consequently, the idealized sampling law for the specific SDE of \cite{Anderson1982ReversetimeDE} satisfies $\mu_{\tau} = p_{T-t_{\tau}}$, and the exact reverse process reproduces the prescribed forward marginal path.

This identity identifies which non-equilibrium trajectory is associated to a particular forward diffusion, but it does not provide a general stability theory for other score-based samplers, such as the annealed score-based Langevin dynamics \cite{annealed_langevin}. Nor does it alone explain why the perturbed law $\mu_{\tau}$ should relax back to a distinguished moving density, how errors caused by initialization away from $p_{T}$ can be erased, or how deviations of the learned field \( \stheta_t - \nabla \log \rho_{t}^{\leftarrow} \) from the underlying score affect the terminal distribution.

In this work, we sharply control the error of the actual non-equilibrium sampling law $\mu_{\tau}$ with respect to the frozen equilibrium $\rho_{t_{\tau}}^{\leftarrow}$ of the canonical generator. Thus, we study the evolution of the sampler $\partial_{\tau} \mu_{\tau} = L_{t_{\tau}}^{*} \mu_{\tau}$ with respect to tracking of the instantaneous frozen equilibrium $\rho_{t_{\tau}}^{\leftarrow}$ satisfying $L_{t_{\tau}}^{*} \rho_{t_{\tau}}^{\leftarrow} = 0$. While the density path of the instantaneous equilibrium $\rho_{t_{\tau}}^{\leftarrow}$ of the diffusion-model generator need not coincide with that of $p_{T-t_\tau}$, Thm.~\ref{thrm:adiabatic_FK} and Thm.~\ref{thrm:ad_diffusion} ensure that the spectral gap $\Delta$ of adiabatic transport \cite{Born1928,Jansen2007,Morita2008} under the Score Hamiltonian provides explicit \emph{damping} of errors. This is in contrast to standard initial-value style bounds \cite{chen2023sampling, lee2022convergence, benton2024nearly, yang2022convergence}. The adiabatic regime provides a boundary-value style control, locally erasing trajectory error accumulation with a spectrally damped memory kernel. This reduces, in the adiabatic limit, the error to an endpoint reconstruction floor (Thm.~\ref{thrm:ad_diffusion}). This ensures that even when the process is displaced from its nominal path, the Langevin dynamics continually damps it toward the moving equilibrium.

\textbf{Quantum Potential.} The quantum potential or Bohm potential $Q(\rho)$ \cite{Bohm1952} is the first-variation of the Fisher-information functional \cite{Gianazza2008, Tsekov2008}
\[
Q (\rho)= -\frac{\nabla^{2} \sqrt{ \rho}}{\sqrt{\rho}} \propto \frac{\delta}{\delta \rho}\left( \frac{1}{2}\int \rho |\nabla\log\rho|^{2} \right).
\]
and is the information-theoretic analogue of the thermodynamic potential, $F(\rho) \propto - \frac{1}{2}\, \log \rho = \delta H/\delta \rho$, given by the first variation of the differential entropy $H(\rho):P_{2}(\mathbb{R}^{d}) \to \mathbb{R}$ \cite{chewi2024statistical}. This potential arises in the Hydrodynamical form of the Schr{\"o}dinger equation, known as the \emph{Madelung Equations} \cite{Madelung1927}, and distinguishes quantum from classical mechanics. We denote $Q = +\frac{\nabla^{2} \sqrt{ \rho}}{\sqrt{\rho}} \propto -{\delta I}/{\delta \rho}$ to be positively signed -- consistent with the variation of the Fisher information -- but note that the quantum potential $Q$ is often notated as the negative $Q = -\frac{\nabla^{2} \sqrt{ \rho}}{\sqrt{\rho}}$ \cite{Madelung1927}. Following \cite{Nelson1966, NELSON2020, Bloch2022}, the quantum potential is expressed in terms of the score function for a smooth, positive density $\rho$ via the \emph{Bohm Score Identity} when it is well-defined,
\begin{equation}\label{eq:bohm_score_id}
Q (\rho)=  \left(\frac{1}{2} \nabla \cdot S(x) + \frac{1}{4} \left|  S\right|^{2}\right).
\end{equation}
where we have notated $S=\nabla\log \rho$. The quantum potential more generally extends to amplitudes $R \in L^{2}$ as $+\frac{\nabla^{2} R}{R}$, generalizing the nodeless case.

\section{Proofs of Main Results}\label{appendix:main-proofs}

We first develop the broader Information-Hamiltonian construction as an information theoretic class of Hamiltonian constructed from the Quantum potential \cite{Madelung1927, Bohm1952} (which has also been referred to as the Information Potential, \cite{Hiley2004}) supposing various observables. These include using the Bohm potential as a function of amplitude $R$, density $\rho$, or score $S$ as the multiplicative potential. The amplitude case is the most general, allowing treatment of signed amplitudes and nodal states.

\begin{definition}[Information-Hamiltonian]\label{def:ScoreHamiltonian}
    Let $R\in L^{2}$ be a smooth, signed amplitude  with $\rho = |R|^{2} \in W^{1,2}_{\mathrm{loc}}$ a probability density. We define the generalized \emph{Information Hamiltonian} as a Schr{\"o}dinger operator in $L^{2}(dx)$ whose multiplication potential is the (negative) Quantum potential of this amplitude:
    \begin{align}\label{eq:FisherHamiltonian_inf_supp}
    &\widehat{H} =  -\nabla^{2} + \nabla^{2}R/R.
    \end{align}
    This form is well-defined wherever $\nabla^{2}R/R$ is finite; it extends across nodal sets when $\nabla^{2}R/R$ admits a removable singularity. When $\rho>0$ is nodeless with amplitude $\sqrt{\rho}$, Equation~\eqref{eq:FisherHamiltonian_inf_supp} reduces to the \emph{Fisher-Hamiltonian},
\begin{align}\label{eq:Fish_sup}
        &\widehat{H} =  -\nabla^{2} + \nabla^{2}\sqrt{\rho}/\sqrt{\rho}.
    \end{align}
    Given $-(1/2) \log \rho$ is sufficiently smooth, the Bohm-Score-Identity holds globally and allows us to expand in terms of the score-field $S = \nabla \log \rho$, yielding the \emph{Score Hamiltonian}
\begin{equation}\label{eq:ScoreHamiltonian}
\widehat{H} = -\nabla^{2} +  \frac{1}{2} \nabla \cdot S(x) + \frac{1}{4} \left|  S\right|^{2} 
\end{equation}
\end{definition}

Whenever $\nabla^{2}R/R$ defines a self-adjoint Information Hamiltonian~\ref{def:ScoreHamiltonian} and $R$ belongs to its operator domain, $R$ is a zero-energy eigenstate of $\widehat{H}$. Under the assumptions of Thm.~\ref{thrm:BH_LG_equiv_amplitude}, a corresponding operator-level unitary conjugation holds (corresponding to the Langevin generator for $R=\sqrt{\rho}$). In the classic case of the Witten Laplacian \cite{Witten1982, Helffer2005}, this equivalence can be found in \cite{setti}.

We remark the Information Hamiltonian is more general and inverts the observable: while the Witten Laplacian is derived from a potential $\phi$ via $\rho \propto e^{-2\phi}$, here the density $\rho$, score $S$, or amplitude $R$ is the primitive observable from which the multiplication potential $Q$ is constructed. This reflects the regime of score-based diffusion, where only the learned (often non-conservative) vector-field $S^{\theta}$ is accessible. Moreover, for a generalized amplitude $R \in L^{2}$ the construction $Q = {\nabla^{2}R}/{R}$ remains well-defined for signed wavefunctions and nodal states where ${\nabla^{2}R}/{R}$ admits a well-defined extension across nodal sets.

\begin{example}[Quantum Harmonic Oscillator]\label{ex:HO}
    Consider the standard multivariate Gaussian $\rho \propto e^{-\frac{1}{2} \lVert x \rVert_{2}^{2}}$ in $\mathbb{R}^{d}$ with score $S(x)=-x$. Given $S$, its score Hamiltonian is the isotropic Harmonic Oscillator
    \vspace{-0.3em}\begin{equation}\label{eq:HO}\widehat{H}_{\mathrm{HO}}=\widehat{H}(0)=-\nabla^{2} + \frac{1}{4} \lVert  x\rVert_{2}^{2}  -\frac{d}{2} \quad\tag{Harmonic Oscillator}
    \end{equation}
    shifted to have zero ground state energy.
\end{example}

\begin{proposition}[Properties of the Score Hamiltonian]\label{prop:score_hamiltonian_props}
Let $\rho$ be a smooth, strictly positive probability density on $\mathbb{R}^d$ with score vector field $S = \nabla \log \rho$. Define its associated Score Hamiltonian on $L^2(dx)$ by \eqref{eq:ScoreHamiltonian}
\[
\widehat{H} = -\nabla^{2} +  \frac{1}{2} \nabla \cdot S(x) + \frac{1}{4} \left|  S\right|^{2} .
\]
Assume $S$ satisfies the quadratic lower envelope condition: there exists $C>0$ such that
\[
    \frac{1}{2}\|S(x)\|^2 + 2C(1+|x|^2) \ge -\nabla\cdot S(x)
    \qquad \text{for all } x\in\mathbb R^d.
\]
Then, the following properties hold:
\begin{enumerate}
    \item $\widehat{H}$ is essentially self-adjoint on $L^2(dx)$.\label{cond:SA}
    
    \item The density amplitude $\sqrt{\rho}$ is the unique, normalized, non-negative ground state of $\widehat{H}$ with zero energy (eigenvalue $E_{0}=0$), satisfying $\widehat{H}\sqrt{\rho} = 0$.\label{cond:GS}
\end{enumerate}
If, in addition, $\rho$ satisfies a Poincar{\'e} inequality then
\begin{enumerate}[resume]
\item The spectral gap $\Delta$ of $\widehat{H}$ is strictly positive $\Delta := E_1 - E_0 = E_1 > 0$ and equals the inverse of the optimal Poincar\'{e} constant of $\rho$ \cite{Helffer2005}: \label{cond:PC_gap}
    \begin{equation}\label{eq:PC_SpecGap}
    \frac{1}{\Delta} = \sup_{\substack{f \in W^{1,2}(\rho) \setminus \{0\} \\ \int f \, d\rho = 0}} \frac{\int f^2 \, d\rho}{\int \|\nabla f\|^2 \, d\rho}
\end{equation}
\end{enumerate}
If $Q(x)=\frac{1}{2}\nabla \cdot S(x)+\frac{1}{4}\norm{S(x)}^2 \uparrow \infty$ is confining as $\norm{x}\uparrow \infty$, then 
\begin{enumerate}[resume]
\item $\Delta > 0$, Equation~\eqref{eq:PC_SpecGap} holds, and $\widehat{H}$ admits a strictly non-negative discrete lower spectrum $0 = E_0 < E_1 \le E_2 \le \dots$.\label{cond:pure_pt}
\end{enumerate}
\end{proposition}

\begin{proof}[Proof of Proposition~\ref{prop:score_hamiltonian_props}]
(1) Self-adjointness of $\widehat{H}$ follows directly from the Faris-Lavine theorem (see, e.g. Thm. X.38 in \cite{Reed1975} or \cite{Sears1950}). For (2) note our multiplication potential, by definition, is $Q(\rho) = +\frac{\nabla^2 \sqrt{\rho}}{\sqrt{\rho}}$. Taking the action of $\widehat{H}$ on $\psi_{0}=\sqrt{\rho}$, we see\[
\widehat{H} \sqrt{\rho} = \left( -\nabla^2 + {\nabla^2\sqrt{\rho}}/{\sqrt{\rho}} \right) \sqrt{\rho} = -\nabla^2 \sqrt{\rho} + \nabla^2\sqrt{\rho} = 0
\]which yields that $\sqrt{\rho}$ is a zero-eigenvalue eigenfunction of $\widehat{H}$. As this eigenfunction is strictly non-negative $\psi_{0}(x) \ge 0$ the Courant nodal domain theorem guarantees it is the unique ground-state. For (3), we show that the equality with the Poincar{\'e} constant holds in Remark~\ref{rem:PC}, and refer the reader to \cite{Helffer2005} for additional discussion on the relation to the Poincar{\'e} constant. A pure point spectrum is ensured by the potential being confining (see, e.g. \cite{Reed1975}).
\end{proof}

\begin{example} For the Gaussian $\rho \propto e^{-|x|^{2}/2}$, we find $|S|^{2} = |x|^2$ and $-\Delta \log \rho = \mathbf{tr}\, \mathbb{I}_{d} = d$. Thus, the inequality $\frac{1}{2} |x|^{2} + 2C (1 + |x|^{2}) \geq d$ is immediate, as $|x|^{2}$ trivially overpowers the constant $d$ so this is satisfied. As $\frac{1}{4}|x|^{2}-d/2 \uparrow \infty$ as $|x| \uparrow \infty$, Properties \eqref{cond:SA}, \eqref{cond:GS}, \eqref{cond:PC_gap}, \eqref{cond:pure_pt} all hold.
\end{example}

\begin{example} 
For the 1D Cauchy distribution $\rho \propto \frac{1}{1+|x|^{2}}$, the score and negative Laplacian $-\Delta \log \rho$ are $S = - {2x}/{(1+|x|^{2})}$ and $- \nabla \cdot S =\frac{2}{1+|x|^{2}} - \frac{4|x|^{2}}{(1+|x|^{2})^{2}} $. Both decay to $0$ as $|x| \uparrow \infty$, so one can easily pick a small constant $C>0$ as a quadratic lower blanket for $\hat{V}$. Thus, properties \eqref{cond:SA} and \eqref{cond:GS} of Proposition~\ref{prop:score_hamiltonian_props} apply to heavy-tailed distributions such as the Cauchy distribution. However, \eqref{cond:PC_gap} and \eqref{cond:pure_pt} do not apply as it does not satisfy a Poincar{\'e} inequality.
\end{example}

\begin{example} 
For the Hydrogen Atom ground-state $\rho \propto e^{-2\lVert x \rVert_{2}}$, the score is $S = - 2\frac{x}{r} $ and the score divergence is $ \nabla \cdot S = -2 \nabla \cdot (\frac{x}{r}) = -2(\frac{2}{r}) =-\frac{4}{r}$. Thus, the potential scales proportional to $\mathrm{const}-\frac{1}{r}$ which similarly decays to a constant as $|x| \uparrow \infty$ so there exists $C$ ensuring the quadratic lower-blanket condition holds, implying properties \eqref{cond:SA} and \eqref{cond:GS} of Proposition~\ref{prop:score_hamiltonian_props}. $\rho \propto e^{-2\lVert x \rVert_{2}}$ \emph{does} satisfy Poincar{\'e} inequality, so \eqref{cond:PC_gap} holds. However, \eqref{cond:pure_pt} does not.
\end{example}

We remark that the existence of a positive Poincar{\'e} constant (and spectral gap $\Delta_t>0$) is commonly ensured by the typical Gaussian convolution $\rho_{t}=\rho_{\mathrm{data}}*\mathcal{N}(0, \sigma_{t}^{2})$ under standard confinement assumptions on $\rho_{\mathrm{data}}$. In particular, the convolution $\rho_{t}$ produces a smooth, strictly positive density everywhere.

We are now ready to prove the main result. 

\subsection{Proof of Theorem \ref{thrm:adiabatic_FK}}\label{sec:main-proof}
Suppose that $\Psi_\tau=\mu_\tau /\sqrt{\rho_{t_\tau}}$ is defined as in \eqref{eq:psi-def} satisfies the evolution \eqref{eq:adiabatic-with-gague}. By Proposition \ref{prop:score_hamiltonian_props}, $\widehat{H}_t$ is essentially self-adjoint with zero-energy ground state $\sqrt{\rho_{t}}$, and the Poincar{\'e} inequality on $\rho_{t}$ holds, yielding a strictly positive spectral gap $\Delta(t)>0$. Write $\Psi_{\perp,\tau} = \Psi_\tau - \psi_0$ the orthogonal lag for $\psi_{0}(t) = \sqrt{\rho_{t}}$ the target ground-state and observe that
\begin{align*}
\langle \Psi_{\tau}, \psi_{0}(\tau) \rangle=  \int \frac{\mu_{\tau}}{\sqrt{\rho_{t_{\tau}}}} \,\sqrt{\rho_{t_{\tau}}} = \int \mu_{\tau} =1
\end{align*}
Thus as the projection is fixed, we may always write $\Psi_{\tau} = \psi_{0}(\tau) + \Psi_{\perp,\tau}$ for $\langle \Psi_{\perp,\tau},  \psi_{0}(\tau) \rangle=0$. Now, let us consider the evolution of the error $\omega = \lVert \Psi_{\perp,\tau} \rVert_{L^2}$. Differentiating the squared error gives
\begin{align*}
\frac{1}{2}\partial_\tau \omega^2 &= \langle \Psi_{\perp,\tau}| \partial_\tau \Psi_\tau - \dot{t}\partial_t \psi_0 \rangle \\
&= \langle \Psi_{\perp,\tau}| -\widehat{H}_\tau \Psi_\tau - \hat{G}_\tau \Psi_\tau - \dot{t}\partial_t \psi_0 \rangle
\end{align*}
Note that $\langle \Psi_{\perp,\tau}|\widehat{H}_\tau \psi_0(\tau) \rangle= 0$. Thus, this simplifies to
\begin{align*}
&\frac{1}{2}\partial_\tau \omega^2 = -\langle \Psi_{\perp,\tau}| \widehat{H}_\tau \Psi_{\perp,\tau} \rangle - \langle \Psi_{\perp,\tau} | \hat{G}_\tau \psi_0 \rangle &
\\
&- \langle \Psi_{\perp,\tau}| \hat{G}_\tau \Psi_{\perp,\tau} \rangle - \dot{t} \langle \Psi_{\perp,\tau}| \partial_t \psi_0 \rangle
\end{align*}
applying Cauchy-Schwarz
\begin{align*}
-\dot{t} \langle \Psi_{\perp,\tau}| \partial_t \psi_0 \rangle &\le |\dot{t}| \omega \lVert \partial_t \psi_0 \rVert_{L^2} \\
-\langle \Psi_{\perp,\tau}| \hat{G}_\tau \psi_0 \rangle &\le \omega \lVert \hat{G}_\tau \psi_0 \rVert_{L^2} = \omega \frac{|\dot{t}|}{2} \lVert \partial_{t} \log \rho_{t} \rVert_{L^2 (\rho_t )}
\end{align*}
where, $\omega \lVert \hat{G}_\tau \psi_0 \rVert_{L^2} = \omega \frac{|\dot{t}|}{2} \lVert \partial_{t} \log \rho_{t} \rVert_{L^2 (\rho_t )}$ holds by evaluating $\lVert \hat{G}_{\tau}\psi_{0} \rVert_{L^{2}}^{2}$,
\begin{align*}
&\lVert \hat{G}_{\tau}\psi_{0} \rVert_{L^{2}}^{2} \\
&= \int \left| \frac{\dot{t}}{2}\,\partial_t\log\rho_{t}(x) \psi_{0}(x) \right|^{2} \\
&=\frac{\dot{t}^{2}}{4}\int \,(\partial_t\log\rho_{t}(x))^{2}  \rho_{t}(x) dx = \frac{\dot{t}^{2}}{4} \lVert \partial_{t} \log \rho_{t} \rVert_{L^2 (\rho_t )}^{2}
\end{align*}
Note that $\lVert \partial_{t} \log \rho_{t} \rVert_{L^2 (\rho_t )}^{2}$ also equals the variance $\operatorname{Var}_{\rho_t}(\partial_t \log \rho_t)$, as $\partial_{t} \log \rho_{t}$ is mean-zero. For the quadratic Gauge term, we isolate the negative part $M_{t,-}(x) = \max\{-\frac{1}{2}\partial_t \log \rho_{t}(x), 0\}$. Since $\Psi_{\perp,\tau} \perp \psi_0$, the assumed relative form-bound guarantees:
\begin{align*}
&-\langle \Psi_{\perp,\tau}| \hat{G}_\tau \Psi_{\perp,\tau} \rangle \le |\dot{t}| \int M_{t,-} \Psi_{\perp,\tau}^2 \, dx \\
&\le |\dot{t}| \Big( a(t) \langle \Psi_{\perp,\tau}|\widehat{H}_{\tau}|\Psi_{\perp,\tau} \rangle + b(t) \|\Psi_{\perp,\tau}\|_{L^2}^2 \Big)
\end{align*}
Thus, the Hamiltonian and Gauge cross-terms combine to bound the energy strictly by the Gauge-corrected gap $\Gamma(\tau) = (1 - |\dot{t}|a(t))\Delta_t - |\dot{t}|b(t)$:
\begin{align*}
&-\langle \Psi_{\perp,\tau}| \widehat{H}_\tau \Psi_{\perp,\tau} \rangle - \langle \Psi_{\perp,\tau}| \hat{G}_\tau \Psi_{\perp,\tau} \rangle \\
&\qquad \le -\Gamma(\tau) \|\Psi_{\perp,\tau}\|_{L^2}^2 = -\Gamma(\tau) \omega^2
\end{align*}
Returning to the full differential error equation, we have:
\begin{align*}
\frac{1}{2}\partial_\tau \omega^2 &\leq -\Gamma(\tau) \omega^2 + \omega \left( \lVert \hat{G}_\tau \psi_0 \rVert_{L^2} + |\dot{t}| \lVert \partial_t \psi_0 \rVert_{L^2} \right).
\end{align*}
Dividing by $\omega$ yields
\[
\, \dot{\omega} \leq -\Gamma(\tau) \omega +  \lVert \hat{G}_\tau \psi_0 \rVert_{L^2} + |\dot{t}| \lVert \partial_t \psi_0 \rVert_{L^2},
\]
since 
\[
\lVert \hat{G}_\tau \psi_0 \rVert_{L^2} = \frac{|\dot{t}|}{2}\lVert \partial_{t} \log \rho_{t} \rVert_{L^2 (\rho_t )},
\]
and
\[
\lVert \partial_t \psi_0 \rVert_{L^2} = \frac{1}{2} \lVert \partial_{t} \log \rho_{t} \rVert_{L^2 (\rho_t )}.
\]
Thus,
\begin{align*}
\dot{\omega} \le - \Gamma(\tau) \omega + |\dot{t}| \lVert \partial_{t} \log \rho_{t} \rVert_{L^2 (\rho_t )},
\end{align*}
and by Gr\"onwall's inequality, 
\begin{align*}
&\omega(\tau) \le  \omega(0) \exp\left(-\int_0^\tau \Gamma(s)ds\right)  \\
&+ \int_0^\tau |\dot{t}(s)| \lVert \partial_{t} \log \rho_{t(s)} \rVert_{L^2 (\rho_{t(s)} )} \exp\left(-\int_s^\tau \Gamma(r)dr\right) ds.  
\end{align*}
To evaluate the steady-state tracking integral, we multiply and divide the integrand by $ \Gamma(s)$ and pull out the supremum:
\begin{align*}
&\int_0^\tau \frac{|\dot{t}(s)| \cdot \lVert \partial_{t} \log \rho_{t(s)} \rVert_{L^2 (\rho_{t(s)} )}}{ \Gamma(s)} \frac{d}{ds} \exp\left(-\int_s^\tau  \Gamma(r)dr\right) ds \\
&\le \sup_{s \le \tau} \left( \frac{|\dot{t}(s)| \lVert \partial_{t} \log \rho_{t(s)} \rVert_{L^2 (\rho_{t(s)} )}}{\Gamma(s)} \right) \left[ 1 - e^{-\int_0^\tau  \Gamma(r)dr} \right] \\
&\le \sup_{s \le \tau} \frac{|\dot{t}(s)| \lVert \partial_{t} \log \rho_{t(s)} \rVert_{L^2 (\rho_{t(s)} )}}{\Gamma(s)}.
\end{align*}
Thus, the total tracking error is bounded by the sum of the exponentially decaying initial condition and the adiabatic tracking error, concluding the proof. \hfill $\square$

\begin{remark}[$L^{2}$ to $\chi^{2}$]
For the translation of the $L^2$ bound into classical probability divergences, substitute $\Psi_\tau = \mu_\tau / \sqrt{\rho_{t_{\tau}}}$ and note that the $L^2(dx)$ distance on the amplitudes is identical to the $\chi^2$ divergence: 
\begin{align*}
&\left\lVert \frac{\mu_\tau}{\sqrt{\rho_{t_{\tau}}}} - \sqrt{\rho_{t_{\tau}}} \,\right\rVert_{L^2(dx)}^{2} \\
&= \int \left( \frac{\mu_\tau - \rho_{t_{\tau}}}{\sqrt{\rho_{t_{\tau}}}} \right)^{2} dx = \chi^{2}(\mu_\tau \| \rho_{t_{\tau}}) .\end{align*}
\end{remark}

\begin{corollary}[]\label{cor:lsi_form} Suppose the target sequence $(\rho_{t})_{t\in[0,T]}$ satisfies a logarithmic Sobolev inequality (LSI) and that, for some $\lambda_{t}>0$, $\int e^{\lambda_{t}M_{t,-}} \, d\rho_{t}<\infty.$ Then, $M_t$ satisfies the relative form-boundedness condition in Thm.~\ref{thrm:adiabatic_FK}.
\end{corollary}

\begin{proof}
For the controlled Gauge term $M_{t}=\frac{1}{2} \partial_{t} \log \rho_{t}$, and $\chi \perp \psi_{0}$ with $\chi = f\sqrt{\rho_{t}}$ for $f$ a mean-zero function, let $M_{t,-} := \max \,\{ -M_{t},0 \}$. One then observes
\[
-\int_{\mathbb{R}^{d}} M_{t} \chi^{2} \leq \int_{\mathbb{R}^{d}} M_{t,-}\, \chi^{2} = \int_{\mathbb{R}^{d}} M_{t,-}\, f^{2} d\rho_{t}
\]
Recall the entropic inequality $\int u g \,d\nu \leq \mathrm{Ent}_{\nu}(u) + \int u \,d\nu \log\left( \int e^{g}d\nu\right)$ for $u\geq0$, $\nu$ a probability measure, and $g$ measurable with $\int e^{g} \,d\nu < \infty$ (see, e.g. Equation 5.1.2 of \cite{Bakry2014}). As $u=f^{2}\geq 0$ and $\rho_{t}$ is a probability measure, for $g=\lambda_{t}M_{t,-}$ and $\lambda_{t}>0$ we have
\begin{align*}
&\int_{\mathbb{R}^{d}} M_{t,-}\, f^{2} \rho_{t} \leq \frac{1}{\lambda_{t}}\mathrm{Ent}_{\rho_{t}}(f^{2}) \\
&+ \frac{1}{\lambda_{t}}\int f^{2} \, d\rho_{t} \log
\int e^{\lambda_{t}M_{t,-}} \, d\rho_{t} .
\end{align*}
Recall $\int \lVert \nabla f \rVert_{2}^{2} d\rho_{t}  = \langle \chi | \widehat{H}_{\tau} | \chi \rangle$ and $\int f^{2} d\rho_{t}  = \langle \chi \,| \,\chi \rangle_{L^{2}}$. By Log-Sobolev inequality (Def. 5.1.1~\cite{Bakry2014}), we then have $\mathrm{Ent}_{\rho_{t}}(f^{2}) \leq 2 C_{\mathrm{LSI}}(t) \lVert \nabla f \rVert_{L^{2}(\rho_{t})}^{2}$ and thus
\begin{align*}
&\leq \frac{2}{\lambda_{t}} C_{\mathrm{LSI}}(t) \lVert \nabla f \rVert_{L^{2}(\rho_{t})}^{2} +  \frac{1}{\lambda_{t}}\int f^{2} \, d\rho_{t}\log
\int e^{\lambda_{t}M_{t,-}} \, d\rho_{t} \\
&=\frac{2C_{\mathrm{LSI}}(t)}{\lambda_{t}} \langle \chi | \widehat{H}_{\tau} | \chi \rangle + \frac{1}{\lambda_{t}}\log
\int e^{\lambda_{t}M_{t,-}} \, d\rho_{t} \langle \chi \,| \,\chi \rangle_{L^{2}}
\end{align*}
which yields the desired relative form bound with $a(t)=\frac{2C_{\mathrm{LSI}}(t)}{\lambda_{t}}$ and $b(t) = \frac{1}{\lambda_{t}}\log\left(
\int e^{\lambda_{t}M_{t,-}} \, d\rho_{t}\right)$.
\end{proof}

\begin{remark}
Alternatively, if $\rho_{t}$ satisfies a weighted $W^{1,2} \to L^{4}$ Sobolev inequality, then $M_{t,-} \in L^{2}(\rho_{t})$ ensures the required form-bound by applying Cauchy-Schwarz,
\[
\int_{\mathbb{R}^{d}} M_{t,-}\, f^{2} \rho_{t} \leq \lVert M_{t,-} \rVert_{L^{2}(\rho_{t})} \,\lVert f \rVert_{L^{4}(\rho_{t})}^{2}
\]
The Sobolev-inequality gives
\[
\lVert f \rVert_{L^{4}(\rho_{t})}^{2} \leq C_{\mathrm{SI}}(t)  \int ( \lVert  \nabla f \rVert_{2}^{2} + f^{2} )d\rho_{t}\,, 
\]
so that in this case one concludes
\begin{align*}
\int_{\mathbb{R}^{d}} M_{t,-}\, f^{2} \rho_{t} &\leq \lVert M_{t,-} \rVert_{L^{2}(\rho_{t})} \,C_{\mathrm{SI}}(t) \langle \chi | \widehat{H}_{\tau} | \chi \rangle \\
&+ \lVert M_{t,-} \rVert_{L^{2}(\rho_{t})} \,C_{\mathrm{SI}}(t) \langle \chi | \chi \rangle,
\end{align*}
which implies form-boundedness with $a(t),b(t) = \lVert M_{t,-} \rVert_{L^{2}(\rho_{t})} \,C_{\mathrm{SI}}(t)$. As $f$ is mean-zero, via Poincar{\'e} inequality one may take $a(t) = \lVert M_{t,-} \rVert_{L^{2}(\rho_{t})} \,C$  for $C=C_{\mathrm{SI}}(t)(1 +\Delta^{-1})$ and $b(t) =0$.
\end{remark}

We seek to determine whether the lag represents a sharp and fundamental physical limit. To do so, we introduce a formal parameter $\epsilon$ to separate time-scales $s=\epsilon \tau$ and perform an asymptotic expansion of the error in Prop.~\ref{prop:LB_tight}. Identifying the first-order quasi steady-state balance requires appropriate form-norm control. Assuming the quasi-steady form-norm scaling below, the entropy-LSI argument of Corollary~\ref{cor:lsi_form} together with exponential integrability of $|\partial_{t} \log \rho_{t}|^{2}$ ensures the moving-frame cross-term is of second-order.

\begin{proposition}\label{prop:LB_tight}
    Suppose $\Psi_{\tau} = \mu_{\tau}/\sqrt{\rho_{t_{\tau}}}$ is the density-ratio of amplitudes of the time-inhomogeneous Fokker-Planck solution $\mu_{\tau}$ to its instantaneous equilibrium $\sqrt{\rho_{t_{\tau}}}$ and that the assumptions of Thm.~\ref{thrm:adiabatic_FK} hold. Suppose additionally that $\rho_t$ satisfies a Log-Sobolev Inequality \cite{Bakry2014} and that $\partial_{t} \log \rho_{t}$ satisfies for some $\lambda_{t}>0$ $\int e^{\lambda_{t} |\partial_{t} \log \rho_{t}|^{2}}d\rho_{t}<\infty$ uniformly on $t\in [0,T]$. Suppose the potential $Q$ is confining in the sense of Prop.~\ref{prop:score_hamiltonian_props}, so that Property~\ref{cond:pure_pt} holds. Let
    \[
    \beta(t) = 2 \lVert \widehat{H}_{\perp}^{-1} \partial_{t}\psi_{0} \rVert_{L^{2}} = \lVert \widehat{H}_{\perp}^{-1} (\partial_{t}\log \rho_{t} )\,\psi_{0}(t) \rVert_{L^{2}}
    \]
    In the adiabatic quasi steady-state regime, the tracking error satisfies
    \begin{align*}
    &\lVert \Psi_{\tau} - \psi_{0}(t_{\tau}) \rVert_{L^{2}}
=\,|\dot{t}|\,\beta(t_{\tau}) + \mathcal{O}(\epsilon^{2}), \\
&|\dot{t}|\beta(t_{\tau}) \leq |\dot{t}|\,\frac{\lVert \partial_{t}\log\rho_{t(s)} \rVert_{L^{2}(\rho_{t})}}{\Delta_t} .
\end{align*}
    Further assuming a gap-dominated regime, where a bounded fraction $c\in(0,1]$ is in the first excited subspace: $\sum_{k\,:\,E_{k}=E_{1}}|\langle \partial_{t}\psi_{0}, \psi_{k} \rangle|^{2} \geq c^{2} \lVert \partial_{t}\psi_{0} \rVert_{L^{2}}^{2}$. Then,
     \begin{align*}
    &c|\dot{t}| \,\frac{\lVert \partial_{t}\log\rho_{t(s)} \rVert_{L^{2}(\rho_{t})}}{\Delta_t} \\
    &\leq |\dot{t}|
    \beta(t)  \leq |\dot{t}|\,\frac{\lVert \partial_{t}\log\rho_{t(s)} \rVert_{L^{2}(\rho_{t})}}{\Delta_t} 
    \end{align*}
    Consequently, for $\chi_{\tau} = \Psi_{\tau} - \psi_{0}(t_{\tau})$, given that the quasi-steady state regime provides $\lVert \chi_{\tau} \rVert_{L^2}^{2} + \langle \chi_{\tau} | \widehat{H}_{\tau} |\chi_{\tau} \rangle = \mathcal{O}(\dot{t}^{2})$ and the quasi-steady state remainder $r_{\epsilon}(t) = \mathcal{O}_{t}(\epsilon^{2})$ in $\chi_\tau=-2\dot t\,\widehat H_\perp^{-1}\partial_t\psi_0+r_\epsilon(t)$ satisfies $\| r_{\epsilon}(t)\|_{L^{2}}=o\left(|\dot{t}|\,\frac{\lVert \partial_{t}\log\rho_{t(s)} \rVert_{L^{2}(\rho_{t})}}{\Delta_t} \right)$, the tracking error obeys 
\[
|\dot{t}|\,\frac{\lVert \partial_{t}\log\rho_{t(s)} \rVert_{L^{2}(\rho_{t})}}{\Delta_t} \asymp \lVert \Psi_{\tau} - \psi_{0}(t_{\tau}) \rVert_{L^{2}}
\]
so that the inverse-gap dependence is asymptotically sharp up to a constant factor.
\end{proposition}

\begin{proof}[Proof of Prop.~\ref{prop:LB_tight}]
As before, let $\Psi_{\tau} = \mu_{\tau}/\sqrt{\rho_{t_{\tau}}}$ with decomposition $\Psi_{\tau} = \psi_{0}(t_{\tau}) + \chi_{\tau}$. The moving-frame evolution is written as
\[
\partial_{\tau}\Psi_{\tau} = - \widehat{H}_{\tau} \Psi_{\tau} - \widehat{G}_{\tau} \Psi_{\tau}
\]
and, by the decomposition into the ground-state and orthogonal error,
\[
\dot{t}\partial_{t}\psi_{0} + \partial_{\tau} \chi_{\tau} = - \widehat{H}_{\perp} \chi_{\tau} - \widehat{G}_{\tau}\psi_{0} - \widehat{G}_{\tau}\chi_{\tau}
\]
Where $\widehat{G}_{\tau}\psi_{0} = \dot{t} \partial_{t} \psi_{0}$. Reorganizing terms, we see
\begin{align*}
2\dot{t}\partial_{t}\psi_{0} + \partial_{\tau} \chi_{\tau} =- \widehat{H}_{\perp} \chi_{\tau}- \widehat{G}_{\tau} \chi_{\tau}
\end{align*}
By the global bound of Thm.~\ref{thrm:adiabatic_FK}, we know that $\lVert \chi_{\tau} \rVert_{L^{2}} = \mathcal{O}(\dot{t})$. However, dropping $\widehat{G}_\tau \chi_\tau$ from the asymptotic limit requires strong $L^2$ convergence. To explicitly bound the $L^2(dx)$ norm of this cross-term, we substitute $\chi_\tau^2 = f^2 \rho_t$ 
\begin{align*}
&\|\widehat{G}_\tau \chi_\tau\|_{L^2(dx)}^2 = \frac{\dot{t}^2}{4} \int (\partial_t \log \rho_t)^2 f^2 \rho_t \, dx
\end{align*}
Applying the entropy inequality used in Corollary~\ref{cor:lsi_form}, with $a_{t}=\partial_{t} \log \rho_{t}$ gives 
\begin{align*}
&\int a_{t}^{2} f^{2}\, d\rho_{t} \leq \frac{2C_{\mathrm{LSI}}(t)}{\lambda_{t}} \langle \chi_{\tau} | \widehat{H}_{\tau} |\chi_{\tau} \rangle \\
&+ \frac{1}{\lambda_{t}} \log \int e^{\lambda_{t}a_{t}^{2}} d\rho_{t} \cdot \lVert \chi_{\tau} \rVert_{L^2}^{2}
\end{align*}
Hence, with the quasi-steady regime, we conclude that $\lVert \widehat{G}_{\tau} \chi_{\tau} \rVert_{L^{2}} = \mathcal{O}(\dot{t}^{2})$, assuming the coefficients $\frac{1}{\lambda_{t}} \log \int e^{\lambda_{t}a_{t}^{2}} d\rho_{t}$ and $\frac{2C_{\mathrm{LSI}}(t)}{\lambda_{t}}$ remain bounded in the regime considered.

To yield an asymptotic quasi-steady state bound we separate time-scales by introducing a ``slow'' macroscopic time $s\in[0,1]$ where $s = \epsilon \tau$, $\partial_{\tau} = \epsilon \partial_{s}$. Then $\dot{t} = \epsilon t'(s) = \mathcal{O}(\epsilon)$ and $\ddot{t} = \mathcal{O}(\epsilon^{2})$.
We take the asymptotic series expanded in powers of $\epsilon$:
\[
\chi_{\tau} = \epsilon \chi_{1}(s) + \epsilon^{2} \chi_{2}(s) + \cdots
\]
taking the $\tau$ partial, see $\partial_{\tau}\chi_{\tau} = \mathcal{O}(\epsilon^{2})$.
Returning to our balance equation,
\begin{align*}
&2\epsilon t' \partial_{t}\psi_{0} + \mathcal{O}(\epsilon^{2}) =- \epsilon \widehat{H}_{\perp} \chi_{1} + \mathcal{O}(\epsilon^{2})
\end{align*}
By collecting the $\mathcal{O}(\epsilon)$ terms which define the first-order balance, $\widehat{H}_{\perp} \chi_{1} = -2t'(s) \partial_{t}\psi_{0}$, and thus in $L^{2}(dx)$
\[
\widehat{H}_{\perp} \chi_{\tau} = -2\dot{t} \partial_{t}\psi_{0} +\mathcal{O}(\epsilon^{2}).
\]
Applying the resolvent on the excited subspace, we have 
\[
\chi_{\tau} = -2\dot{t} \widehat{H}_{\perp}^{-1} \partial_{t}\psi_{0}+r_{\epsilon}(t)
\]
in $L^{2}(dx)$ with remainder $r_{\epsilon}(t)$ satisfying $\|r_{\epsilon}(t)\|_{L^2} = \mathcal{O}_{t}(\epsilon^{2}).$ Hence, we have the relation for the first-order lag
\[
\lVert \chi_{\tau} \rVert_{L^{2}}^{2} = |\dot{t}|^{2} \beta(t)^{2} + \mathcal{O}_{t}(\epsilon^{3})= 4 |\dot{t}|^{2} \lVert \widehat{H}_{\perp}^{-1} \partial_{t}\psi_{0} \rVert_{L^{2}}^{2} + \mathcal{O}_{t}(\epsilon^{3}).
\]
Recall that $\partial_{t}\psi_{0} = \sum_{k\geq 1} \psi_{k} \langle \partial_{t}\psi_{0} | \psi_{k} \rangle$. Thus,
\begin{align*}
&\lVert \widehat{H}_{\perp}^{-1} \partial_{t}\psi_{0} \rVert_{L^{2}}^{2}= \sum_{k \geq 1} \frac{|\langle \partial_{t} \psi_{0}|\psi_{k} \rangle |^{2}}{E_{k}^{2}} \\
&\leq \frac{1}{\Delta_t^{2}} \lVert \partial_{t}\psi_{0} \rVert_{L^{2}}^{2} = \frac{\lVert \partial_{t}\log \rho_{t} \rVert_{L^{2}(\rho_{t})}^{2}}{4\Delta_t^{2}} 
\end{align*}
Which implies
\[
|\dot{t}|\beta(t) \leq  |\dot{t}|\,\frac{\lVert \partial_{t}\log \rho_{t}\rVert_{L^{2}(\rho_{t})}}{\Delta_t} 
\]
For the lower-bound, suppose a gap-dominated regime:
\[
\sum_{k\,:\,E_{k}=E_{1}}|\langle \partial_{t}\psi_{0}, \psi_{k} \rangle|^{2} \geq c^{2} \lVert \partial_{t}\psi_{0} \rVert_{L^{2}}^{2}.
\]
given the assumption of gap-domination, one finds
\begin{align*}
&|\dot{t}|^{2}\beta(t)^{2} = 4 |\dot{t}|^{2}\sum_{k \geq 1} \frac{|\langle \partial_{t} \psi_{0}|\psi_{k} \rangle |^{2}}{E_{k}^{2}} 
\\
&\geq \frac{4 c^{2}|\dot{t}|^{2}}{\Delta_t^{2}} \lVert \partial_{t}\psi_{0} \rVert_{L^{2}}^{2} 
= c^{2}|\dot{t}|^{2}\, \frac{\lVert \partial_{t}\log\rho_{t(s)} \rVert_{L^{2}}^{2}}{\Delta_t^{2}}.
\end{align*}
Thus,
\begin{align*}
    |\dot{t}|\,\beta(t) \geq  \,c|\dot{t}|\,\frac{\lVert \partial_{t}\log\rho_{t(s)} \rVert_{L^{2}}}{\Delta_t}
\end{align*}
Since $\lVert \Psi_{\tau} - \psi_{0}(t_{\tau}) \rVert_{L^{2}} = \lVert \chi_{\tau} \rVert_{L^{2}} =|\dot{t}|\beta(t) + \mathcal{O}(\epsilon^{2})$, one concludes
\begin{align*}
&|\dot{t}|\,\frac{\lVert \partial_{t}\log\rho_{t(s)} \rVert_{L^{2}(\rho_{t})}}{\Delta_t} + \mathcal{O}(\epsilon^{2}) \geq \lVert \Psi_{\tau} - \psi_{0}(t_{\tau}) \rVert_{L^{2}} \\
&\geq  \,c|\dot{t}|\,\frac{\lVert \partial_{t}\log\rho_{t(s)} \rVert_{L^{2}(\rho_{t})}}{\Delta_t} - \mathcal{O}(\epsilon^{2})
\end{align*}
Provided $c$ remains bounded away from zero and the $\mathcal{O}_{t}(\epsilon^{2})$ remainder is $\| r_{\epsilon}(t)\|_{L^{2}}=o\left(|\dot{t}|\,\frac{\lVert \partial_{t}\log\rho_{t(s)} \rVert_{L^{2}(\rho_{t})}}{\Delta_t} \right)$, we obtain asymptotic two-sided scaling in the gap-dominated and quasi-steady state regime:
\[
|\dot{t}|\,\frac{\lVert \partial_{t}\log\rho_{t(s)} \rVert_{L^{2}(\rho_{t})}}{\Delta_t} \asymp  \lVert \Psi_{\tau} - \psi_{0}(t_{\tau}) \rVert_{L^{2}}.
\]
\end{proof}

\begin{example}[Translated Gaussian gap-saturation]
In one dimension, take $\rho_t = \sqrt{\frac{\kappa}{2\pi}}e^{-\kappa/2 (x - m(t))^{2}}$ with Hamiltonian $\widehat{H}_{t} = - \partial_{x}^{2} + \frac{\kappa^{2}}{4}(x-m(t))^{2} - \frac{\kappa}{2}$. The spectral gap is simply $\Delta_t = \kappa$, the ground-state is $\psi_{0}(t,x) = \sqrt{\rho_{t}}$, and the first excited state is 
\[
\psi_{1}(t,x) = \sqrt{\kappa}(x - m(t))\, \psi_{0}(t,x).
\]

Differentiating $\psi_{0}$, one finds
\begin{align*}
&\partial_{t}\psi_{0} = \frac{1}{2 \sqrt{\rho_{t}}} \frac{\partial}{\partial t}\sqrt{\frac{\kappa}{2\pi}}e^{-\kappa/2 (x - m(t))^{2}} \\
&= \frac{\sqrt{\rho_{t}}\kappa(x-m(t)) \partial_{t}m(t)}{2} = \frac{\sqrt{\kappa}\,\partial_{t}m(t)}{2} \psi_{1}.
\end{align*}
Thus, $\partial_{t}\psi_{0}$ lies entirely in the first excited eigenspace, and gap domination holds with $c=1$. As $\widehat{H}_\perp \psi_{1} = \kappa\psi_{1}$ we have $\widehat{H}_\perp^{-1}\psi_{1} = \kappa^{-1}\psi_{1}$ so that
\begin{align*}
&\beta(t)
=
2\|\widehat H_\perp^{-1}\partial_t\psi_0\|_{L^2} 
\\
&= \sqrt{\kappa}| \partial_{t}m(t)|\|\widehat H_\perp^{-1} \psi_{1}\|_{L^2}\\
&= \frac{| \partial_{t}m(t)|}{\sqrt{\kappa} }.
\end{align*}
Since $\partial_t\log\rho_t = \kappa (x-m(t)) \partial_{t}m(t)$, one has
\begin{align*}
&\|\partial_t\log\rho_t \|_{L^{2}(\rho_{t})} \\
&= \kappa|\partial_{t}m(t)| \left(\int (x-m(t))^{2} d\rho_{t}\right)^{1/2} \\
&= {\kappa |\partial_{t}m(t)|}/{\sqrt{\kappa}} = \sqrt{\kappa}|\partial_{t}m(t)|,
\end{align*}
so $|\partial_{t}m(t)| = \|\partial_t\log\rho_t \|_{L^{2}(\rho_{t})}/\sqrt{\kappa}$ and
\begin{align*}
&\beta(t)=\frac{\|\partial_t\log\rho_t\|_{L^2(\rho_t)}}{\kappa} = \frac{\|\partial_t\log\rho_t\|_{L^2(\rho_t)}}{\Delta_{t}}.
\end{align*}
Thus, the spectral upper and lower coefficient bounds exactly coincide, so the spectral estimate is saturated at leading order.
\end{example}

\begin{remark}[ Witten Laplacian]\label{rem:witt_lap}
The $0$-form Witten Laplacian \cite{Witten1982, Helffer2005} is constructed from a pre-defined potential $\phi$ via $\widehat{H}_{W}=-\nabla^{2}+ |\nabla \phi |^{2}- \nabla^{2} \phi$. It is the classic mapping of Langevin dynamics on a potential landscape $\phi$ into a unitarily equivalent \cite{setti} Schr{\"o}dinger operator $\widehat{H}_{W}$ that embeds the amplitude of the Gibbs density $\rho_{t}$ as its ground-state $\widehat{H}_{W}\sqrt{\rho_{t}}=0$. 
\end{remark}

\begin{remark}[Nodal Stability of the Information and Score Hamiltonians]\label{rem:fisher_v_score}
We remark on the distinction between the Information and Score Hamiltonian in Definition~\ref{def:ScoreHamiltonian}. As $f=-\frac{1}{2}\ln \rho$ is a mapping from a density to the thermodynamic potential, at nodes $\rho \downarrow 0$ this diverges. For instance, the density $\rho=|\psi|^{2}$ near a simple node crossing zero is locally parabolic $\sim cx^{2}$ -- for $(1/4)|S|^{2}+(1/2)\nabla\cdot S \propto -x^{-2}+x^{-2}$ this involves the difference of two divergent terms $+\infty-\infty$, whereas $\tilde{\psi} = cx$ so that $\nabla^{2}\tilde{\psi}/\tilde{\psi} = \frac{0}{cx}$ which is $\equiv 0$ away from zero and admits a removable extension $0$ at the node.
This illustrates the potential remains well-defined and finite at a removable singularity -- even when the score-representation involves cancellation of divergent terms.
\end{remark}

The unitary equivalence of the Score-Hamiltonian and the Langevin generator follows the expansion of the density in the Bohm potential in terms of the classical potential $\phi$ to yield the Witten Laplacian \cite{Witten1982, setti}. We extend this to the generalized Information-Hamiltonian (the amplitude-form of Definition~\ref{def:ScoreHamiltonian}), which is defined through the quantum potential on a general (signed) amplitude $\psi \in L^{2}$.

\begin{theorem}[Unitary Equivalence of the Information-Hamiltonian to the Langevin Generator on Amplitude-form Score]\label{thrm:BH_LG_equiv_amplitude}
Let $\tilde{\psi}$ be a real-valued, sufficiently regular amplitude with $\rho = |\tilde{\psi}|^{2}$ a probability density, and assume its nodal set has Lebesgue measure zero. Suppose the operator
    \[\widehat{H} = -\nabla^{2} + \frac{\nabla^{2}\tilde{\psi}}{\tilde{\psi}}\]
admits a self-adjoint realization on $L^{2}(dx)$, with nodal singularities interpreted through the assumed extension of $\frac{\nabla^{2}\tilde{\psi}}{\tilde{\psi}}$.

Define the multiplication operator $\hat{U}: L^{2} (\rho) \to L^{2}(dx)$ by $\hat{U} f = f \tilde{\psi}$, and define $\mathcal{L} = -\hat{U}^{-1}\,\widehat{H}\, \hat{U}$ where the operator domain satisfies $D(\mathcal{L}) = \hat{U}^{-1}D(\widehat{H})$. Then, $\hat{U}$ is unitary with $\widehat{H} = \hat{U}(-\mathcal{L}) \hat{U}^{-1}$.

Moreover, for smooth test functions $f$ supported away from the nodal set, $\mathcal{L}$ has differential representation $\mathcal{L}f = \nabla^{2}f + S_{*}\cdot \nabla f$
for the amplitude-form score $S_{*}$
\begin{equation}\label{eq:score_amplitude}
S_{*} = 2\,\frac{\nabla \tilde{\psi}}{\tilde{\psi}}.
\end{equation}
When $\tilde{\psi}$ is nodeless and has fixed sign, 
\[
S_{*} = \nabla \rho/\rho = \nabla \log \rho,
\]
and $\mathcal{L}$ is the standard Langevin generator. For a signed nodal amplitude, the Langevin representation holds on each nodal domain, while the global operator is defined by the above unitary conjugation.
\end{theorem}

\begin{proof}
Define the transform $\hat{U}: L^{2} (\rho) \to L^{2}(dx)$ by multiplication with the target wavefunction $\tilde{\psi}$
    \[
    \hat{U} f = f \tilde{\psi} 
    \]
    $\hat{U}$ is unitary, as before, since
    \begin{align*}
    \langle \hat{U} f, \hat{U}g \rangle_{L^{2}(dx)} 
    = \int (f\tilde{\psi})^{\dagger}(g\tilde{\psi}) dx 
    =\langle  f, g \rangle_{L^{2}(\rho)}.
    \end{align*}
    For smooth test functions $f$ supported away from the nodal set, without invoking the score, we have
    \begin{align*}
        \widehat{H}(\hat{U} f) &= \widehat{H}(f \tilde{\psi}) =\left( -\nabla^{2} + \frac{\nabla^{2}\tilde{\psi}}{\tilde{\psi}} \right) f \tilde{\psi} \\
        &= -\nabla^{2}( f \tilde{\psi}) + \frac{\nabla^{2}\tilde{\psi}}{\tilde{\psi}}( f \tilde{\psi}) \\
        &=-\nabla^{2}( f \tilde{\psi}) + \nabla^{2}\tilde{\psi}f  
    \end{align*}
    Via product rule for the Laplacian, we see
    \begin{align*}
        \widehat{H}(\hat{U} f) &=-\nabla^{2}( f \tilde{\psi})  + \nabla^{2}\tilde{\psi}f  \\
        &= - (f \nabla^{2} \tilde{\psi} + \tilde{\psi} \nabla^{2}f + 2\nabla f \cdot \nabla \tilde{\psi}) + \nabla^{2}\tilde{\psi}f  \\
        &= - \tilde{\psi} \,\nabla^{2}f - 2\nabla f \cdot \nabla \tilde{\psi}\\
        &=-\left(   \,\nabla^{2}f + \nabla f \cdot 2\nabla \tilde{\psi} /\tilde{\psi}\right)\tilde{\psi} \\
        &= - (\mathcal{L}f) \tilde{\psi} = - \hat{U}(\mathcal{L}f)
    \end{align*}
    Thus, away from the nodal set, $\mathcal{L}$ has the Langevin differential expression with an ``amplitude-form'' score
    \[
    S_{*} = 2\frac{\nabla \tilde{\psi}}{\tilde{\psi}}
    \]
    When $\tilde{\psi}$ is nodeless and has fixed sign, $\rho = \tilde{\psi}^{2}$ while $\frac{2\nabla \tilde{\psi}}{\tilde{\psi}}=\frac{\nabla (\tilde{\psi}^{2})}{\tilde{\psi}^{2}}$. Hence
    \begin{align*}
        2\frac{\nabla \tilde{\psi}}{\tilde{\psi}} = \frac{\nabla(\tilde{\psi}^{2})}{\tilde{\psi}^{2}}  = \frac{\nabla\rho}{\rho} = \nabla \log \rho =S
    \end{align*}
    recovering the standard score for both positive and negative branches of $\tilde{\psi}$. The global operator identity then follows from $\mathcal{L}=-\hat{U}^{-1}\widehat{H}\hat{U}$ with $D(\mathcal{L}) = \hat{U}^{-1}D(\widehat{H})$.
\end{proof}

Thm.~\ref{thrm:BH_LG_equiv_amplitude} avoids division by $\rho$ or evaluating $\nabla \log \rho$ in establishing an operator-level unitary conjugation, isolating $S_{*}$ only in the final differential representation. Conversely, starting from the classical expansion of the Witten Laplacian or Score Hamiltonian yields individually divergent terms in its intermediate steps at nodes where $\rho \downarrow 0$. While the amplitude-form conjugation remains well-defined under the assumptions of the theorem, the Langevin drift coefficient $S_{*}$ generally diverges upon approaching a node, and the corresponding differential representation is only valid away from the nodal set. Thus, the singularity reflects a breakdown of the classical Langevin SDE description at a node, rather than a failure of operator-level unitary equivalence.

\begin{remark}[Spectral Gap and Poincare Constant]\label{rem:PC}
We remark that the spectral gap of $\widehat{H}$ corresponds to the optimal Poincare constant, and refer the reader to a more thorough discussion in \cite{Helffer2005}. Let $p=2$, and recall that a distribution $\rho$ satisfies the Poincare inequality if there exists $C$ so that for all $u \in W^{1,2}(\Omega)$
\[
\lVert u - \bar{u} \rVert_{L^{2}(\rho)}^{2} \leq C\lVert \nabla u \rVert_{L^{2}(\rho)}^{2}
\]
with $\bar{u} = \int_{\Omega} u\,d\rho$. The smallest such $C$ is the Poincare constant. Now, recall that
\[
\Delta := \inf_{\psi: \,\langle \psi, \psi_{0}\rangle=0} \,\,\frac{\langle \psi | \widehat{H}|\psi \rangle}{ \lVert \psi \rVert_{L^{2}(dx)}^{2}}
\]
As noted above, we have by unitary equivalence that $\hat{U}f = f \sqrt{\rho}=\psi$. The condition $\langle \psi, \psi_{0}\rangle=0$ implies that $\int f\rho(dx) =0$ so that $f$ is zero-mean. Additionally, observe
\begin{align*}
    &\lVert \psi \rVert_{L^{2}(dx)}^{2} = \int f^{2} \rho(dx), \\
    &\langle \psi | \widehat{H} | \psi \rangle_{L^{2}(dx)} = \langle \psi | \hat{U}(-\mathcal{L})\hat{U}^{-1} | \psi \rangle_{L^{2}(dx)} \\
    &= \langle f | (-\mathcal{L})|f \rangle_{L^{2}(\rho)} = \int |\nabla f|^{2}\,\rho(dx)
\end{align*}
so that for $f=u -\bar{u},\,\nabla f = \nabla u$ defining a zero-mean function, we have
\begin{align*}
\Delta &= \inf_{f: \,\mathbb{E}_{\rho}f=0} \,\,\frac{\int |\nabla f|^{2}\,\rho(dx)}{\int f^{2} \rho(dx)} \\
&\leq \frac{\int |\nabla f|^{2}\,\rho(dx)}{\int f^{2} \rho(dx)} 
\end{align*}
and thus,
\begin{align*}
 \lVert u -\bar{u} \rVert_{L^{2}(\rho)}^{2}&=\lVert f \rVert_{L^{2}(\rho)}^{2}=\int f^{2} \rho(dx) \\
&\leq \frac{1}{\Delta}\int |\nabla f|^{2}\,\rho(dx) = \frac{1}{\Delta} \lVert \nabla f \rVert_{L^{2}(\rho)}^{2} \\
&= \frac{1}{\Delta} \lVert \nabla u \rVert_{L^{2}(\rho)}^{2}.
\end{align*}
So that $\Delta^{-1}$ necessarily defines the optimal Poincare constant.
\end{remark}

\begin{remark}[Generative Modeling as a Variational Optimization]
An optimization-based perspective is via the Rayleigh quotient. For any test state $\psi$, the expected energy is $E[\psi] = \langle \psi | \widehat{H} | \psi \rangle$. By expanding the score-form of the potential $\hat{V} = \frac{1}{2}\nabla \cdot S + \frac{1}{4}|S|^2$ and integrating by parts, we see $\int \psi^2 \nabla \cdot S\, dx = -2 \int \psi S \cdot \nabla \psi\, dx$ and that the energy is a sum-of-squares (the relative Fisher information)
\begin{align}
&\langle \psi | \widehat{H} | \psi \rangle \\
&= \int \left( \|\nabla \psi\|^2 - S \cdot (\psi \nabla \psi) + \frac{1}{4}\|S\|^2 \psi^2 \right) dx \nonumber \\
&= \int \lVert \nabla \psi - \frac{1}{2} S(x) \psi \rVert^2 dx \ge 0 .\label{eq:rayleigh_squares}\end{align}
Thus, $\widehat{H}$ is a positive semi-definite operator with a global minimum of the energy landscape strictly at $0$. This is achieved if and only if $\nabla \psi = \frac{1}{2} S \psi = (\frac{1}{2} \nabla \log \rho) \psi$. Enforcing the Lagrange multiplier $\langle \psi| \psi \rangle = 1$ implies the unique normalized solution to the Rayleigh quotient problem
\begin{equation}
\inf_{\psi \in L^{2} \setminus \{ 0 \}} \frac{\langle \psi | \widehat{H} | \psi \rangle}{\langle \psi | \psi \rangle}
\end{equation}
is $\psi_0 = \sqrt{\rho}$. This defines a variational problem with the target $\sqrt{\rho}$ as its unique minimum: given that the Score Hamiltonian is self-adjoint, generative modeling is \emph{variationally convex} when $\psi$ lacks a normalization constraint, controlled by the condition number (spectral gap) of $\widehat{H}$. The normalization constraint $\lVert \psi \rVert=1$ breaks convexity, requiring minimization over the sphere. Thus the unnormalized problem is convex in the amplitude $\psi$, in contrast to non-convex density-fitting objectives which must enforce normalization.
\end{remark}

\begin{lemma}[Ground-State Perturbation in Score Error] \label{lem:static_perturbation} 
Let $\widehat{H}$ and $\widehat{H}_{\theta}$ be the Score-Hamiltonians for the target density $\rho$ and the learned density $\rho_{\theta}$, respectively. 

Assume both are essentially self-adjoint on $L^{2}(\mathbb{R}^{d})$ and that $\widehat{H},\widehat{H}_{\theta}$ have strictly positive spectral gap $\Delta, \Delta_{\theta} > 0$. Let $\psi_{0} = \sqrt{\rho}$ and $\psi_{0}^{\theta} = \sqrt{\rho_{\theta}}$ be their normalized ground states. Then, the $L^2(dx)$ distance between the ground states is strictly bounded by:
\begin{align*}
&\lVert \psi_0^\theta - \psi_0 \rVert_{L^2(dx)}^2 \\
&\leq\min \left\{ \frac{1}{2 \Delta}\mathbb{E}_{\rho_{\theta}}\lVert \stheta -  S \rVert_{2}^{2},
\,\, \frac{1}{2 \Delta_{\theta}}\mathbb{E}_{\rho}\lVert \stheta -  S \rVert_{2}^{2}
\right\}.
\end{align*}
\end{lemma}

\begin{proof}[Proof of~\ref{lem:static_perturbation}]
We notate $\psi:= \psi_{0}$ and $\psi_{\theta}:= \psi_{0}^{\theta}$ as the ground-state amplitudes. Because $\psi_0$ is the zero-energy ground state of $\widehat{H}$ with spectral gap $\Delta$, the spectral theorem guarantees $\langle \psi_0^\theta | \widehat{H} | \psi_0^\theta \rangle \ge \Delta (1 - c^2)$ and $\langle \psi_0 | \widehat{H}_{\theta} | \psi_0 \rangle \ge \Delta_{\theta} (1 - c^2)$. In particular, we have
\begin{align*}
\langle \psi| \widehat{H}_{\theta}| \psi \rangle &= \int \psi \left( - \nabla^{2}\psi+
\frac{1}{2} (\nabla \cdot \stheta)\psi + \frac{1}{4} |\stheta|^{2}\psi
\right) \\
&=\int \left( | \nabla \psi |^{2} -  \stheta \cdot (\psi \nabla \psi) + \frac{1}{4} |\stheta|^{2}\psi^{2} \right)\,dx \\
&=\int \lVert \nabla \psi - \frac{1}{2} \stheta\,\psi \rVert^{2}\,dx 
\end{align*}
and in reverse,
\[
\langle \psi_{\theta}| \widehat{H}| \psi_{\theta} \rangle = \int \lVert \nabla \psi_{\theta} - \frac{1}{2} S\,\psi_{\theta} \rVert^{2}\,dx
\]
Where $\nabla \psi /\psi = \nabla \sqrt{\rho}/\sqrt{\rho} = \frac{1}{2} S$ and $\nabla \psi_{\theta} /\psi_{\theta}= \frac{1}{2} \stheta$ and thus
\begin{align*}
\langle \psi_{\theta}| \widehat{H}| \psi_{\theta} \rangle &= \frac{1}{4}\int \lVert \stheta -  S\, \rVert^{2} \rho_{\theta}\,dx = \frac{1}{4}\mathbb{E}_{\rho_{\theta}}\lVert \stheta -  S \rVert_{2}^{2}, \\
\langle \psi| \widehat{H}_{\theta}| \psi \rangle &= \frac{1}{4}\int \lVert \stheta -  S\, \rVert^{2} \rho\,dx = \frac{1}{4}\mathbb{E}_{\rho}\lVert \stheta -  S \rVert_{2}^{2}
\end{align*}
Since we have that $c=\langle \psi, \psi_{\theta}\rangle$ is shared and applies symmetrically to the bound on $\langle \psi_{\theta}| \widehat{H}| \psi_{\theta} \rangle$ and $\langle \psi| \widehat{H}_{\theta}| \psi \rangle$,
\begin{align}\label{eq:min}
&\frac{1}{4}\mathbb{E}_{\rho_{\theta}}\lVert \stheta -  S \rVert_{2}^{2}=\langle \psi_{\theta}| \widehat{H}| \psi_{\theta} \rangle \geq \Delta (1-c^{2}), \\
& \label{eq:min2} \frac{1}{4}\mathbb{E}_{\rho}\lVert \stheta -  S \rVert_{2}^{2}=\langle \psi| \widehat{H}_{\theta}| \psi \rangle \geq \Delta_{\theta} (1-c^{2})
\end{align}
Thus, combining this with the geometric bound $2(1-c) \le 2(1-c^2)$, the $L^2(dx)$ distance is bounded via 
\begin{align*} 
\lVert \psi_0^\theta - \psi_0 \rVert_{L^2(dx)}^2 = 2(1 - c) \le 2(1 - c^2)  
\end{align*}
Then, taking the minimum of ~\eqref{eq:min} and~\eqref{eq:min2}, we conclude.
\end{proof}

\begin{proof}[Proof of Thm.~\ref{thrm:ad_diffusion}]
Let $\Psi_{\tau_T} = \mu_{\tau_T}/\sqrt{\rho_{\theta,T}}$, for $\mu_{\tau_T}$ the diffusion model density of the Fokker-Planck equation, and $\sqrt{\rho_{\theta,T}}$ the ground-state of the model Hamiltonian $\widehat{H}_{\theta}(T)$. Applying triangle inequality, 
\begin{align}
    d_{\mathrm{TV}}(\mu_{\tau_{T}}, \rho_{T}) \leq  d_{\mathrm{TV}}(\mu_{\tau_{T}}, \rho_{\theta,T}) + d_{\mathrm{TV}}(\rho_{T}, \rho_{\theta,T})
\end{align}
Thm.~\ref{thrm:adiabatic_FK} controls the following term
\begin{align}
\label{eq:TV-Hellinger} d_{\mathrm{TV}}(\mu_{\tau_{T}}, \rho_{\theta,T}) &\leq \frac{1}{2} \sqrt{\chi^2(\mu_{\tau_T} \| \rho_{\theta,T})}\\
\notag&= \frac{1}{2} \| \Psi_{\tau_T} - \sqrt{\rho_{\theta,T}} \|_{L^2}    
\end{align}
Meanwhile, as $d_{\mathrm{TV}}(P,Q) \leq \sqrt{2} d_{H}(P,Q)\leq\|\sqrt{P}-\sqrt{Q}\|_{L^2}$,
\[
d_{\mathrm{TV}}(\rho_{T}, \rho_{\theta,T}) \leq \|\sqrt{\rho_{T}}-\sqrt{\rho_{\theta,T}}\|_{L^2}
\]
is controlled by the Perturbation Lemma~\ref{lem:static_perturbation}, concluding the proof.
\end{proof}

\subsection{Non-Conservative Extension}\label{sec:noncons_exten}

If $\stheta$ is non-conservative, it may be decomposed via the Helmholtz-Hodge Thm. into a curl-free conservative component $\nabla \phi^{\theta}$ and divergence-free component $R^{\theta}$
\[
\stheta = \nabla \phi^{\theta} + R^{\theta},\qquad \nabla \cdot R^{\theta} = 0.
\]
The Langevin generator with non-conservative $\stheta$ applied to a test function $f$ gives
\[
L^{\theta}f = \Delta f + \stheta \cdot \nabla f = \Delta f + \nabla \phi^{\theta} \cdot \nabla f + R^{\theta} \cdot \nabla f
\]
as the field exhibits a divergence-free current $R^{\theta} \cdot \nabla f$, unitary equivalence ceases to hold. We remark, however, that using a non-conservative score in an adiabatic tracking on the score Hamiltonian $\widehat{H}_{\theta}$ is comparatively well-behaved. The divergence-term of the potential $\hat{V}_{\theta}$ isolates only the induced conservative score of $\stheta$, leaving
\begin{align*}
    \hat{V}_{\theta} &= \frac{1}{2} \nabla \cdot \stheta + \frac{1}{4} |\stheta|^{2} = \frac{1}{2} \nabla \cdot \nabla \phi^{\theta} + \frac{1}{4} | \nabla \phi^{\theta} + R^{\theta} |^{2} \\
    &= \left( \frac{1}{2} \nabla \cdot \nabla \phi^{\theta} + \frac{1}{4} | \nabla \phi^{\theta}|^{2} \right) + \left( \frac{1}{2} \nabla \phi^{\theta} \cdot R^{\theta} + \frac{1}{4}|R^{\theta} |^{2} \right).
\end{align*}
We thus find that the effect of non-conservativity is entirely isolated to the score magnitude $|\stheta|^{2}$. As the self-adjointness conditions for the Score-Hamiltonian do not strictly require $\stheta$ be conservative, $\widehat{H}_{\theta}$ exhibits a well-defined sequence of ground-states $\sqrt{\tilde{\rho}_{\theta,t}}$ defining proper normalized densities $\tilde{\rho}_{\theta,t} = |\psi_{\theta,t}|^{2}$ with $\int |\psi_{\theta,t}|^{2} = 1$. As a consequence, the tracked ground-state family becomes $\tilde{\rho}_{\theta,t}$ with rotational corrections to the tracking error detailed below.

\begin{figure*}[tbp]
    \centering
    \includegraphics[width=0.8\linewidth]{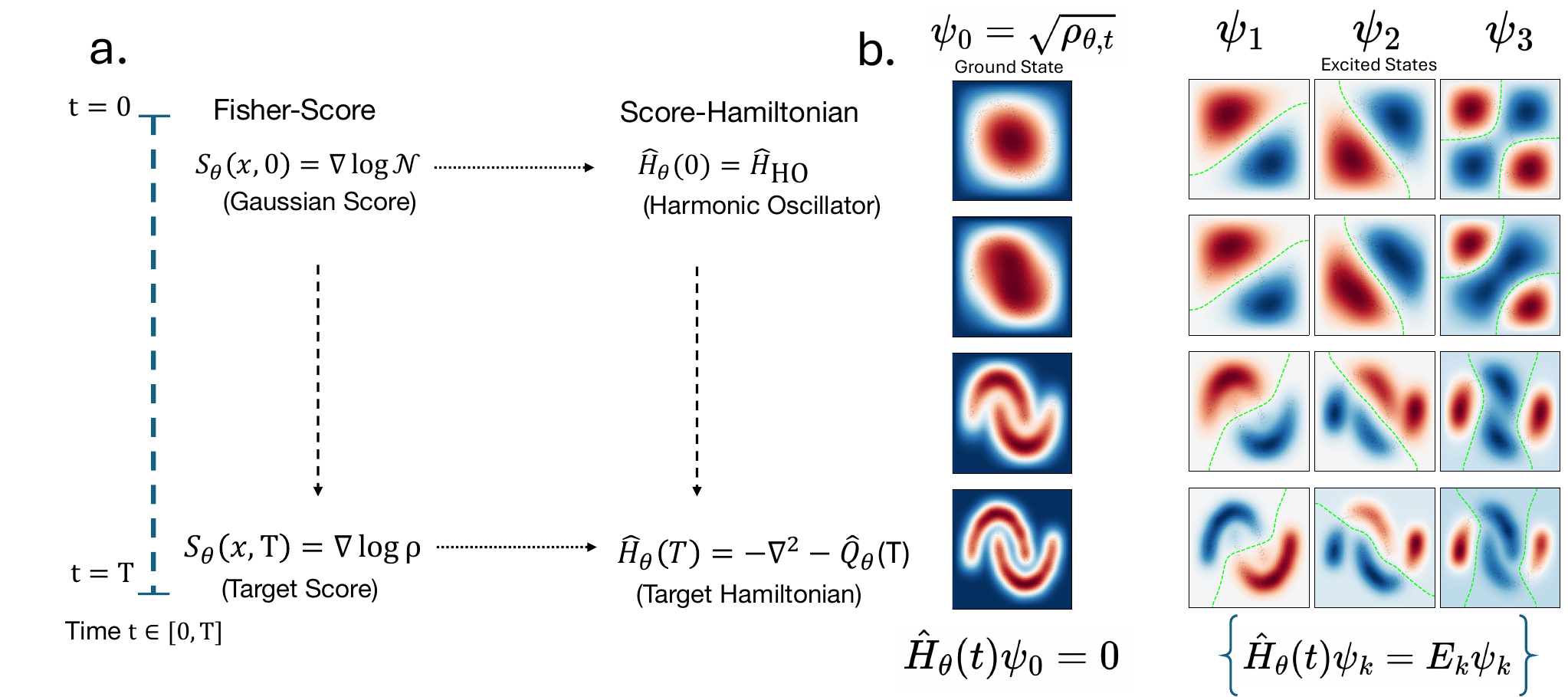}
    \caption{\textbf{(a.)} Illustration of the adiabatic transport of a diffusion-model across varying $t$. \textbf{(b.)} $\widehat{H}_{\theta}(t)$ encodes the model density $\rho_{\theta,t}$ as the ground state of the time $t$ diffusion model, and the excited states $\psi_{k\geq 1}$ encode the spectrum of $\widehat{H}_{\theta}(t)$ and the associated Langevin generator ${L}_{\theta}(t)$ of the score $\stheta$.
    }
    \label{fig:Fig.1}
\end{figure*}

\begin{proposition}[Non-Conservativity of the Score]\label{prop:non_conservative_score}
Suppose the learned score vector field $\stheta$ is non-conservative and assume Prop.~\ref{prop:score_hamiltonian_props}(\ref{cond:pure_pt}) so $\hat{H}_{\theta}$ has a discrete ground-state. Then, the ground-state eigenvalue $\lambda_\theta$ quantifies the non-conservative projection error of the score:
\begin{equation}\label{eq:lambda_theta}
    \lambda_{\theta} = \frac{1}{4}\mathbb{E}_{\tilde{\rho}_{\theta}} \lVert \nabla \log \tilde{\rho}_{\theta} - \stheta \rVert^{2}
\end{equation}
The distance between the generated and true ground states is bounded by:
\[
\lVert \psi_{\theta} - \psi_{0} \rVert_{L^{2}(dx)} \leq \epsilon_{\theta, \mathrm{nc}} + \sqrt{\frac{2 \lambda_{\theta}}{\Delta }}
\]
for $\epsilon_{\theta, \mathrm{nc}}$ the upper-bound in~\ref{lem:static_perturbation},
\[
\epsilon_{\theta, \mathrm{nc}}=\min\left\{
        \frac{\epsilon_{\theta,{\rm score}}(T)}{\sqrt{\Delta_T}},
        \frac{\epsilon_{\rm score}(T)}{\sqrt{\Delta_{\theta,T}}}
    \right\}
\]
\end{proposition}

\begin{proof}
Observe $\nabla \log \tilde{\rho}_{\theta,t} = \frac{2\nabla \psi_{\theta,t}}{\psi_{\theta,t}}$ and that
\begin{align*}
\lambda_{\theta,t}&=\langle \widehat{H}_{\theta}\, \psi_{\theta,t}, \psi_{\theta,t} \rangle \\
&= \int \psi_{\theta,t} \left( - \nabla^{2}\psi_{\theta,t}+
\frac{1}{2} (\nabla \cdot \stheta)\psi_{\theta,t} + \frac{1}{4} |\stheta|^{2}\psi_{\theta,t}
\right) \\
&=\int \left( | \nabla \psi_{\theta,t} |^{2} -  \stheta \cdot (\psi_{\theta,t} \nabla \psi_{\theta,t}) + \frac{1}{4} |\stheta|^{2}\psi_{\theta,t}^{2} \right)\,dx \\
&=\int \lVert \nabla \psi_{\theta,t} - \frac{1}{2} \stheta\,\psi_{\theta,t} \rVert^{2}\,dx 
\end{align*}
Since $\nabla \log \tilde{\rho}_{\theta,t} = \frac{2\nabla \psi_{\theta,t}}{\psi_{\theta,t}}$ $\nabla\psi_{\theta,t} = \frac{1}{2}\nabla \log \tilde{\rho}_{\theta,t}\psi_{\theta,t}$, we identify this as the expected error between $\stheta$ and its induced conservative projection
\begin{align}
\label{eq:lambda_bound}
\lambda_{\theta,t}&=\frac{1}{4}\int \lVert \nabla \log \tilde{\rho}_{\theta,t} -  \stheta\, \rVert^{2} \tilde{\rho}_{\theta,t}\,dx \\
&= \frac{1}{4}\mathbb{E}_{\tilde{\rho}_{\theta,t}}\lVert \nabla \log \tilde{\rho}_{\theta,t} -  \stheta\, \rVert^{2}
\end{align}
Now, at the terminal time $t=T$ denotes $\psi_{\theta,T}:=\psi_{\theta}$ and the target density ground state as $\psi_{0}$. We see that the energy under the true score Hamiltonian (with exact score $S$) is given by
$\langle \widehat{H}\, \psi_{\theta}, \psi_{\theta} \rangle = \frac{1}{4}\mathbb{E}_{\rho_{\theta}}\lVert \nabla \log \tilde{\rho}_{\theta} -  S\, \rVert^{2}$
Now, applying the perturbation Lemma~\ref{lem:static_perturbation}, we see
\begin{align}\label{spec_bnd}
\lVert \psi_{\theta} - \psi_{0} \rVert_{L^{2}} &\leq \sqrt{\frac{2\langle \widehat{H}\, \psi_{\theta}, \psi_{\theta} \rangle }{\Delta }} \\
&= \frac{\lVert \nabla \log \tilde{\rho}_{\theta} -  S\, \rVert_{L^{2}(\tilde{\rho}_{\theta})}}{\sqrt{2 \Delta} } 
\end{align}
For the induced conservative score $\nabla \log \tilde{\rho}_{\theta}$. However, as $\nabla \log \tilde{\rho}_{\theta}$ itself is not accessible, we apply triangle-inequality on the non-conservative $\stheta$ to find
\begin{align}\label{eq:TI}
&\lVert \nabla \log \tilde{\rho}_{\theta} -  S \rVert_{L^{2}(\tilde{\rho}_{\theta})} \\
&\qquad \leq \lVert S -  \stheta \rVert_{L^{2}(\tilde{\rho}_{\theta})} + \lVert \nabla \log \tilde{\rho}_{\theta}  -  \stheta \rVert_{L^{2}(\tilde{\rho}_{\theta})}
\end{align}
By definition of the score matching error, $\frac{1}{\sqrt{2}}\lVert S - \stheta \rVert_{L^{2}(\tilde{\rho}_{\theta})} = \epsilon_{\theta,\mathrm{score}}$, which contributes the bounding term $\frac{\epsilon_{\theta,\mathrm{score}}}{\sqrt{\Delta}} = \epsilon_{\theta, \mathrm{nc}}$ via Lemma~\ref{lem:static_perturbation}.

From equation~\eqref{eq:lambda_bound}, the non-conservative projection error is exactly $\lVert \nabla \log \tilde{\rho}_{\theta} - \stheta \rVert_{L^{2}(\tilde{\rho}_{\theta})} = 2\sqrt{\lambda_{\theta}}$. Substituting these into~\eqref{eq:TI} and returning to \eqref{spec_bnd}, we find:
$$
\begin{aligned}
\lVert \psi_{\theta} - \psi_{0} \rVert_{L^{2}} &\leq \frac{1}{\sqrt{2\Delta}} \left( \sqrt{2}\,\epsilon_{\theta,\mathrm{score}} + 2\sqrt{\lambda_{\theta}} \right) \\
&= \frac{\epsilon_{\theta,\mathrm{score}}}{\sqrt{\Delta}} + \sqrt{\frac{2 \lambda_{\theta}}{\Delta}}.
\end{aligned}
$$
Via a symmetric argument that evaluates the energy of the true ground-state $\psi_{0}$ under the model $\widehat{H}_{\theta}$, one finds the parallel bound $\lVert \psi_{\theta} - \psi_{0} \rVert_{L^{2}} \leq \frac{\epsilon_{\mathrm{score}}}{\sqrt{\Delta_{\theta}}} + \sqrt{\frac{2 \lambda_{\theta}}{\Delta_{\theta}}}$. By taking the minimum of the two over the conservative error yields $\epsilon_{\theta, \mathrm{nc}}$, which concludes the proof.
\end{proof}

Thus, the standard static floor holds with an additive penalty proportional scaling as the square root of the non-conservativity defect.

\begin{remark}[Langevin Generator on $\rhot$]\label{rem:NESS}
The action of the Langevin Generator $\mathcal{L}^{*}$ on $\rhot$ is given by
\begin{align*}
\mathcal{L}^{*}\rhot 
&=-(\nabla\cdot A^{\theta} + \nabla \log \rhot \cdot A^{\theta})\rhot
\\
&=- \rhot\left(2E_{0} - \frac{1}{2}|A^{\theta}|^{2} \right)
\end{align*}
Thus $\tilde{\rho}_{t}$ is the non-equilibrium steady-state of Langevin evolution on the non-conservative $\stheta$ if the difference to the induced conservative score $A^{\theta}=S^{\theta}-\nabla \log \rhot$ is divergence-free $\nabla \cdot A^{\theta} = 0$ (solenoidal) and if the rotational flux acts orthogonally to the isoclines of the induced conservative density $\nabla \log \rhot \cdot A^{\theta} = 0$. Thus, while detailed balance is broken, this $\tilde{\rho}_{t}$ can represent the non-equilibrium steady state (NESS) of the Fokker-Planck equation.
\end{remark}

\begin{proof}
Recall that locally, the stationary Schr{\"o}dinger equation $\hat{H}_{\theta}\psi_{0} = E_{0} \psi_{0}$ implies by~\eqref{eq:bohm_score_id} that
\begin{align*}
&-\nabla^{2}\psi_{0} + \left(
\frac{1}{2}\nabla\cdot \stheta + \frac{1}{4} \lVert \stheta \rVert^{2}
 \right)\psi_{0} = E_{0}\psi_{0} \\
 &-\frac{\nabla^{2}\psi_{0}}{\psi_{0}} + 
\frac{1}{2}\nabla\cdot \stheta + \frac{1}{4} \lVert \stheta \rVert^{2}
 = E_{0} \\
 &-\frac{1}{2}\nabla \cdot \nabla \log \tilde{\rho} - \frac{1}{4} \lVert \nabla \log \tilde{\rho} \rVert^{2} +\frac{1}{2}\nabla\cdot \stheta + \frac{1}{4} \lVert \stheta \rVert^{2} = E_{0} 
\end{align*}
Defining $A^{\theta}=\stheta-\nabla \log \rhot$ yields
\begin{align*}
&\frac{1}{2}\nabla\cdot (\stheta- \nabla \log \tilde{\rho}) - \frac{1}{4} \lVert \nabla \log \tilde{\rho} \rVert^{2}  + \frac{1}{4} \lVert \stheta \rVert^{2} = E_{0} \\
&\frac{1}{2}\nabla\cdot A^{\theta} - \frac{1}{4} \lVert \nabla \log \tilde{\rho} \rVert^{2}  + \frac{1}{4} \lVert \stheta \rVert^{2} = E_{0}\\
&\frac{1}{2}\nabla\cdot A^{\theta} +\frac{1}{4} \lVert A^{\theta} \rVert^{2} + \frac{1}{2} A^{\theta} \cdot \nabla \log \rhot = E_{0}
\end{align*}
Expanding the square and cancelling terms, this recovers
\begin{align*}
&- \mathcal{L}^{*} \rhot =- \nabla^{2}\rhot + 
\nabla\cdot( \stheta \rhot) \\
&=- \nabla^{2}\rhot + 
\nabla\cdot \stheta \rhot + \nabla \log \rhot \cdot \stheta \rhot  \\
&=- \nabla^{2}\rhot + 
\nabla\cdot (A^{\theta} + \nabla \log \rhot) \rhot + \nabla \log \rhot \cdot (A^{\theta} + \nabla \log \rhot) \rhot
\end{align*}
Since
\[\nabla \cdot \nabla \log \rhot = \frac{\nabla^{2}\rhot}{\rhot} - |\nabla \log \rhot |^{2}\]
We then find
\begin{align*}
&= \nabla\cdot A^{\theta}\rhot  -| \nabla \log \rhot|^{2} \rhot + \nabla \log \rhot \cdot A^{\theta}\rhot  + |\nabla \log \rhot|^{2} \rhot \\
&=\nabla\cdot A^{\theta}\rhot + \nabla \log \rhot \cdot A^{\theta}\rhot  = \rhot  \left(2E_{0} - \frac{1}{2}|A^{\theta}|^{2} \right)
\end{align*}
\end{proof}

We remark that the non-asymptotic upper-bound of Theorem~\ref{thrm:adiabatic_FK} requires a modification for when $\stheta$ has rotational contributions. In particular, one defines the moving ground-state instead to be the non-equilibrium steady-state given by the ground-state of the Score Hamiltonian $\psi_{0}(t) = \sqrt{\tilde{\rho}_{t}}$ and instead defines tracking in terms of this non-equilibrium steady-state.

\begin{proposition}\label{prop:adiabatic-bound_NC}
Suppose $\stheta_{t_\tau}$ is a vector-field such that,
for all $t$, $\widehat H_t$ is self-adjoint and the spectral gap $\Delta_t$ of
$\widehat H_t$ is strictly positive with ground-state $\sqrt{\tilde{\rho}_{t}}$ and induced conservative score $\nabla \log \tilde{\rho}_{t_\tau}$. Denote 
\begin{align*}
&M_{t,-} =\max\left\{0,-\tfrac{1}{2} \partial_t \log \tilde{\rho}_t(x)\right\}, \\
&N_{t,-} = \frac{1}{4}|\nabla \log \rhot_{t}-S^{\theta}_{t}|^{2}.
\end{align*}
Assume $\partial_t \log \tilde{\rho}_t \in L^2(\tilde{\rho}_t)$ and $\frac{1}{4}|\nabla \log \rhot_{t}-S^{\theta}_{t}|^{2} \in L^2(\tilde{\rho}_t)$, and that there exist $a(t),b(t)\geq 0$ and $\alpha(t),\beta(t)\geq 0$ so that they respectively satisfy the form-bounds
\begin{align*}
&\langle u,M_{t,-}u\rangle_{L^2}
    \leq
    a(t)\langle u,\widehat H_t u\rangle_{L^2}
    +
    b(t)\langle u,u\rangle_{L^2},\\
&\langle u,N_{t,-}u\rangle_{L^2}
    \leq
    \alpha(t)\langle u,\widehat H_t u\rangle_{L^2}
    +
    \beta(t)\langle u,u\rangle_{L^2} .
\end{align*}
for all $u$ orthogonal to the ground state $\sqrt{\tilde{\rho}_t}$. Assume in addition that the rotational coefficients satisfy $\alpha(t)<1$ and $\Delta_{t}> \frac{\beta(t)}{1-\alpha(t)}$. Then:
\begin{align*}
\notag \sqrt{
        \chi^2\!\left(
            \mu_\tau\,\middle\|\,\tilde{\rho}_{t_{\tau}}
        \right)
}&\leq    \sqrt{\chi^2(\mu_0\|\tilde{\rho}_0)}\exp\left(-\int_0^\tau \Gamma(s)ds\right)  \\
&\quad + \sup_{s \le \tau} \biggl(\,\, \frac{|\dot{t}(s)| \lVert \partial_{t} \log \rhot_{t(s)} \rVert_{L^2 (\rho_{t(s)} )}}{\Gamma(s)} \\
&\quad +  \frac{1}{2\Gamma(s)} \sqrt{\mathrm{Var}_{\tilde{\rho}_{t(s)}}|S^{\theta}_{t(s)} - \nabla \log \rhot_{t(s)}|^2} \biggr).
\end{align*}
where 
\begin{align*}
&\Gamma(\tau):=
    \left(1-|\dot t(\tau)|a(t_\tau)-
    \alpha(t_\tau)\right)(\Delta_{t_\tau}+ \lambda_{\theta}(t)) \\
    &- |\dot t(\tau)|b(t_\tau) - \beta(t_{\tau})
\end{align*}
is the effective spectral gap, and we take $\dot t=dt/d\tau$ sufficiently small that $\Gamma(\tau)$ is strictly positive
for all $\tau$.
\end{proposition}

\begin{proof}
Let us denote the Helmholtz decomposition of the score into its divergence-free and gradiential components as $S^{\theta} = \nabla \phi^{\theta} + R^{\theta}$ where $\nabla \cdot R^{\theta} =0$. The associated Score Hamiltonian, $\widehat{H}_t^\theta = -\nabla^2 + \frac{1}{2}\nabla \cdot S^\theta + \frac{1}{4}\|S^\theta\|^2,$ has unique ground-state defined by the density of the induced conservative score $\psi_{0}(t) = \sqrt{\tilde{\rho}_{t}}$ with ground-state energy $\lambda_{\theta}$ from \eqref{eq:lambda_theta} -- even in the case of non-conservative $\stheta$. Now, let us denote $\Psi_{\tau} = {\mu_{\tau}}/{\sqrt{\rhot_{t_{\tau}}}}$. The Fokker-Planck evolution on $\mu_\tau = \Psi \rhot$ implies
\begin{align*}
    \partial_{\tau}\mu = \partial_{\tau}(\Psi_\tau \psi_{0}(t_\tau )) = \partial_{\tau } \Psi \psi_{0} + \frac{\dot{t}}{2} \Psi \psi_{0} \partial_{t} \log \rhot_{t}
\end{align*}
Meanwhile, the action of the Langevin adjoint $\mathcal{L}^{*}$ implies
\begin{equation}\label{eq:FP_rhs}
\mathcal{L}^{*} \mu_{\tau} = \nabla^{2}\mu_{\tau} - 
\nabla\cdot( \stheta \mu_{\tau}) 
\end{equation}
where
\[
\nabla^{2}\mu_{\tau} = \nabla^{2}(\Psi_\tau \psi_{0}) = \psi_{0} \nabla^{2}\Psi_\tau + 2 \nabla \Psi_\tau \cdot \nabla \psi_{0} + \Psi_\tau \nabla^{2}\psi_{0}
\]
and 
\[
- \nabla \cdot (\Psi_\tau \psi_{0} \stheta) = -\Psi_{\tau} (\nabla \psi_{0} \cdot \stheta) - \psi_{0} \nabla \Psi_{\tau} \cdot \stheta -\Psi_\tau \psi_{0} (\nabla \cdot \stheta).
\]
Now, recall that $\mathcal{L}^{*}\psi_{0}^{2} =\mathcal{L}^{*} \rhot = -\rhot (2E_{0} - \frac{1}{2}|A^{\theta}|^{2})$ where $\stheta-\nabla\log \rhot = A^{\theta}$, and that
\begin{align*}
-\psi_{0} (E_{0} - \frac{1}{4}|A^{\theta}|^{2})&=\frac{1}{2\psi_{0}}\mathcal{L}^{*}\psi_{0}^{2} \\
&= \nabla^{2}\psi_{0} + \frac{|\nabla\psi_{0}|^{2}}{\psi_{0}} - \frac{\psi_{0}}{2} \nabla \cdot \stheta -  \stheta \cdot \nabla \psi_{0}
\end{align*}
Now, grouping terms and using the condition, we simplify the $\Psi_{\tau}$-leading terms,
\begin{align*}
&\Psi_{\tau}\left(
\nabla^{2}\psi_{0} - \nabla \psi_{0} \cdot \stheta - \psi_{0} \nabla \cdot \stheta
\right) \\
&=\Psi_{\tau}\left(
\frac{\mathcal{L}^{*}\psi_{0}^{2}}{2\psi_{0}} - \frac{|\nabla \psi_{0}|^{2}}{\psi_{0} } -\frac{\psi_{0}}{2} \nabla \cdot \stheta
\right) \\
&= \Psi_{\tau}\left(-\psi_{0} (E_{0} - \frac{1}{4}|A^{\theta}|^{2}) - \frac{|\nabla \psi_{0}|^{2}}{\psi_{0} } -\frac{\psi_{0}}{2} \nabla \cdot \stheta
\right) \\
&=\Psi_{\tau} \psi_{0}\left(-E_{0} +\frac{1}{4}|A^{\theta}|^{2} -\frac{1}{4} |\nabla \log \rhot|^{2} -\frac{1}{2} \nabla \cdot \stheta
\right) \\
&=\Psi_{\tau} \psi_{0}\biggl(
\left(-\frac{1}{4}|\stheta|^{2}-\frac{1}{2} \nabla \cdot \stheta \right)
\\
&+\left(\frac{1}{4}|\stheta|^{2}
-E_{0} +\frac{1}{4}|A^{\theta}|^{2} -\frac{1}{4} |\nabla \log \rhot|^{2} \right)
\biggr)
\end{align*}
Observe that
\begin{align*}
&\frac{1}{4}|\stheta|^{2}-E_{0} +\frac{1}{4}|A^{\theta}|^{2} -\frac{1}{4} |\nabla \log \rhot|^{2}  \\
&=\frac{1}{4}|\stheta|^{2}-\frac{1}{2}\nabla\cdot A^{\theta} -\frac{1}{4} | A^{\theta} |^{2} - \frac{1}{2} A^{\theta} \cdot \nabla \log \rhot \\
&+\frac{1}{4}|A^{\theta}|^{2} -\frac{1}{4} |\nabla \log \rhot|^{2}  \\
&=\frac{1}{4}|A^{\theta}|^{2} + \frac{1}{4}|\nabla \log \rhot|^{2} + \frac{1}{2} A^{\theta} \cdot \nabla \log \rhot
\\
&-\frac{1}{2}\nabla\cdot A^{\theta} - \frac{1}{2} A^{\theta} \cdot \nabla \log \rhot  -\frac{1}{4} |\nabla \log \rhot|^{2} \\
&=\frac{1}{4}|A^{\theta}|^{2} -\frac{1}{2}\nabla\cdot A^{\theta}
\end{align*}
Letting $\frac{1}{4}|A^{\theta}|^{2} -\frac{1}{2}\nabla\cdot A^{\theta}  = -W_{\theta}$ and adding back the remaining terms from the adjoint we see
\begin{align*}
&\mathcal{L}^{*}\mu_{\tau}=\Psi_{\tau} \psi_{0}\biggl(  -\frac{1}{2} \nabla \cdot \stheta
-\frac{1}{4} |\stheta|^{2} -W_{\theta}
\biggr) \\
&+\psi_{0}  \nabla^{2}\Psi_{\tau} + 2 \nabla \Psi_{\tau} \cdot \nabla \psi_{0} - \psi_{0} \nabla \Psi_{\tau} \cdot \stheta\\
&=\psi_{0} \biggl( 
\left(\nabla^{2} -\frac{1}{2} \nabla \cdot \stheta
-\frac{1}{4} |\stheta|^{2} \right)\Psi_{\tau}\\
&- 
W_{\theta}\Psi_{\tau} - \nabla \Psi_{\tau} \cdot \stheta \biggr) \\
&+ 2 \nabla \Psi_{\tau} \cdot \nabla \psi_{0} \\
&=\psi_{0}\biggl( 
-\widehat{H}_t^\theta \Psi_{\tau}
 -W_{\theta}\Psi_{\tau} - \nabla \Psi_{\tau} \cdot \stheta \biggr) + 2 \nabla \Psi_{\tau} \cdot \nabla \psi_{0}
\end{align*}
Equating both sides,
\begin{align*}
&\partial_{\tau}\mu = \partial_{\tau}(\Psi_\tau \psi_{0}(t_\tau )) \\
&= \partial_{\tau } \Psi_{\tau} \psi_{0} + \frac{\dot{t}}{2} \Psi_{\tau} \psi_{0} \partial_{t} \log \rhot_{t} = \mathcal{L}^{*}\mu_{\tau} 
\end{align*}
and then dividing both by $\psi_{0}$ and rearranging, we find
\begin{align*}
\partial_{\tau } \Psi_{\tau} &=-\widehat{H}_t^\theta \Psi_{\tau} - \frac{\dot{t}}{2} \Psi \partial_{t} \log \rhot_{t} \\
&- 
W_{\theta}\Psi_{\tau}  -  \nabla \Psi_{\tau} \cdot A^{\theta}
\end{align*}
Using $\frac{1}{4}|A^{\theta}|^{2} -\frac{1}{2}\nabla\cdot A^{\theta} = -W_{\theta}$, we see
\begin{align*}
&- 
W_{\theta}\Psi_{\tau}  -  \nabla \Psi_{\tau} \cdot A^{\theta}\\
&\quad =\frac{1}{4}|A^{\theta}|^{2}\Psi_{\tau} -\frac{1}{2}\nabla\cdot A^{\theta}\,\Psi_{\tau}  -  A^{\theta}\cdot \nabla \Psi_{\tau} \\
&\quad =\frac{1}{4}|A^{\theta}|^{2}\Psi_{\tau} -(A^{\theta}\cdot \nabla + \frac{1}{2}\nabla\cdot A^{\theta})\,\Psi_{\tau}  \\
&\quad :=\frac{1}{4}|A^{\theta}|^{2}\Psi_{\tau} -\mathcal{T}_{A,\theta}\,\Psi_{\tau} 
\end{align*}
So that the evolution may be given as
\begin{align*}
    &\partial_{\tau } \Psi_{\tau} =-\widehat{H}_t^\theta \Psi_{\tau} - \frac{\dot{t}}{2} \Psi \partial_{t} \log \rhot_{t} +\frac{1}{4}|A^{\theta}|^{2}\Psi_{\tau} -\mathcal{T}_{A,\theta}\,\Psi_{\tau}
\end{align*}
Recalling the proof of Theorem~\ref{thrm:adiabatic_FK}, we need only control the new terms $-\langle \Psi_{\perp,\tau} | \mathcal{T}_{A,\theta} \Psi_{\tau}\rangle$ and $+\frac{1}{4}\langle \Psi_{\perp,\tau} |\,  |A^{\theta}|^2\,\,| \Psi_{\tau} \rangle$ in addition to the original Gr{\"o}nwall bound. Observe that for $\Psi_{\tau} = \Psi_{\perp,\tau} + \psi_{0}$ one has
\begin{align*}
&-\langle \Psi_{\perp,\tau} | \,(A^{\theta}\cdot \nabla + \frac{1}{2}\nabla\cdot A^{\theta}) \,\Psi_{\tau}\rangle \\
&\quad =-\langle \Psi_{\perp,\tau} | \,(A^{\theta}\cdot \nabla + \frac{1}{2}\nabla\cdot A^{\theta}) \,\Psi_{\perp,\tau}\rangle\\
&\quad -\langle \Psi_{\perp,\tau} | \,(A^{\theta}\cdot \nabla + \frac{1}{2}\nabla\cdot A^{\theta}) \,\psi_{0}\rangle
\end{align*}
For the first, observe
\begin{align*}
&-\langle \Psi_{\perp,\tau} | \,(A^{\theta}\cdot \nabla + \frac{1}{2}\nabla\cdot A^{\theta}) \,\Psi_{\perp,\tau}\rangle \\
&= -\int \Psi_{\perp,\tau} A^{\theta} \cdot \nabla\Psi_{\perp,\tau} - \frac{1}{2}\int \Psi_{\perp,\tau}^{2} (\nabla \cdot A^{\theta})
\end{align*}
As $\mathcal{T}_{A,\theta}$ is skew-symmetric, with an integration by parts, we find that the previously displaced expression becomes
\begin{align*}
\frac{1}{2}\int \Psi_{\perp,\tau}^{2} (\nabla \cdot A^{\theta})  - \frac{1}{2}\int \Psi_{\perp,\tau}^{2} (\nabla \cdot A^{\theta}) = 0
\end{align*}

For $\psi_{0}$ terms, we are left with
\begin{align*}
&-\langle \Psi_{\perp,\tau} | \mathcal{T}_{A,\theta} \psi_{0}-\frac{1}{4} |A_{\theta}|^{2} \psi_{0}\rangle \\
&=-\langle \Psi_{\perp,\tau} | A^{\theta}\cdot \nabla \psi_{0} + \frac{1}{2}\nabla\cdot A^{\theta} \psi_{0}-\frac{1}{4} |A_{\theta}|^{2} \psi_{0}\rangle \\
&=-\langle \Psi_{\perp,\tau} | (\tfrac{1}{2}A^{\theta}\cdot \nabla \log \rhot+ \tfrac{1}{2}\nabla\cdot A^{\theta} + \tfrac{1}{4} |A^{\theta}|^{2} -\tfrac{1}{2} |A^{\theta}|^{2} )\psi_{0}\rangle \\
&=-\langle \Psi_{\perp,\tau} | (E_{0} -\tfrac{1}{2} |A^{\theta}|^{2} ) \psi_{0}\rangle 
\end{align*}
Use that $\Psi_{\perp,\tau} \perp \psi_{0}$ we can subtract the mean to find that this becomes
\begin{align*}
&=+\langle \Psi_{\perp,\tau} | \tfrac{1}{2} |A^{\theta}|^{2}  \psi_{0}\rangle = \langle \Psi_{\perp,\tau} | \tfrac{1}{2} (|A^{\theta}|^{2} -4 \lambda_{\theta}) \psi_{0}\rangle \\
&\leq \frac{1}{2} \lVert \Psi_{\perp,\tau} \rVert_{L^{2}} \lVert (|A^{\theta}|^2-4\lambda_{\theta})\psi_{0}\rVert_{L^{2}} \\
&= \frac{\omega}{2} \sqrt{\mathrm{Var}_{\tilde{\rho}}|A^{\theta}|^2} .
\end{align*}
The final term is $+\frac{1}{4}\langle \Psi_{\perp,\tau} |  \,\,|A^{\theta}|^{2}\,\Psi_{\perp,\tau}\rangle$. If we denote our multiplication operator $\frac{1}{4}|A^{\theta}|^{2} \in L^{2}$ and there exist $\alpha(t),\beta(t) \geq 0$ so that the relative form bound holds
\[
\langle u, \frac{1}{4}|A^{\theta}|^{2}  u \rangle_{L^{2}} \leq \alpha(t) \langle u, \widehat{H}_{t} u \rangle_{L^{2}} + \beta(t) \langle u, u \rangle_{L^{2}}
\]
then,
\begin{align*}
&\langle \Psi_{\perp,\tau} |\, \frac{1}{4}|A^{\theta}|^{2} \,\Psi_{\perp,\tau}\rangle \\
&\quad \leq \alpha(t) \langle \Psi_{\perp,\tau}, \widehat{H}_{t} \Psi_{\perp,\tau} \rangle_{L^{2}} 
+ \beta(t) \lVert \Psi_{\perp,\tau} \rVert_{L^{2}}^{2} 
\end{align*}
Thus, collecting with the bounds of Theorem~\ref{thrm:adiabatic_FK} one has
\begin{align*}
&- \langle \Psi_{\perp,\tau} |\, \widehat{H}_{\tau} \,\Psi_{\perp,\tau}\rangle
- \langle \Psi_{\perp,\tau} |\, \widehat{G}_{\tau} \,\Psi_{\perp,\tau}\rangle+
\langle \Psi_{\perp,\tau} |\, \frac{1}{4}|A^{\theta}|^{2} \,\Psi_{\perp,\tau}\rangle \\
&\leq -(1-(|\dot{t}|a(t) + \alpha(t))) \langle \Psi_{\perp,\tau} |\, \widehat{H}_{\tau} \,\Psi_{\perp,\tau}\rangle\, \\
&+ (|\dot{t}|b(t) + \beta(t)) \lVert \Psi_{\perp,\tau} \rVert_{L^{2}}^{2}.
\end{align*}
As we assume the rotational coefficients satisfy $\alpha(t)<1$ when $|\dot{t}|$ is taken sufficiently small one has $(1-(|\dot{t}|a(t) + \alpha(t)))>0$. Now, note for $\Psi_{\perp,\tau}$ orthogonal to the ground-state that
\begin{align*}
\langle \Psi_{\perp,\tau}, \widehat{H}_{t} \Psi_{\perp,\tau} \rangle_{L^{2}} &\geq E_{1}(t) \lVert \Psi_{\perp,\tau} \rVert_{L^{2}}^{2} \\
&= E_{1}(t)\omega^{2} = (\Delta_t + \lambda_{\theta}(t)) \,\omega^{2}
\end{align*}
for $\Delta_{t}=E_{1}-E_{0}$ is the spectral gap of the non-conservative Score Hamiltonian, and where the eigenvalue $E_{1}(t)$ of the non-conservative Score Hamiltonian is lifted by the non-zero ground state energy $\lambda_{\theta}(t)$ to $E_{1}(t)=E_{0}(t)+\Delta_t = \lambda_{\theta}(t) + \Delta_t$. Thus, for $1-(|\dot{t}|a(t) + \alpha(t))>0$ one has
\begin{align*}
&-(1-(|\dot{t}|a(t) + \alpha(t))) \langle \Psi_{\perp,\tau} |\, \widehat{H}_{\tau} \,\Psi_{\perp,\tau}\rangle \\
&\leq - (1-(|\dot{t}|a(t) + \alpha(t)))(\Delta_t + \lambda_{\theta}(t))\lVert \Psi_{\perp,\tau} \rVert_{L^{2}}^{2}
\end{align*}
Which, when combined with the right term yields the quadratic component of the bound on $\omega^{2}$
\begin{align*}
&\leq -[ (1-(|\dot{t}|a(t) + \alpha(t)))(\Delta_t + \lambda_{\theta}(t)) \\
&- (|\dot{t}|b(t) + \beta(t)) ] \lVert \Psi_{\perp,\tau} \rVert_{L^{2}}^{2} \\
&:= - \Gamma(\tau) \omega^{2}
\end{align*}
Where $\Delta_{t}>\frac{\beta(t)}{1-\alpha(t)}$ implies that there exists a sufficiently slow schedule $|\dot{t}|$ so that $\Gamma(\tau)>0$.
Thus, the argument of Theorem~\ref{thrm:adiabatic_FK} proceeds with a rotationally-taxed effective spectral gap, expressed as
\[
\Gamma(\tau) = (1-|\dot{t}|a(t) - \alpha(t))(\Delta_t + \lambda_{\theta}(t)) - |\dot{t}| b(t) - \beta(t)
\]
We then are left with a linear $\omega$-dependent term of
\begin{align*}
&\omega\left( |\dot{t}| \lVert \partial_{t}\log \rhot_{t}\rVert_{L^{2}({\rhot}_{t})}
+ \frac{1}{2} \sqrt{\mathrm{Var}_{\tilde{\rho}_{t}}|A^{\theta}|^2}
\right). 
\end{align*}
Analogously to Theorem~\ref{thrm:adiabatic_FK}, this yields the bound of
\begin{align*}
\omega(\tau) &\le  \omega(0) \exp\left(-\int_0^\tau \Gamma(s)ds\right)  \\
&+ \int_0^\tau e^{-\int_s^\tau \Gamma(r)dr}\biggl(|\dot{t}(s)| \lVert \partial_{t} \log \rhot_{t(s)} \rVert_{L^2 (\rhot_{t(s)} )} \\
&\qquad \qquad +  \frac{1}{2} \sqrt{\mathrm{Var}_{\tilde{\rho}_{t(s)}}|A^{\theta}|^2} \biggr) ds \\
&\leq  \omega(0) \exp\left(-\int_0^\tau \Gamma(s)ds\right)  \\
&+ \sup_{s \le \tau} \biggl( \frac{|\dot{t}(s)| \lVert \partial_{t} \log \rhot_{t(s)} \rVert_{L^2 (\rhot_{t(s)} )}}{\Gamma(s)} \\
&+  \frac{1}{2\Gamma(s)} \sqrt{\mathrm{Var}_{\tilde{\rho}_{t(s)}}|A^{\theta}|^2} \biggr).
\end{align*}
Since $A^{\theta}_{t(s)}=S^{\theta}_{t(s)} - \nabla \log \rhot_{t(s)}$, one may express the upper bound in terms of the error of the induced conservative score, concluding the proof.
\end{proof}

One of the key features of the non-conservative case is that the final term in the supremum is not controlled by taking $|\dot{t}|$ asymptotically smaller, indicating a harder limit to adiabatic tracking imposed by non-conservativity. A natural question is whether there is a setting of $\alpha,\beta$ in the form bound implies rotational fields may accelerate the adiabatic process via the raised ground-state $\lambda_{\theta}(t)$. Let $\langle \psi_{1}, \frac{1}{4}|A^{\theta}|^{2}  \psi_{1} \rangle_{L^{2}} = \gamma_{\theta}$. Taking $u = \psi_{1} \perp \psi_{0}$ to be the first excited state implies the universal inequality
\begin{align*}
&\langle \psi_{1}, \frac{1}{4}|A^{\theta}|^{2}  \psi_{1} \rangle_{L^{2}} \leq \alpha(t) \langle \psi_{1}, \widehat{H}_{t} \psi_{1} \rangle_{L^{2}} + \beta(t) \lVert \psi_{1} \rVert_{L^{2}}^{2} \\
&=\alpha(t) (\Delta_t + \lambda_{\theta}(t)) + \beta(t) \\
&-\beta \leq \alpha(t) (\Delta_t + \lambda_{\theta}(t)) - \gamma_{\theta}
\end{align*}
While, in the spectral gap, this implies
\begin{align*}
&\Gamma(\tau) = (1-|\dot{t}|a(t) - \alpha(t))(\Delta_t + \lambda_{\theta}(t)) - |\dot{t}| b(t) - \beta(t) \\
&\leq (\Delta_t + \lambda_{\theta})(t)-|\dot{t}|a(t)(\Delta_t + \lambda_{\theta}(t)) \\
&- \alpha(t)(\Delta_t + \lambda_{\theta}(t)) - |\dot{t}| b(t) \\
&+\alpha(t) (\Delta_t + \lambda_{\theta}(t)) - \gamma_{\theta} \\
&=(1-|\dot{t}|a(t))\Delta_t  - |\dot{t}| b(t) + (1-|\dot{t}|a(t))\lambda_{\theta} - \gamma_{\theta} 
\end{align*}
For $|\dot{t}| \downarrow 0$, this reveals
\begin{align*}
\Gamma(\tau) \leq \Delta_t + \lambda_{\theta} - \gamma_{\theta}
\end{align*}
So that the raised ground-state $\lambda_{\theta}$ is not ``free'' in raising the effective gap, depending on the sign of the residual
\begin{align*}
&\lambda_{\theta} - \gamma_{\theta} = \langle \psi_{0}, \frac{1}{4}|A^{\theta}|^{2} \psi_{0} \rangle - \langle \psi_{1}, \frac{1}{4}|A^{\theta}|^{2} \psi_{1} \rangle
\end{align*}
This can be positive and raise the upper-bound on the gap beyond the conservative case of $\Delta_{t}$; for instance when $A^{\theta}$ has high rotational kinetic energy (stirs heavily) over $\psi_{0}$ and low rotational energy (stirs weakly) over the first excited state $\psi_1$.

The low-variance regime $\mathrm{Var}_{\tilde{\rho}}[X] \downarrow 0$ implies $|A^{\theta}|^2 = \mathrm{const}$ $\tilde{\rho}$-almost everywhere. Since the Stationary Schr{\"o}dinger condition implies
\[
\frac{1}{2}\nabla\cdot A^{\theta} +\frac{1}{4} \lVert A^{\theta} \rVert^{2} + \frac{1}{2} A^{\theta} \cdot \nabla \log \rhot = E_{0}
\]
the rotational kinetic energy is a constant $E_{0}=\frac{1}{4} \lVert A^{\theta} \rVert^{2}$ with kinetic variance zero if $A^{\theta}$ is solenoidal $\nabla\cdot A^{\theta} = 0$ and isoclinic $A^{\theta} \cdot \nabla \log \rhot = 0$. Then, $\lambda_{\theta} - \gamma_{\theta} = 0$ and thus the upper bound reduces to the conservative case
\begin{align*}
\omega(0) &\exp\left(-\int_0^\tau \Gamma(s)ds\right) \\
&+ \sup_{s \le \tau} \,\, \frac{|\dot{t}(s)| \lVert \partial_{t} \log \rhot_{t(s)} \rVert_{L^2 (\rhot_{t(s)} )}}{\Gamma(s)} 
\end{align*}
Thus, isoclinic solenoidal fields are inherently benign relative to the standard conservative setup.

\section{Relation to Quantum Inverse Problems}

\label{app:ham_identification} \begin{theorem}[Hamiltonian Identification via Score Matching] \label{thrm:Ham_Identification} Let $\widehat{H} = -(\hbar^2/2m)\nabla^2 + V$ be a time-reversible Schrödinger operator with ground-state wave function 
\begin{align*} 
\psi(x,t) = \sqrt{\rho(x)}\,e^{-iEt/\hbar}, \end{align*} 
where $\rho(x) = |\psi(x,t)|^2$ is the ground-state probability density on nodal domain $\rho > 0$ with $E$ the ground-state energy. Let $S_\theta(x) \approx \nabla\log\rho(x)$ be any approximation of the score of $\rho$, for instance obtained by score matching on samples of $\rho$. Then the potential $V$ is recovered from $S_\theta$ via \begin{align} V_{\theta}(x) = E + \frac{\hbar^2}{2m} \left( \frac{1}{2}\nabla\cdot S_\theta(x) + \frac{1}{4}|S_\theta(x)|^2 \right). \label{eq:V_from_score} \end{align} In particular, if $S_\theta = \nabla\log\rho$ exactly, then $V_{\theta}(x) = V(x)$ \eqref{eq:V_from_score} holds with equality up to the spectral constant $E$. \end{theorem}

\begin{proof}[Proof of Thm.~\ref{thrm:Ham_Identification}] By the Madelung transform, any (nodeless) solution $\psi = \sqrt{\rho}\,e^{i\phi/\hbar}$ to the Schr{\"o}dinger equation satisfies the quantum Hamilton-Jacobi equation 
\begin{align} 
\partial_t \phi + V + \frac{\|\nabla\phi\|^2}{2m} - \frac{\hbar^2}{2m}\frac{\nabla^2\sqrt{\rho}}{\sqrt{\rho}} = 0. \label{eq:HJ} 
\end{align} 
We use $\phi$ for the quantum phase to distinguish it from the score $S_\theta = \nabla\log\rho$. For the stationary ground state $\psi(x,t) = \sqrt{\rho(x)}\,e^{-iEt/\hbar}$, the phase is purely temporal: $\phi(x,t) = -Et$. Therefore \begin{align*} \nabla\phi = 0, \qquad \partial_t\phi = -E. \end{align*} The spatial gradient vanishes because the ground-state density $\rho(x)$ carries no spatial phase and there are no probability currents in a stationary state. Substituting into \eqref{eq:HJ} eliminates the kinetic term $\|\nabla\phi\|^2/2m$ and gives \begin{align*} 
-E + V(x) - \frac{\hbar^2}{2m}\frac{\nabla^2\sqrt{\rho}}{\sqrt{\rho}} = 0, 
\end{align*} hence 
\begin{align} 
V(x) = E + \frac{\hbar^2}{2m}\frac{\nabla^2\sqrt{\rho}}{\sqrt{\rho}}. \label{eq:V_amplitude} 
\end{align}
This is the amplitude form of the result. 

Applying the Bohm Score Identity, one observes for the domain $\rho >0$ that
\begin{align*} \frac{\nabla^2\sqrt{\rho}}{\sqrt{\rho}} = \frac{1}{2}\nabla\cdot S + \frac{1}{4}|S|^2.
\end{align*} 
Substituting $S_\theta = S$ gives \eqref{eq:V_from_score}. 
\end{proof}

A Hamiltonian with nodeless ground-state $\psi_{0} = \sqrt{\rho_{0}}$ has potential expressible in terms of $\psi_{0}$, $V=E_{0} + \frac{\hbar^{2}}{2m}\frac{\nabla^{2} \psi_{0}}{\psi_{0}}$. Thm.~\ref{thrm:Ham_Identification} corresponds to the minimal observation that $\frac{\nabla^{2} \psi_{0}}{\psi_{0}}$ can be expanded in terms of the ground-state score $S_{0}$ and thus that score-matching on $\stheta$ can learn the Hamiltonian up to $E_{0}$.

More generally, any $\hat{H}$ with nodeless ground-state $\psi_{0}$ can be expressed exactly in the Score (Information) Hamiltonian form of
\begin{align*}
&\hat{H} = -\frac{\hbar^{2}}{2m}\nabla^{2}+\frac{\hbar^{2}}{2m}\frac{\nabla^{2}\psi_{0}}{\psi_{0}}+E_{0} \\
&=-\frac{\hbar^{2}}{2m}\nabla^{2}+\frac{\hbar^{2}}{2m}\left(
\frac{1}{2}\nabla^{2} \log |\psi_{0}|^{2} + \frac{1}{4}
|\nabla \log |\psi_{0}|^{2}|^{2}\right)
+E_{0}
\end{align*}
Real-time evolution on $i\hbar\partial_{t}\psi =  \hat{H} \psi$, under standard Wick-rotation $\tau = it$, $t=-i\tau$, and $\partial_{t}=i\partial_{\tau}$, yields $\partial_{\tau}\psi =  -\frac{1}{\hbar}\hat{H} \psi$ and thus
\begin{align*}
&\partial_{\tau}\mu_{\tau} = -\frac{\hbar}{m}\nabla\cdot \left(\mu_{\tau} \left(\frac{\nabla \psi_{0}}{ \psi_{0}}
\right) \right) + \frac{\hbar}{2m} \nabla^{2}\mu_{\tau} \\
&=-\frac{\hbar}{2m}\nabla\cdot \left(\mu_{\tau} \left(\nabla \log |\psi_{0}|^{2}
\right) \right) + \frac{\hbar}{2m} \nabla^{2}\mu_{\tau} 
\end{align*}
which states that the dissipative rotation of real-time evolution is simply governed by a Fokker-Planck equation with drift given by the ground-state score $S_0=\nabla \log |\psi_{0}|^{2}$.

\begin{proposition}\label{prop:Schrodinger_relFisher}
Consider a state $\psi_{t} = \sqrt{\mu_{t}}e^{i\phi_{t}/\hbar}$ in density and phase variables $(\mu_{t},\phi_{t})$. Let $\hat{H}$ be a time-independent Hamiltonian on $L^{2}$ with nodeless ground-state $\psi_{0}=\sqrt{\rho}$. Then, real-time Madelung flow of $(\mu_{t},\phi_{t})$ conserves the Hamiltonian expectation value
\begin{align}
&\langle \psi_{t}|\hat{H}|\psi_{t}\rangle-E_{0}\\
&=\frac{\hbar^{2}}{8m}\int \left\lVert  \nabla \log \mu_t - S_{0} \, \right\rVert^{2}\mu_t\,dx + \frac{1}{2m}\int \left\lVert \nabla \phi_t \right\rVert^{2}\,\mu_{t}\,dx \label{eq:rel_fish_QM}
\end{align}
For $S_{0}=\nabla \log |\psi_{0}|^{2}$ the score of the ground-state density and $E_{0}$ the ground-state energy.
\end{proposition}
\begin{proof}
In the case of evolution on a real-time state $\psi_{t} = \sqrt{\mu_{t}}\,e^{+i\phi_t/\hbar}$ with phase $\phi$, the Score Hamiltonian offers an Information-Geometric perspective tying the Hamiltonian of the Schr{\"o}dinger equation to a Fisher-divergence. In particular, the extension of the sum-of-squares of Equation~\eqref{eq:rayleigh_squares} to complex states gives for $\hat{H}=-\frac{\hbar^{2}}{2m}\nabla^{2}+\frac{\hbar^{2}}{2m}\frac{\nabla^{2}\psi_{0}}{\psi_{0}}+E_{0}$ the energy
\begin{align*}
\langle \psi_{t}|\hat{H}|\psi_{t}\rangle = \frac{\hbar^{2}}{2m}\int \lVert \nabla \psi_t - \frac{1}{2} S_0 \psi_t \rVert^2 dx + E_{0}
\end{align*}
Now, taking the gradient one sees
\begin{align*}
&\nabla \psi_{t} = \nabla\sqrt{\mu_t}e^{i\phi_t/\hbar} + \frac{i}{\hbar}\sqrt{\mu_t} \nabla \phi_t e^{i\phi_t/\hbar} \\
&=\left( \frac{\nabla\mu_t}{2 \sqrt{\mu_{t}}} +\frac{i}{\hbar}\sqrt{\mu_t} \nabla \phi_t\right) e^{i\phi_t/\hbar}
\end{align*}
So that
\begin{align*}
&\frac{\hbar^{2}}{2m}\int \lVert \nabla \psi_t - \frac{1}{2} S_{0}\,\psi_t \rVert^{2}\,dx \\
&=\frac{\hbar^{2}}{2m}\int \lVert \left( \frac{\nabla\mu_t}{2 \mu_{t}} +\frac{i}{\hbar} \nabla \phi_{t}\right) \sqrt{\mu_t}e^{i\phi_{t}/\hbar} - \frac{1}{2} S_{0}\,\psi_t \rVert^{2}\,dx \\
&=\frac{\hbar^{2}}{2m}\int \lVert \left( \frac{\nabla\mu_t}{2 \mu_{t}} +\frac{i}{\hbar} \nabla \phi_{t} - \frac{1}{2} S_{0} \,\right) \psi_t \rVert^{2}\,\,dx 
\end{align*}
Expanding the square of the complex-term, we see
\begin{align*}
&=\frac{\hbar^{2}}{2m}\int \left\lVert  \frac{\nabla\mu_t}{2 \mu_{t}}  - \frac{1}{2} S_{0} \, \right\rVert^{2}\mu_t\,dx + \frac{\hbar^{2}}{2m}\int \left\lVert  \frac{i}{\hbar} \nabla \phi_{t} \right\rVert^{2}\,\mu_{t}\,dx \\
&=\frac{\hbar^{2}}{2m}\int \left\lVert  \frac{1}{2}\nabla \log \mu_t  - \frac{1}{2} S_{0} \, \right\rVert^{2}\mu_t\,dx + \frac{\hbar^{2}}{2m}\int \left\lVert  \frac{i}{\hbar} \nabla \phi_{t} \right\rVert^{2}\,\mu_{t}\,dx
\end{align*}
From which we see, in terms of the Born density $\mu_{t}$, the conserved energy can be given as a sum of the relative Fisher divergence (as in score-matching \cite{hyvarinen05a}) to the ground-state score $S_{0}$, and a kinetic energy of the phase $\phi_t$:
\begin{align*}
&\frac{\hbar^{2}}{8m}\int \left\lVert  \nabla \log \mu_t - S_{0}  
\right\rVert^{2}\mu_tdx \\
&+ \frac{1}{2m}\int \left\lVert \nabla \phi_t \right\rVert^{2}\mu_{t}dx +E_{0}
\end{align*}
If $\hat{H}$ is time-independent one has $\frac{d}{dt}\langle \psi_{t}|\hat{H}|\psi_{t}\rangle = 0$ so that~\eqref{eq:rel_fish_QM} is conserved.
\end{proof}
While the imaginary-time relaxation treats the relative Fisher as a Lyapunov functional settling to $S_{0}$, the real-time variant conserves the energy and is akin to an \emph{information-pendulum}. The phase behaves as kinetic energy and the relative Fisher information as the potential. Their sum being conserved implies that for a finite energy budget $E-E_{0}>0$ the score $\nabla \log \mu_{t}$ can never come to rest at $S_{0}$. Like a pendulum through equilibrium, $\nabla\log \mu_{t}$ may pass through $S_{0}$ and transfer the budget to phase kinetic energy, but a stationary configuration would require a vanishing phase gradient and $E=E_{0}$.

Returning to the decomposition of Prop.~\ref{prop:Schrodinger_relFisher}, an integration by parts further expresses this relative Fisher term as
\begin{align*}
&=\frac{\hbar^{2}}{8m}\int\lVert  \nabla \log \mu_t \rVert_{2}^{2}\mu_t \\
&+ \int \left(\frac{\hbar^{2}}{2m}\left( \frac{1}{2}\nabla \cdot S_{0} + \frac{1}{4} |S_{0}|^{2} \right) +E_{0}\right) \mu_{t} \\
&+ \frac{1}{2m}\int \left\lVert \nabla \phi_{t} \right\rVert^{2}\,\mu_{t}\,dx \\
&=\frac{\hbar^{2}}{8m}\int\lVert  \nabla \log \mu_t \rVert_{2}^{2}\mu_t + \int \left( \frac{\hbar^{2}}{2m}\frac{\nabla^{2}\psi_{0}}{\psi_{0}} +E_{0}\right) \mu_{t} \\
&+ \frac{1}{2m}\int \left\lVert \nabla \phi_{t} \right\rVert^{2}\,\mu_{t}\,dx
\end{align*}
Which, as $\frac{\hbar^{2}}{2m}\frac{\nabla^{2}\psi_{0}}{\psi_{0}} +E_{0} = V$ defines the potential, recovers the canonical Hamiltonian of Madelung quantum mechanics in Wasserstein space \cite{von_Renesse_2012}
\[
\frac{\hbar^{2}}{8m}\int\lVert  \nabla \log \mu_t \rVert_{2}^{2}\mu_t + \int V \mu_{t} + \frac{1}{2m}\int \left\lVert \nabla \phi_{t} \right\rVert^{2}\,\mu_{t}\,dx
\]
While mathematically equivalent to the perspective of \cite{von_Renesse_2012}, this offers an Information-Hydrodynamical interpretation of the Madelung equations. In particular for $|\psi_{0}|^{2}=\rho_{0}$ up to the $\phi_t \to \phi_t-E_{0}t$ Gauge, the equations become
\begin{align*}
&\partial_{t}\mu_t = - \nabla \cdot \left(\mu_{t}\frac{\nabla \phi_{t}}{m} \right), \\
&\partial_{t}\phi_{t} + \frac{1}{2m}|\nabla \phi_{t}|^{2} + \frac{\hbar^{2}}{2m}\left( \frac{\nabla^{2}\sqrt{\rho_{0}}}{\sqrt{\rho_{0}}} - \frac{\nabla^{2}\sqrt{\mu_{t}}}{\sqrt{\mu_{t}}} \right) = 0
\end{align*}
Which reveals that the Madelung equations \cite{Madelung1927} may be expressed in terms of the quantum-pressure gap between the current $\mu_{t}$ and the ground-state $\rho_{0}$, $\delta P_{Q}=Q_{M}[\mu_{t}]-Q_{M}[\rho_{0}]$ where $Q_{M}$ is the standard Madelung writing of the quantum potential signed negatively with physical pre-factor $\frac{\hbar^{2}}{2m}$, $Q_{M}[p]=-\frac{\hbar^{2}}{2m}\frac{\nabla^{2}\sqrt{p}}{\sqrt{p}}$.
\[
\partial_{t}\phi_{t} + \frac{1}{2m}|\nabla \phi_{t}|^{2} + \delta P_{Q} =0
\]
Thus, Madelung hydrodynamics can be expressed in a form only dependent on the phase and information densities $\rho_{0},\mu_{t}$ without reference to a potential. By definition, if $\mu_t$ is nodeless then $\delta P_{Q}$ can be expressed in terms of the scores $\nabla \log \mu_t, \nabla \log \rho_{0}$ alone \cite{Bohm1952,SBITNEV2009, Nelson1966}.

\section{Numerical Setup}\label{appendix:numerical-setup}

\subsection{Hydrogen Atom (3D) Hamiltonian via Ground-State Samples}

\textbf{Setup.} We test whether a score-network $\stheta$ trained only on samples from the hydrogen ground-state density recovers the Coulomb Hamiltonian and its full excited-state spectrum. The landmark theorem of Hohenberg-Kohn in density functional theory (DFT) states that the ground-state density of a system uniquely determines its potential, and thus all excited states of the system and all of its properties -- including many-body wave functions \cite{Hohenberg1964}. Thus, learning from samples of the ground-state $\rho_{0}=|\psi|^{2}$ should theoretically guarantee the capacity to recover $\hat{V}$, $\widehat{H}$, and all excited states $\psi_{k\geq1}$. We take samples from the 1s ground state orbital of the Hydrogen atom, with density
\[
\rho_0(\mathbf{x}) \propto |\psi_{1s}(\mathbf{x})|^2 = \frac{1}{\pi}e^{-2r}, \qquad r = \|\mathbf{x}\|.
\]
We draw $20{,}000$ exact samples by sampling radii $r \sim \mathrm{Gamma}(k{=}3, \theta{=}1/2)$ with isotropic angles and mapping to Cartesian coordinates.

\textbf{Score Network.} We train a Denoising Score-Matching model \texttt{QuantumStateNet} which has a Coulomb cusp inductive bias $\log\rho_\theta(\mathbf{x}, \sigma) = f_\theta(\mathbf{x}, \sigma) - 2Z_{\mathrm{eff}}\,r,$ where $f_{\theta}$ is a 3-hidden-layer MLP with width $256$ and SiLU activations. The network is zero-initialized in its final layer. The noise embedding uses Gaussian Fourier projections with noise $\sigma \in [0.01,0.3]$ and the model is trained for $2{,}000$ iterations with a cosine learning rate. We also train a RealNVP normalizing flow \cite{dinh2017density, pmlr-v37-rezende15} baseline $\rho_{\eta}$ with 8 affine-coupling layers and a similar Coulomb cusp inductive bias.

From the trained score network, we infer the Score Hamiltonian potential $\hat{V}_{\mathrm{Score\text{-}H}}(r)=Q_\theta(r)$, the thermodynamic potential $\hat{V}_{\mathrm{thermo}}(r) = -\int_0^r \bigl(S_r(r') - S_r(\infty)\bigr)\,dr',$ and the potential $\hat{V}_{\mathrm{Flow}}=-\log\rho_\eta(r)$. As we trained a conservative score  $S = \nabla \log \rho_{\theta}$ for a network $\log \rho_{\theta}$, we simply use the shared potential $\hat{V}_{\mathrm{thermo}}(r) = - \log \rho_{\theta}$ directly.

\textbf{Extraction of the Spectrum.} For each potential recovered, we (i) vacuum align the potential by subtracting the tail mean, (ii) smooth each potential with a Gaussian filter $\sigma=0.01$, and (iii) Gauge-shift the potential so the recovered ground-state energy matches that of the true 1s orbital at -0.5 Hartree. This is done, as recovery is up to the spectral constant offset in Thm.~\ref{thrm:Ham_Identification}, so energies are reported relative to this anchoring. We then solve the radial Schr\"{o}dinger equation via tridiagonal discretization on a $2{,}000$ point grid over $[0.01,60]$ Bohr with the effective potential
\[
V_{\mathrm{eff}}(r) = \hat{V}(r) + \frac{\ell(\ell+1)}{2r^2}, \qquad \ell \in \{0, 1, 2\},
\]
We retain bound states up to $n\leq 4$.

We visualize the excited states of the Score Hamiltonian in Fig.~\ref{fig:HydrogenOrbitals}, showing the recovered orbital densities $|\psi_{nl}(\mathbf{x})|^2$ from the radial eigenfunctions of the score-Hamiltonian $\widehat{H}_\theta$ in a spherical harmonic basis $Y_{\ell }^{m}(\bm{r })$. In Table~\ref{tab:hydrogen_spectrum} we find the Score Hamiltonian recovers the nodal structure for all excited-state orbitals considered to $n=3$, including the $\mathrm{s}$, $\mathrm{p}$, and $\mathrm{d}$ families from training from $1\mathrm{s}$ samples.

\begin{table*}[ht]
\centering
\caption{Hydrogen atom spectrum recovery. The Score Hamiltonian (Definition~\ref{def:ScoreHamiltonian}) successfully reconstructs the excited state spectrum solely from samples of the ground state density, whereas baseline methods fail to find accurate bound states.}
\label{tab:hydrogen_spectrum}
\vskip 0.15in
\begin{small}
\begin{tabular}{lc cc cc cc}
\toprule
 & & \multicolumn{2}{c}{Score Hamiltonian} & \multicolumn{2}{c}{Thermodynamic Integral} & \multicolumn{2}{c}{Boltzmann Generator} \\
\cmidrule(lr){3-4} \cmidrule(lr){5-6} \cmidrule(lr){7-8}
State & Exact $E_n$ & Energy & Abs. Err. & Energy & Abs. Err. & Energy & Abs. Err. \\
\midrule
1s & -0.5000 & -0.5000 & 0.0000 & -0.5000 & 0.0000 & -0.5000 & 0.0000 \\
2s & -0.1250 & -0.1281 & 0.0031 & Unbound & - & Unbound & - \\
2p & -0.1250 & -0.1230 & 0.0020 & Unbound & - & Unbound & - \\
3s & -0.0556 & -0.0607 & 0.0051 & Unbound & - & Unbound & - \\
3p & -0.0556 & -0.0593 & 0.0037 & Unbound & - & Unbound & - \\
3d & -0.0556 & -0.0598 & 0.0042 & Unbound & - & Unbound & - \\
\bottomrule
\end{tabular}
\end{small}
\end{table*}

\begin{figure}[tbp]
    \centering
    \includegraphics[width=\linewidth]{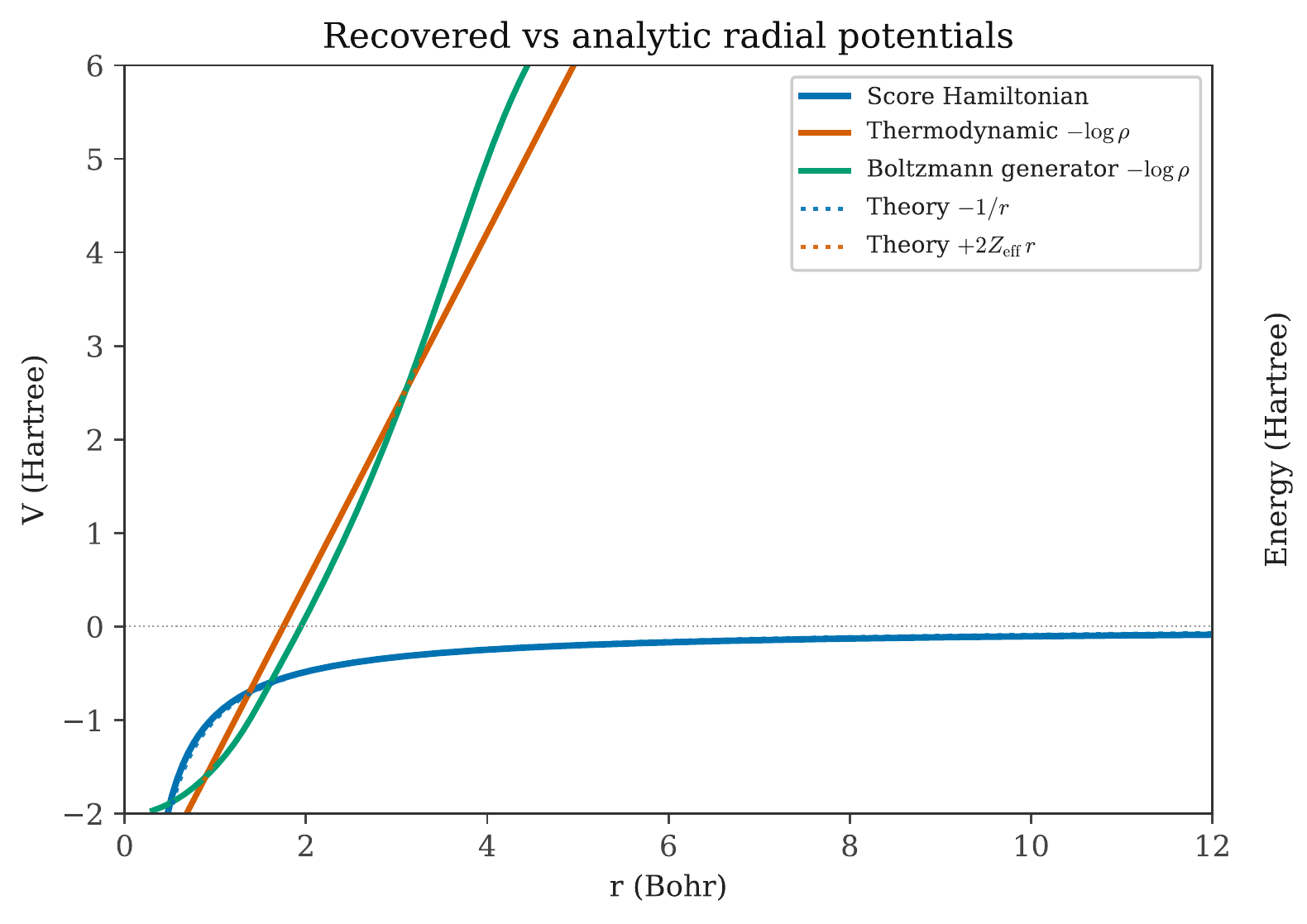}
    \caption{Comparison of thermodynamic and quantum potential of a diffusion model and flow, with ground-truth Coulombic and thermodynamic potentials (Hartrees v. Bohr radius).}\label{fig:HydrogenPotential}
\end{figure}

\subsection{Coupled Harmonic Oscillator}\label{sec:CHO} 

\begin{figure}[tbp]
    \centering
    \includegraphics[width=\linewidth]{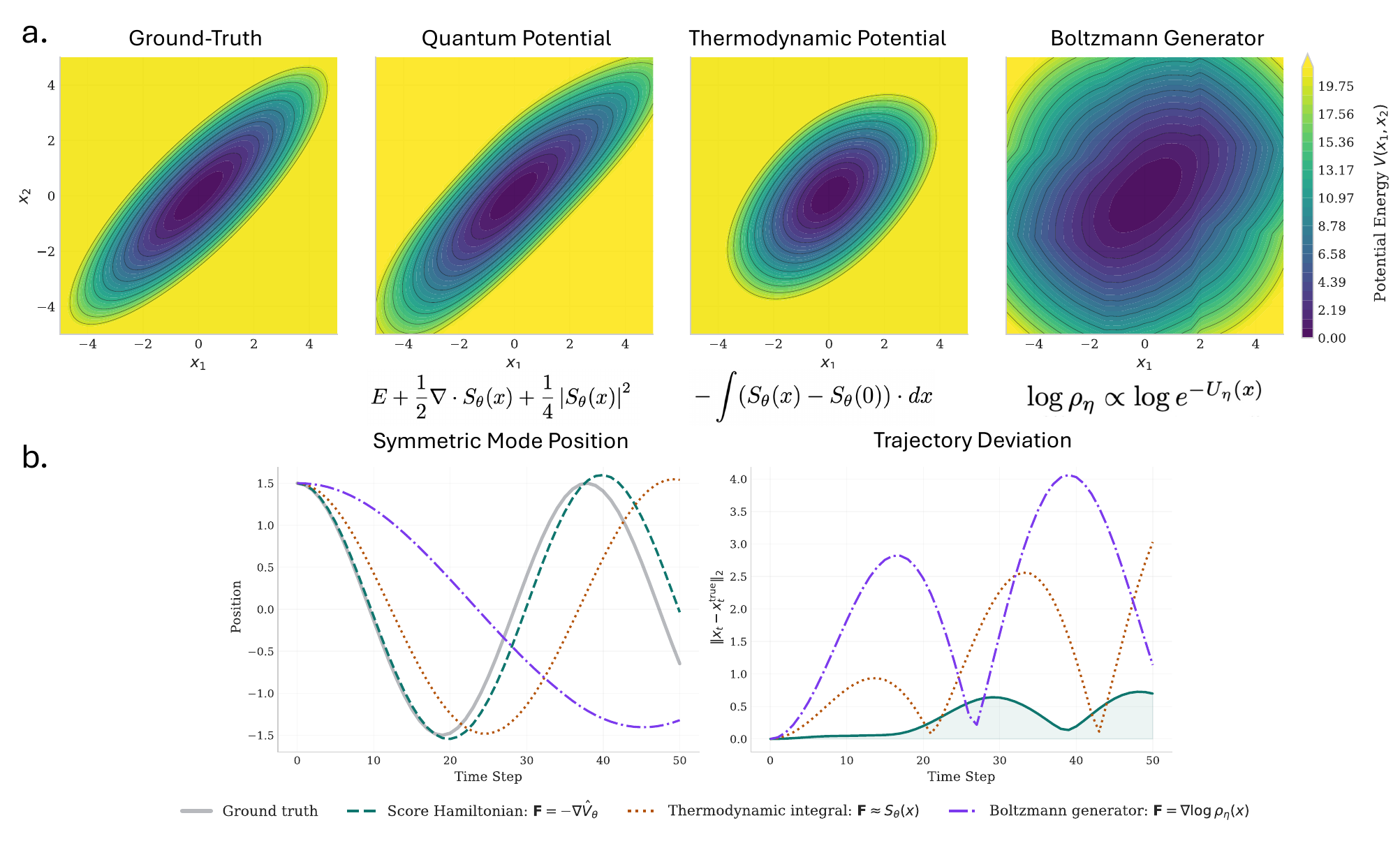}
    \caption{\textbf{Hamiltonian Inference on the Coupled Harmonic Oscillator.} \textbf{(a.)} Visualization of the true coupled harmonic-oscillator potential $\hat{V}(x_{1},x_{2})$ and sample-derived potentials inferred through (1) the quantum potential on the score $\stheta(x,t)$, (2) the thermodynamic potential derived via integration of the score, and (3) the implicit potential defined via a flow $\rho_{\eta}$ learned on the samples. \textbf{(b.)} Plot of the position of the symmetric mode across the symplectic integration and the deviation in trajectory from the ground-truth.
    }
    \label{fig:HO}
\end{figure}

To evaluate the ability of the Score Hamiltonian framework to recover non-separable interactions, we utilize a Coupled Harmonic Oscillator (CHO). Unlike the single-particle Hydrogen atom, the CHO involves two particles whose motion is correlated through a coupling potential, providing a test for identifying system Hamiltonians from non-trivially coupled data.

We define the system Hamiltonian in atomic units ($m=1, \hbar=1$ so that the coefficient is reduced by a factor of 2 relative to the $\frac{\hbar^{2}}{2m}=1$ convention in the main text) as: 
\begin{equation}\label{eq:coupled_HO}
\widehat{H} = -\frac{1}{2}\nabla^2 + \frac{1}{2}k(x_1^2 + x_2^2) + \frac{1}{2}\lambda(x_1 - x_2)^2
\end{equation}
where $k$ is the local spring constant and $\lambda$ represents the coupling strength. For our experiments, we set $k=1.0$ and $\lambda=5.0$, resulting in a highly correlated ground state density. The analytical normal mode frequencies are $\omega_{sym} = \sqrt{k} = 1.0$ and $\omega_{anti} = \sqrt{k + 2\lambda} \approx 3.317$. 

We sample 40,000 observations directly from the ground state density $|\psi_0|^2$. Samples are generated in the normal mode basis and subsequently rotated back to physical coordinates $(x_1, x_2)$. 

\textbf{Model Architecture and Training.} For the Score Hamiltonian, we use a conservative score network $S_\theta$ consisting of a 3-layer MLP with 128 hidden units and Softplus activations. Softplus is utilized to ensure smooth second derivatives, which are essential for computing the Bohm potential $Q$. The model is trained via denoising score-matching for 3,000 epochs. For the Boltzmann Generator baseline we use a RealNVP normalizing flow \cite{dinh2017density} implemented via the \texttt{nflows} library, consisting of 4 affine coupling layers and random permutations. The model is trained by minimizing the negative log-likelihood (NLL) of the data. 

A key contribution of this work is the comparison of three distinct interpretations of the learned score $S_\theta$.
\begin{itemize}
    \item The Quantum Interpretation (that of the Bohm or Score Hamiltonian, Definition~\ref{def:ScoreHamiltonian}) is that the potential is recovered via the Bohm-score identity: $\hat{V}_{\theta}(x) \propto \frac{1}{4}\nabla \cdot \stheta + \frac{1}{8}|\stheta|^2$.
    \item The thermodynamic interpretation (diffusion) by comparison is that the potential is recovered by performing path integration of the same score field: $V = -\int (S(x) - S(0)) \cdot dx$.
    \item For the Boltzmann Generator \cite{No2019}, the potential is defined as the negative log-density of a normalizing flow \cite{pmlr-v37-rezende15} density model: $V = -\log \rho_{\eta}$. This is, formally, an identical treatment to score-integration but supposing a fit of the density rather than the gradient of the log density.
\end{itemize}
In Table~\ref{tab:cho_metrics_supp} we compute the mean-absolute error of the potentials to the ground-truth $\hat{V}(x_{1},x_{2})$ and the root-mean squared error to the ground-truth to evaluate the learned potentials. To evaluate the physical accuracy of the learned Hamiltonians, we compute the Hessian of the recovered potentials at the equilibrium point to derive the learned normal mode frequencies. Furthermore, we perform a symplectic (Velocity Verlet) integration to simulate real-time dynamics under each potential, measuring the Euclidean divergence from the ground truth trajectory over $T=50$ time steps. In the thermodynamical case of the score, this amounts to the use of the score as a force $\mathbf{F} \propto \stheta(x)$, as in \cite{plainer2026consistent}.

\begin{table*}[ht]
\centering
\caption{Potential recovery and spectral metrics for the learned Hamiltonians on the coupled harmonic oscillator.}
\label{tab:cho_metrics_supp}
\vskip 0.15in
\begin{small}
\begin{tabular}{lcccccc}
\toprule
 & \multicolumn{2}{c}{Potential error} & \multicolumn{2}{c}{Symmetric mode ($\omega_1$)} & \multicolumn{2}{c}{Anti-symmetric mode ($\omega_2$)} \\
\cmidrule(lr){2-3} \cmidrule(lr){4-5} \cmidrule(lr){6-7}
Method & MAE & RMSE & Freq. & Rel. err. (\%) & Freq. & Rel. err. (\%) \\
\midrule
Ground truth           & 0.000 & 0.000 & 1.000 & 0.000 & 3.317 & 0.000 \\
Score Hamiltonian      & \textbf{5.341} & \textbf{10.434} & \textbf{1.049} & \textbf{4.855} & \textbf{3.497} & \textbf{5.451} \\
Thermodynamic integral & 19.534 & 29.876 & 1.432 & 43.150 & 2.623 & 20.903 \\
Boltzmann generator    & 40.846 & 61.017 & 0.736 & 26.377 & 1.364 & 58.886 \\
\bottomrule
\end{tabular}
\end{small}
\end{table*}

\begin{figure}[tbp]
    \centering
    \includegraphics[width=0.8\linewidth]{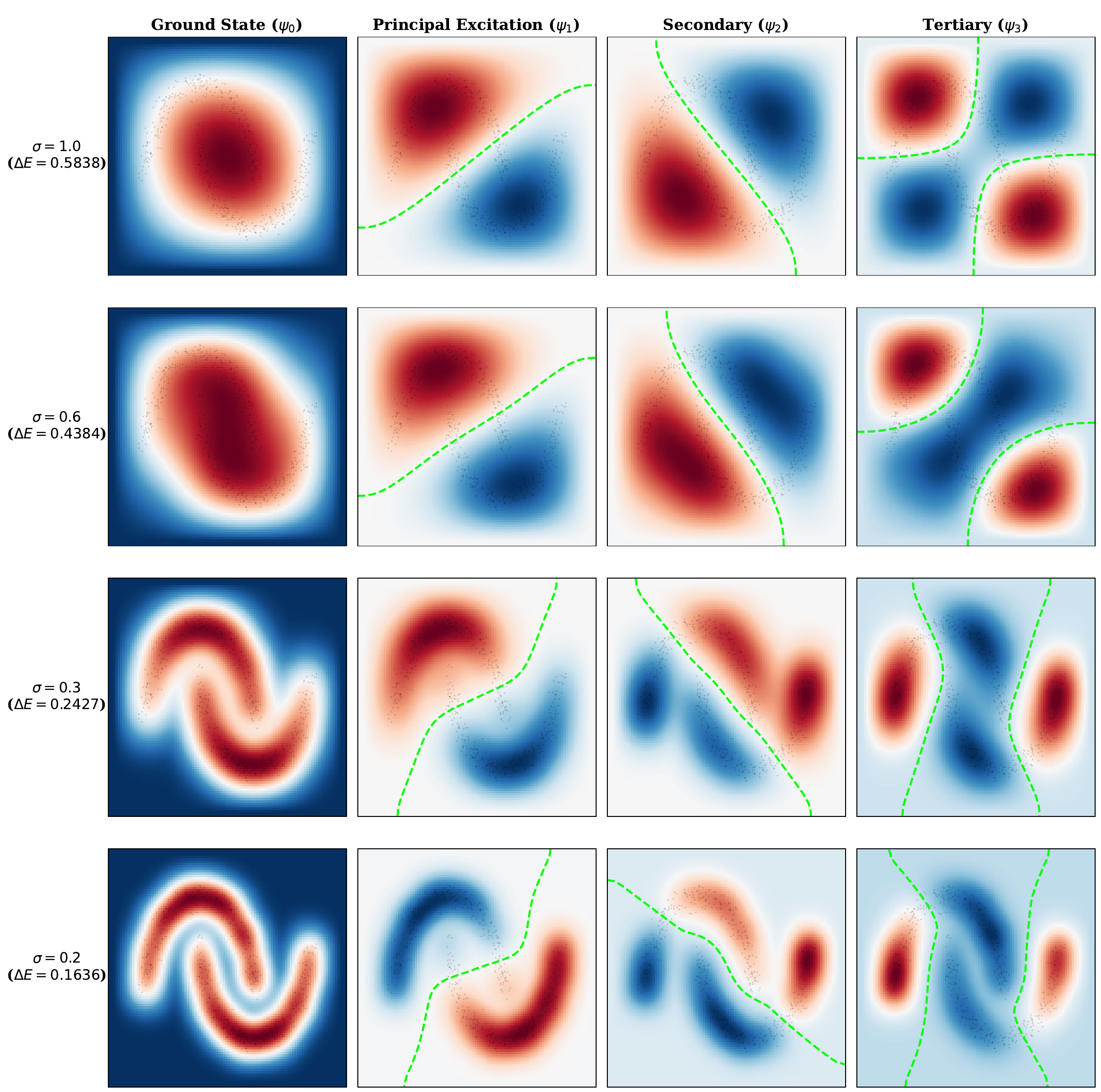}
    \caption{\textbf{Spectral Decomposition of the Two-Moons synthetic density.} Eigenfunctions of the score Hamiltonian $\widehat{H}_{\theta}(t)$ across noise scales (diffusion times) $\sigma \in \{1.0,0.6,0.3,0.2 \}$ computed via the Stationary Schr{\"o}dinger equation $\widehat{H}_{\theta}(t)\psi_{l} = E_{l}\psi_{l}$ for $l \in \{0,1,2,3\}$. The spectral gap $\Delta$ is reported for each noise-scale of the diffusion.
    }
    \label{fig:adiabatic}
\end{figure}

\section{Demonstration of the Theoretical Bounds on Synthetic Densities}\label{sec:demonstration}

We demonstrate the Score Hamiltonian bounds on score networks $\stheta(x,t)$ trained via denoising score matching in order to understand the generative process across the spectrum $\Delta(s)$ for diffusion times $s \in [0,T]$. For this analysis to admit a clear ground-truth, we sample from a simple and easy to analyze synthetic data density $\rho_{\mathrm{data}}(x)$; two simple and low-dimensional synthetic datasets of 2D Gaussian mixtures
\[
\rho_{\mathrm{data}}(x) =  \sum_{k=1}^K \pi_{k}\,\mathcal{N}(x; \mu_k, \sigma_k^2 \mathbf{I})
\]
These enable simple analysis of the dynamics of the eigenspectrum $\{ E_{k}(t) \}$ and thus the underlying spectral gap $\Delta$ across diffusion times while admitting numerically tractable total variation distance (TVD). In addition, computing the grid-based eigenspectrum of the score Hamiltonian via the Stationary Schr{\"o}dinger equation $\widehat{H}_{\theta}\psi_{k} = E_{k} \psi_{k}$ is highly computationally efficient in this setting, and we use sparse Lanczos eigensolver to extract the lowest $k=5$ eigenvalues of the discretized $3136 \times 3136$ Hamiltonian matrix.

\paragraph{Bimodal 2D GMM (Monotone Spectral Gap).} We first demonstrate on symmetric two-mode GMMs where the inter-mode distance $d$ grows, increasing the spatial bottleneck and thus decreasing the spectral gap $\Delta$. Let $x\in\mathbb{R}^2$. For a mode-separation parameter $d>0$, we let $\mu_1=(-d,0),\,\,\mu_2=(d,0),\,\,\Sigma = (0.5)^2 \bm{I}_2.$ and set the training distribution to be the equally weighted two-component Gaussian mixture
\[
\rho_{\text{data}}(x;d)=\tfrac12\,\mathcal N(x;\mu_1,\Sigma)+\tfrac12\,\mathcal N(x;\mu_2,\Sigma).
\]
Sampling is implemented by first drawing $c\sim\mathrm{Uniform}\{1,2\}$ and then sampling $x$ according to $x=\mu_c+\epsilon$ for $\epsilon\sim\mathcal N(0,\Sigma)$. While the first illustrative baseline training run uses $d=2.0$, the spectral sweep uses $n=30$ points spaced across $d\in[0.05,1.0]$.

\paragraph{Hierarchical GMM (Hierarchically Varying Spectral Gap).}

For scheduler evaluation, we define a harder hierarchical dataset where the spectral gap decreases in the middle of the diffusion time, owing to the hierarchical resolution of fine sub-clusters developing within a set of coarse clusters. This experiment thus tests the capacity of diffusion-based annealers to behave as spectral filters for a data density with hierarchical modal structure and evaluates whether the schedule $|\dot{t}|$ is adapted to and thus able to capture the dynamics of mode-formation and mode-splitting.

In particular, we use a hierarchical GMM consisting of 3 well-separated macro-modes, each containing a local ring of 3 micro-modes. This creates non-trivial and varying energy barriers across diffusion times. Define macro-centers
\[
M_1=(-3.2,-2.6),\quad M_2=(3.1,-2.8),\quad M_3=(0,3.3),
\]
with macro weights $(w_1,w_2,w_3)=(0.35,0.35,0.30)$.
Each macro-center has $K=3$ micro-modes on a ring of radius $r=0.95$:
\begin{align*}
&C_{m,k}=M_m+r(\cos\theta_k,\sin\theta_k)+\eta_{m,k},\\
&\theta_k=\frac{2\pi k}{3},\quad
\eta_{m,k}\sim\mathcal N(0,(0.08)^2I_2),
\end{align*}
with \texttt{jitter\_angle=False} and thus no random phase shift. Given micro-centers $\{C_{m,k}\}$, the final distribution is constructed via
\[
\rho_{\text{Hierarchical}}(x)=\sum_{m=1}^3\sum_{k=1}^3 \frac{w_m}{3}\,
\mathcal N\!\bigl(x;C_{m,k},(0.15)^2\, \mathbf{I}_2\bigr).
\]
This creates a two-stage topological bottleneck: the macro-modes segregate at high noise ($t \approx 0.4$), and the micro-modes segregate at low noise ($t \approx 0.02$).

\paragraph{Score Networks and Training.}
\emph{For the training phase} (as distinguished from the inference phase we consider where $\stheta$ and $\hat{H}_{\theta}$ become accessible), we use a VP-SDE \cite{song2021scorebased} with continuous linear schedule $\beta(t)=\beta_0+t(\beta_1-\beta_0)$ and $(\beta_0,\beta_1)=(0.1,20.0)$ and perturbation scale 
\begin{align*}
&\log \bar{\alpha}(t)= -\frac14(\beta_1-\beta_0)t^2-\frac12\beta_0 t,\\
&\sigma(t)=\sqrt{1-\exp\!\big(2\log\bar\alpha(t)\big)}.
\end{align*}
For each batch, with $t\sim\mathrm{Unif}[10^{-4},1]$, $z\sim\mathcal N(0,I)$, and $x_0\sim \rho_{\mathrm{data}}$ the empirical target scores of the forward $(\rho_{t})$ interpolation density are generated from samples $x_{0}$ of our data via
\begin{align*}
&x_t=\sqrt{1-\sigma(t)^2}\,x_0+\sigma(t)z,\\
&S^\star(x_t,t)=-\frac{z}{\sigma(t)}.
\end{align*}
Both networks are trained with weighted denoising score matching \cite{denoising_sm, ho2020denoising}
\[
\mathcal L(\theta)=\mathbb E\!\left[\sigma(t)^2\,\|\stheta(x_t,t)-S^\star(x_t,t)\|_2^2\right].
\]

\paragraph{(1) Baseline time-conditioned MLP (Bimodal GMM).} For this dataset, the score network input is $(x_1,x_2,t)\in\mathbb R^3$.  
Architecture: $3 \,\,({\text{Linear}}) \times 128 \,\, {(\text{SiLU})} \times
128 \,\, {(\text{SiLU})} \times 
128 \,\, {(\text{SiLU})} \times
2.$ Thus $s_\theta:\mathbb R^2\times(0,1]\to\mathbb R^2$ is a direct vector regressor. Training uses Adam ($\mathrm{lr}=10^{-3}$), batch size $256$, typically $900$--$1500$ epochs (single random minibatch per epoch in implementation).

\paragraph{(2) Conservative potential network (Hierarchical GMM scheduler benchmark).}
To enforce conservativity, we parameterize a scalar potential $\Phi_\theta(x,t) = - \log \rho_{\theta}(x,t) + \mathrm{const}$ and define
\[
\stheta(x,t)=-\nabla_x \Phi_\theta(x,t).
\]
Time is embedded with fixed Fourier features using $n_{\text{freq}}=16$ log-spaced frequencies
$\omega_i\in[1,100]$:
\[
\gamma(t)=\big[\sin(\omega_i t),\cos(\omega_i t)\big]_{i=1}^{16}\in\mathbb R^{32}.
\]
The potential network takes $[x,\gamma(t)]\in\mathbb R^{34}$ and uses four hidden blocks:
\[
34 \times 192 \times 192 \times 192 \times 192\times 1,
\]
each hidden layer is composed as a Linear layer, LayerNorm, and a subsequent SiLU. The score is obtained by automatic differentiation with respect to $x$.

Training uses Adam (weight decay $10^{-5}$), cosine-annealed LR
($\eta_{\min}=0.02\,\eta_0$), gradient clipping ($\le 1$), and EMA over all parameters, $\theta_{\text{EMA}}\leftarrow 0.999\,\theta_{\text{EMA}}+0.001\,\theta.$ At evaluation, EMA weights are loaded.

\paragraph{Spectrum Calculation.} At a given diffusion time $t$, evaluating the Score Hamiltonian requires the Bohm potential (we use atomic units $\frac{1}{2}\nabla^{2}$) of $Q_\theta = \frac{1}{4}\nabla \cdot \stheta + \frac{1}{8}\|\stheta\|^2$. We evaluate $\stheta$ on a dense spatial grid ($52\times52$ or $60\times60$) and compute the divergence via finite differences. The Hamiltonian $H_\theta = -\frac{1}{2}\nabla^2 + V_\theta$ is discretized using a Kronecker-sum Laplacian, and the lowest eigenvalues $\{E_k\}_{k=0}^K$ are extracted using a sparse Lanczos eigensolver (\texttt{SciPy eigsh}). To avoid boundary artifacts at high noise, the nominal spectral gap $\Delta$ is taken as the median of $E_1 - E_0$ over a log-spaced small-time $t$ band.

\begin{figure*}[tbp]
    \centering
    \includegraphics[width=\linewidth]{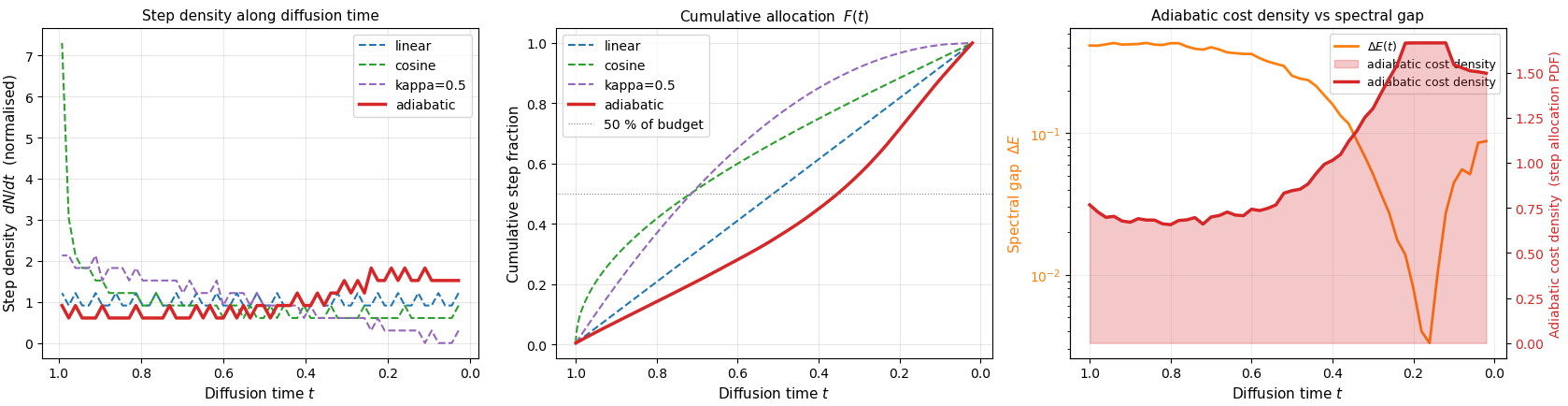}
    \caption{\textbf{Analysis of Adiabatic Schedule of Equation~\eqref{eq:score-based-adiabatic-schedule} for Hierarchical Density.} \textbf{(Left)} Plot of step-density of the schedules across diffusion time $t\in[0,1]$ from $t=0$ representing the harmonic oscillator and $t=1$ the hierarchical well potential. \textbf{(Middle)} Cumulative allocation of steps across diffusion time $t$. \textbf{(Right)} Adiabatic cost density (PDF) and spectral gap as a function of diffusion time $t \in [0,1]$.
    }
    \label{fig:AdiabaticAlloc}
\end{figure*}

\begin{figure*}[tbp]
    \centering
    \includegraphics[width=\linewidth]{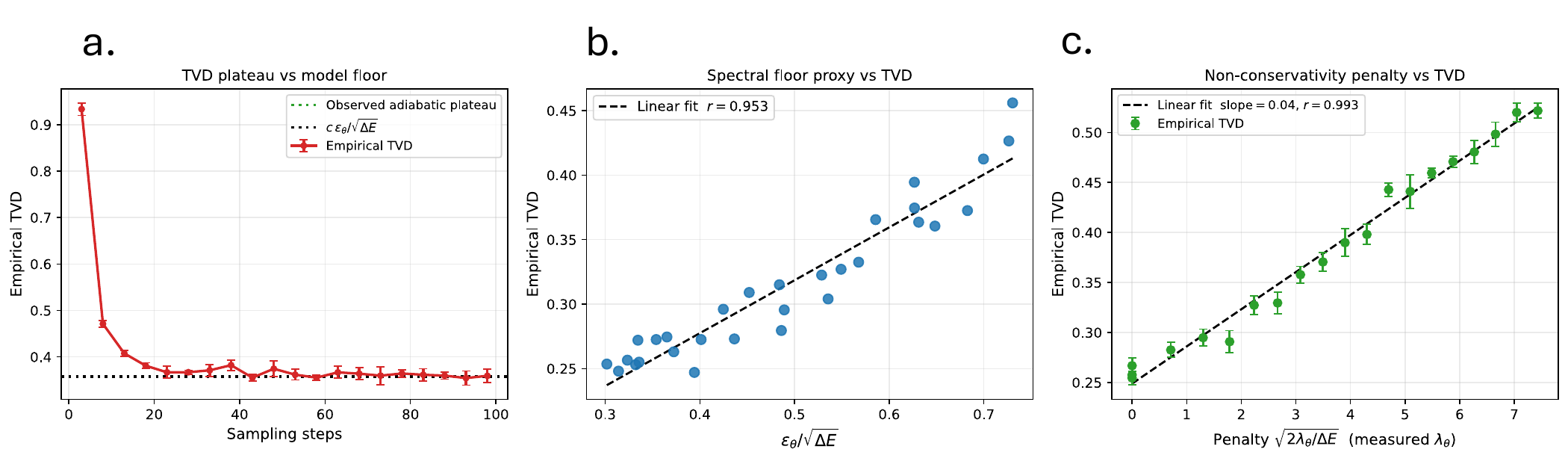}
    \caption{\textbf{Validation of the Adiabatic Bounds on the 2D GMM.} \textbf{(a)} Plot of $\mathrm{TVD}$ over the total sampling steps (larger step budgets correspond to slower schedules, i.e. smaller $|\dot{t}|$) to illustrate the plateau of Thm.~\ref{thrm:ad_diffusion} only varying $|\dot{t}|$. \textbf{(b)} Demonstration of the linear relationship between $\mathrm{TVD}$ and $\epsilon_{\theta,\mathrm{score}}/\sqrt{\Delta}$ predicted by Thm.~\ref{thrm:ad_diffusion}. \textbf{(c)} Illustration of the empirical $\mathrm{TVD}$ scaling linearly with injected non-conservativity as $\sqrt{2\lambda_{\theta}/\Delta}$. The added curl raises the ground-state energy $E_{0}=\lambda_{\theta}$ of the Score-Hamiltonian from zero, as predicted by Proposition~\ref{prop:non_conservative_score}.
    }
    \label{fig:gmm_bounds_validation}
\end{figure*}

\paragraph{Validation of Bounds on the 2D-GMM.} In this experiment, we measure generation quality via an empirical Total Variation Distance (TVD), calculated exactly by binning reverse-SDE or ODE samples and ground-truth validation samples onto a shared 2-D histogram grid. The model floor bound scales as $\epsilon_\theta / \sqrt{\Delta }$. We sweep the mode distance $d \in [0.05, 1.0]$ across 30 identical training runs. As $d$ increases, the minimum inter-mode spectral gap $\Delta $ exponentially shrinks while the network error $\epsilon_\theta$ remains stable. Empirical TVD strongly tracks $\epsilon_\theta /\sqrt{\Delta}$ ($r \approx 0.953$). Partial correlation analysis confirms $\Delta$ is the primary driver; the predicted term $\epsilon_\theta / \sqrt{\Delta}$ tightly bounds the floor error.

The full error decomposes into an intrinsic floor and a dynamical tracking error bounded by $\mathcal{O}(\dot{t}/\Delta)$. Sweeping the number of Euler-Maruyama sampling steps $N_\text{steps} \in [3, 98]$, we observe TVD monotonically decreasing until it reliably saturates at $N \gtrsim 20$. The observed adiabatic plateau aligns with a calibrated floor $c\,\epsilon_\theta/\sqrt{\Delta}$ -- thus, empirically finite-compute tracking errors indeed reach a ground-state, after which the intrinsic score-error to spectral gap ratio serves as a floor approximation limit.

Lastly, we remark that the theory penalizes non-conservative (curl) components by an additive factor of $\sqrt{2\lambda_\theta/\Delta}$ following Proposition~\ref{prop:non_conservative_score}. As MLPs are not inherently conservative, we explicitly test this scaling law in $\Delta$ and the non-conservativity by injecting a divergence-free vector field $v(x) = \kappa(-x_2, x_1)$ during sampling. Sweeping $\kappa$ from $0.0$ to $1.0$, the empirical TVD grows linearly with the penalty $\sqrt{2\lambda_\theta/\Delta}$. We directly compute $\lambda_\theta$ as the ground state energy $E_{0}$ of the Score Hamiltonian $\hat{H}_{\theta}$, which grows with added curl from zero as predicted by Proposition~\ref{prop:non_conservative_score} (see Figure~\ref{fig:gmm_bounds_validation}c).

\paragraph{Validation of Bounds on the Hierarchical-GMM.} To validate the practical utility of the adiabatic schedule, we evaluate reverse-ODE sample allocations on the hierarchical 3$\times$3 GMM using equal step budgets. We contrast standard grid strategies (\texttt{linear}, \texttt{cosine}, \texttt{power/kappa=0.5}) against an \texttt{adiabatic} grid. The adiabatic grid assigns discrete time steps proportional to the theoretically derived continuous cost density ${\|\partial_t S_t\|_{L^{2}(\rho_{t})} }/{ (\Delta)^{3/2}}$. The adiabatic scheduler reallocates steps to the small-$t$ regime, exactly increasing its budget as the spectral gap $\Delta(s)$ closes. Across extremely constrained budgets ($N \in [2, 14]$), the adiabatic schedule achieves significantly lower TVD than the other schedules, demonstrating that an adaptive schedule computed from the Score Hamiltonian gap $\Delta(t)$ yields more efficient and topologically-informed tracking than naive approaches which do not use the data-spectrum.


\end{document}